\documentclass[11pt]{article}
\usepackage[margin=0.75in]{geometry}
\pdfoutput=1
\usepackage[utf8]{inputenc}
\usepackage[english]{babel}
\usepackage[T1]{fontenc}
\usepackage{qcircuit}
\usepackage{amsmath,amssymb}
\usepackage{graphicx}           
\usepackage{relsize}
\usepackage{tcolorbox}
\usepackage{amsthm}
\usepackage{thm-restate}
\usepackage{braket}
\usepackage{csquotes}
\usepackage[range-phrase=--]{siunitx}
\usepackage{lscape}

\usepackage{soul}
\usepackage{mathtools}
\usepackage{multirow}
\usepackage{enumitem}

\usepackage{cite}
\usepackage[figurename=Fig.,font=footnotesize]{caption}
\usepackage{subcaption}
\theoremstyle{definition}
\usepackage{float}
\makeatletter
\newlength\min@xx

\makeatother

\usepackage{hyperref}
\hypersetup{
    colorlinks=true,
    linkcolor=blue,
    citecolor=magenta,
    filecolor=magenta,      
    urlcolor=blue,
    pdftitle={},
    pdfpagemode=FullScreen,
    }

\usepackage{lipsum}
\usepackage[capitalize]{cleveref}
\AddToHook{cmd/appendix/before}{
  \crefalias{section}{appendix}
  \crefalias{subsection}{appendix}
}
\usepackage{bbold}

\usepackage[compat=0.6]{yquant}
\yquantset{/yquant/operator/minimum width=2mm}
\yquantset{/yquant/operators/every text/.style={
    inner timesep=-1mm,
}}
\yquantset{/yquant/operators/every rectangular box/.style={
    shape=yquant-rectangle,
    draw,
    inner timesep=0.8mm,
    inner spacesep=0.5mm,
    outer timesep=1mm,
    time radius=1mm,
    /yquant/default fill,
}}
\yquantdefinebox{dots}[
    anchor=south,
    inner timesep=-5pt,
    inner spacesep=-5pt,
    time radius=-5pt,
]{$\vdots$}

\usetikzlibrary{quotes}
\useyquantlanguage{groups}

\newtheorem{theorem}{Theorem}
\newtheorem{lemma}{Lemma}
\newtheorem{corollary}{Corollary}
\newtheorem{definition}{Definition}

\newtheorem{proposition}{Proposition}
\newtheorem{example}{Example} 

\newtheorem*{informal_1}{(Informal version of) Theorem 9}
\newtheorem*{informal_2}{(Informal version of) Theorem 10}

\newcommand{\codepar}[1]{\ensuremath{[\![#1]\!]}}

\newcommand{\g}[1]{\mathsf{#1}}
\newcommand{\logg}[1]{\overline{\mathsf{#1}}}
\newcommand{\up}[1]{\mathrm{#1}}

\usepackage{amsbsy}

\newcommand{\bal}{\begin{equation}\begin{aligned}}
\newcommand{\eal}{\end{aligned}\end{equation}}

\crefname{statement}{Statement}{Statements}
{\setlength{\extrarowheight}{1.5pt}
\usepackage{tabulary}
\newcolumntype{K}[1]{>{\centering\arraybackslash}p{#1}}

\title{Construction of the full logical Clifford group for high-rate quantum Reed--Muller codes using only transversal and fold-transversal gates}

\author{Theerapat Tansuwannont$^{1,2,}$\footnote{\href{mailto:t.tansuwannont@phys.s.u-tokyo.ac.jp}{t.tansuwannont@phys.s.u-tokyo.ac.jp}} , \;Tim Chan$^{3,2,}$\footnote{\href{mailto:timothy.chan@materials.ox.ac.uk}{timothy.chan@materials.ox.ac.uk}} , \;and Ryuji Takagi$^{4,}$\footnote{\href{mailto:ryujitakagi@g.ecc.u-tokyo.ac.jp}{ryujitakagi@g.ecc.u-tokyo.ac.jp}}}
\date{\small
$^1$Department of Physics, Graduate School of Science, The University of Tokyo, 7-3-1 Hongo, Bunkyo-ku, Tokyo, 113-0033, Japan\\
$^2$Center for Quantum Information and Quantum Biology, The University of Osaka, Toyonaka, Osaka 560-0043, Japan \\
$^3$Department of Materials, University of Oxford, Parks Road, Oxford OX1 3PH, United Kingdom \\
$^4$Department of Basic Science, The University of Tokyo, 3-8-1 Komaba, Meguro-ku, Tokyo, 153-0041, Japan \\
}

\begin{document}

\maketitle

\begin{abstract}
To build large-scale quantum computers while minimizing resource requirements, one may want to use high-rate quantum error-correcting codes that can efficiently encode information. However, realizing an addressable gate---a logical gate on a subset of logical qubits within a high-rate code---in a fault-tolerant manner can be challenging and may require ancilla qubits. Transversal and fold-transversal gates could provide a means to fault-tolerantly implement logical gates using a constant-depth circuit without ancilla qubits, but available gates of these types could be limited depending on the code and might not be addressable. In this work, we study a family of \codepar{n=2^m,k={m \choose m/2}\approx n/\sqrt{\pi \log_2(n)/2},d=2^{m/2}=\sqrt{n}} self-dual quantum Reed--Muller codes, where $m$ is a positive even number. 
For any code in this family, we construct a generating set of the full logical Clifford group comprising only transversal and fold-transversal gates, thus enabling the implementation of any addressable Clifford gate. 
To our knowledge, this is the first known construction of the full logical Clifford group using only transversal and fold-transversal gates without requiring ancilla qubits for a family of codes in which $k$ grows near-linearly in $n$.
\end{abstract}

\section{Introduction} \label{sec:intro}

\subsection{Background}

In quantum computing, undesirable interactions between quantum systems and the environment can cause errors that lead to unreliable computation results. To protect a quantum state against noise, quantum error correcting codes (QECCs) can be used to encode some number of logical qubits into a larger number of physical qubits. In a setting where any physical
component---including physical gates, qubit preparation, and qubit measurement---can
be faulty, fault-tolerant error correction (FTEC) and fault-tolerant quantum computation (FTQC) schemes are required to ensure that errors due to circuit faults cannot propagate too severely \cite{Shor96,AB08,Kitaev97,KLZ96,Preskill98,TB05,ND05,AL06,AGP06}. Such fault-tolerant schemes generally require additional space and time overhead (e.g., ancilla qubits and quantum gates) \cite{Steane03,PR12,CJL17,TYC17}; thus, schemes that require small overhead are favorable for practical implementation.

QECCs with high encoding rates have gained a lot of recent attention since they use fewer physical qubits per logical qubit compared to well-known codes such as surface codes \cite{Kitaev97,BK98} or color codes \cite{BM06} with the same code distance, making them good candidates for building large-scale quantum computers. 
However, manipulating the logical information encoded in a block of high-rate codes can be challenging, especially if one wants to implement an addressable gate, i.e., a logical gate that targets a specific subset of logical qubits. One general technique to implement any addressable Clifford gate on any stabilizer code is gate teleportation \cite{ZLC00}. This requires two ancilla blocks prepared in a specific entangled state depending on the logical gate to be applied; one of the ancilla blocks is then jointly measured with the data block in the logical Bell basis. For a Calderbank--Shor--Steane (CSS) code, the gate teleportation scheme can be reformulated in terms of logical Pauli measurements \cite{BZHJL15}, and schemes to prepare required ancilla states have been proposed \cite{ZLB18,ZLBK20}. Still, these teleportation-based techniques are not ideal for practical implementation since the required ancilla blocks can be large. Several FTQC techniques that utilize ancilla states with joint measurements have been tailored for specific families of codes to further minimize the required overhead. Examples of such techniques are lattice surgery \cite{HHFDM12} and homomorphic logical measurements \cite{HJY23}. 

Transversal gates \cite{Gottesman97} and fold-transversal gates \cite{Moussa16,BB24} can provide a means to implement logical gates with constant time overhead and without ancilla qubits.
Transversal gates are fault-tolerant since they do not spread errors. For fold-transversal gates, although fault tolerance is not guaranteed in general, numerical simulations of several specific code and gate constructions have demonstrated performance approaching that expected from fault-tolerant operations under various noise models \cite{GV24,CCLP24,CBZGM+25,SMB25}.
However, the set of available logical gates from transversal and fold-transversal gates can be limited and depend on the symmetries of the code. Recently, there has been an active search on families of codes that admit a rich set of logical gates from transversal and fold-transversal gates. Some of these results include families of codes that admit addressable and transversal non-Clifford gates \cite{ZSPCB25,HVWZ25a,HVWZ25b,Guemard25,GL25}, and constructions of fold-transversal Clifford gates for quantum low-density parity check (qLDPC) codes \cite{QWV23,KZ23,GV24,TB25,ES25,EPS24,GK25,MGFCS+25}, floquet codes \cite{MJ26}, and concatenated symplectic double codes \cite{BD25}. Ref. \cite{SKWRB25} provides computational tools to search for available transversal and fold-transversal gates for any stabilizer code. 

Universal quantum computation can be achieved by Clifford gates plus any gate not in the Clifford group \cite{NRS01}. One may want to perform universal quantum computation in the logical level using only transversal gates. Unfortunately, this possibility is ruled out by the Eastin--Knill theorem \cite{ZCC11,CCCZC08,EK09}. A common attempt to achieve universality is to find a code in which Clifford gates can be implemented with low overhead, then apply an additional technique such as magic state preparation and injection \cite{KLZ97,BK05,ITHF25,GSJ24} or code switching \cite{PR13,ADP14} to implement non-Clifford gates. However, for the aforementioned families of codes \cite{QWV23,KZ23,GV24,TB25,ES25,EPS24,GK25,MGFCS+25,MJ26,BD25} in which Clifford gates can be implemented by transversal and fold-transversal gates, either the number of logical qubits is fixed for all codes in the same family, or in the case of high-rate codes, additional techniques are required to achieve the full logical Clifford group. The only exception is the family of subsystem hypergraph product simplex (SHYPS) codes developed in \cite{MGFCS+25}, for which the full logical Clifford group can be generated using only transversal and fold-transversal gates while the number of logical qubits $k$ scales poly-logarithmically with the block length $n$.

Here we aim to find a family of codes that admits a rich set of logical Clifford gates implementable by transversal and fold-transversal gates and has encoding rates better than the SHYPS codes in \cite{MGFCS+25}. We turn our attention to quantum Reed--Muller codes \cite{ADP14}, which are high-rate codes that can be constructed from two classical Reed--Muller codes \cite{MS77} through CSS code construction \cite{CS96,Steane96}. For each code in the family, physical qubits can be laid out on a hypercube, and depending on the parameters of its corresponding classical Reed--Muller codes, various logical Clifford or non-Clifford operations can be implemented transversally \cite{BCHK26}. Two well-known quantum Reed--Muller codes are the \codepar{4,2,2} code \cite{VGW96}, and the \codepar{16,6,4} tesseract code \cite{DR20}. Punctured versions of some quantum Reed--Muller codes, such as the \codepar{15,1,3} code, are known for their utility in magic state preparation \cite{DK25} and code switching \cite{PR13,ADP14}. Fold-transversal gates of some quantum Reed--Muller codes have also been studied recently \cite{GR24}. Thanks to their simple geometrical interpretations, implementations of certain quantum Reed--Muller codes have been demonstrated on ion trap \cite{RACCD+24,DBBHK25} and neutral atom platforms \cite{BGLEB25}.

\subsection{Our contributions}

In this work, we focus on a family of self-dual quantum Reed--Muller codes with code parameters \codepar{n=2^m,k={m \choose m/2}\approx n/\sqrt{\pi \log_2(n)/2},d=2^{m/2}=\sqrt{n}}, where $m$ is a positive even number, which can be constructed from classical Reed--Muller codes $\mathrm{RM}(m/2-1,m)$ through CSS construction. 
Here we use an approach different from \cite{GR24} to construct fold-transversal gates for the quantum Reed--Muller codes (see discussions in \cref{sec:discussion} for a comparison).
Our contributions are as follows.
\begin{enumerate}
    \item For each $m$, we construct a family of fold-transversal gates of the quantum code from certain automorphisms of the corresponding classical Reed--Muller code, and evaluate
    the logical action of each gate on our choice of logical $\g{X}$ and $\g{Z}$ operators of the code.
    Using these fold-transversal gates, we prove that all addressable phase gates ($\logg{S}$) of the quantum code as well as addressable controlled-$\g{Z}$ gates ($\logg{C_{00}Z}$) on some pairs of logical qubits can be constructed. We then provide explicit constructions of addressable Hadamard gates ($\logg{H}$), addressable swap gates ($\logg{SW}$), and addressable controlled-$\g{Z}$ gates ($\logg{C_{00}Z}$) on any logical qubit (or any pair of logical qubits) from a sequence of addressable gates from the previous step and a transversal Hadamard gate $\g{H}^{\otimes n}$. With these addressable gates, the full logical Clifford group $\overline{C}_k$ can be generated.
    \item We discuss fundamental limitations on circuit depth for realizing an arbitrary logical Clifford gate, and prove that for any code family that admits the full logical Clifford group from transversal and fold-transversal gates, it is impossible to realize all logical Clifford gates in constant time if $k=\omega(\sqrt{n\log n})$, which is true for the quantum Reed--Muller code family of our interest.
    \item To visualize our construction and assist future studies of quantum Reed--Muller codes, we provide an open-source Python package that can explicitly construct any addressable $\logg{H}$, $\logg{S}$, $\logg{C_{00}Z}$, or $\logg{SW}$ gate for any code in the studied family as a sequence of transversal and fold-transversal gates. 
    The package can be used to verify all theorems related to our construction of logical gates (\crefrange{thm:PnQ_preserve_stb}{thm:addr_SW_CZ} and \cref{cor:addr_S_CZ}).
    More generally,
    our package can construct and compute the logical action of any transversal or fold-transversal gate introduced in this paper.
    Our Python package is available at \cite{Chan2026}.
\end{enumerate}
Our main results are summarized in the following theorems.

\begin{informal_1} \label{thm:full_Clifford_sim}
Let $m$ be a positive even number and let $\mathrm{QRM}(m)$ be the \codepar{n=2^m,k={m \choose m/2},d=2^{m/2}} self-dual quantum Reed--Muller code whose stabilizers are constructed from a generator matrix of the classical Reed--Muller code $\mathrm{RM}(m/2-1,m)$ through CSS construction. For any $m$, there exists a set of transversal and fold-transversal gates $\{\g{U}_{i}\}_i$ constructed from automorphisms of $\mathrm{RM}(m/2-1,m)$ such that the full logical Clifford group $\overline{C}_k$ of $\mathrm{QRM}(m)$ is $\overline{C}_k=\left\langle\g{H}^{\otimes n},\g{U}_{i}\right\rangle_i$.
\end{informal_1}

\begin{informal_2} \label{thm:depth fully addressable Clifford_sim}
For any \codepar{n,k,d} code that admits the full logical Clifford group $\overline{C}_k$ from transversal and fold-transversal gates comprising single- and two-qubit physical Clifford gates, there exists a logical Clifford gate $\logg{G} \in \overline{C}_k$ such that any implementation of $\logg{G}$ by transversal and fold-transversal gates has depth $\Omega(\frac{k^2}{n\log n})$.
\end{informal_2}

Our results imply that on any quantum Reed--Muller code in the studied family, any logical Clifford gate can be implemented by a sequence of transversal and fold-transversal gates without requiring ancilla qubits, although the depth might not be constant in the block length $n$. To limit the propagation and accumulation of errors, FTEC gadgets can be interleaved into the gate sequence, allowing the overall circuit behavior to better approximate fault-tolerant operations.
We note that although our proofs are based on one particular symplectic basis (which defines the logical $\g{X}$ and $\g{Z}$ operators of the code), both of the main theorems hold for any symplectic basis; see \cref{subsubsec:full_Clifford} for discussions.

Our results can be viewed as a generalization of the logical controlled-$\g{Z}$ gate of the \codepar{4,2,2} code, a result that was introduced in Gottesman's proposal of small experimental schemes to demonstrate quantum fault tolerance \cite{Gottesman16}.
In fact, the \codepar{4,2,2} code is the smallest code in the code family focused in this work. As this family allows implementation of any addressable Clifford gate by a sequence of transversal and fold-transversal gates without ancilla qubits, and as $k$ grows near-linearly in $n$, we believe that other codes in this family would also be good candidates for demonstrating fault-tolerant quantum computation with high-rate codes on near- or mid-term quantum devices. A code table showing parameters of the five smallest instances in the studied quantum Reed--Muller code family is provided in \cref{tab:code_table}.

\begin{table}[!tbp]
    \centering        
    \begin{tabular}{| c | c | c | c | c |}
         \hline
         Parameter & \# physical qubits & \# logical qubits & code distance & encoding rate \\
         $m$ & $n=2^m$ & $k= {m \choose m/2}$ & $d=2^{m/2}$ & $R = k/n$ \\
         \hline
         2 & 4 & 2 & 2 & 0.5 \\
         \hline
         4 & 16 & 6 & 4 & 0.375 \\
         \hline
         6 & 64 & 20 & 8 & 0.3125 \\
         \hline
         8 & 256 & 70 & 16 & 0.2734 \\
         \hline
         10 & 1024 & 252 & 32 & 0.2461 \\
         \hline
    \end{tabular}
    \caption{Code parameters of the five smallest instances of the quantum Reed--Muller codes $\mathrm{QRM}(m)$ focused in this work.}
        \label{tab:code_table}
    \end{table}
    
Around the same time the preprint version of this work appeared, Chakraborty and Gottesman independently established a no-go theorem on fold-transversal gates for any stabilizer code in Ref. \cite{CG26}; they introduced a notion of $k$-fold transversal gadgets and proved that to achieve the full logical Clifford group on $k$ logical qubits, at least $k$-fold transversal gadgets are required at the physical level. We point out that although fold-transversal gates considered in this work are similar to 2-fold transversal gadgets in Ref. \cite{CG26}, their results do not apply to our case since their work assumes that all $k$-fold transversal gadgets in a circuit must be defined with respect to the same qubit partition. In other words, their results establish connections between the constraints on error propagation and circuit depth, and the limitation on logical Clifford computation. Our construction can circumvent their no-go theorem with the cost of increased circuit depth and intermediate FTEC gadgets between fold-transversal gates.

This paper is organized as follows. In \cref{sec:CRM_QRM}, we review classical Reed--Muller codes, discuss some code properties and code automorphisms that will be used in our construction of fold-transversal gates, and review a construction of quantum Reed--Muller codes. In \cref{sec:logical_for_QRM}, we propose a construction of fold-transversal gates for quantum Reed--Muller codes from automorphisms of classical Reed--Muller codes, then construct addressable $\logg{H}$, $\logg{S}$, $\logg{C_{00}Z}$, and $\logg{SW}$ gates from a sequence of transversal and fold-transversal gates, which in turn generate the full logical Clifford group. In \cref{sec:limitations}, we discuss fundamental limitations on circuit depth for realizing an arbitrary logical Clifford gate. We then discuss and conclude our work in \cref{sec:discussion}.

\section{Classical and quantum Reed--Muller codes} \label{sec:CRM_QRM}

In this section, we provide some mathematical definitions, propositions, and theorems relevant to classical and quantum Reed--Muller codes. These notations and statements will be used in our main construction of logical gates for quantum Reed--Muller codes in \cref{sec:logical_for_QRM}.

\subsection{Definition of classical Reed--Muller codes} \label{subsec:CRM_def}

Given an integer $m \geq 0$, we let the number of physical bits be $n=2^m$. We label each physical bit with a number in base 2, where $x_{m}x_{m-1}\cdots x_{1}$ is equivalent to $x_{m}2^{m-1}+x_{m-1}2^{m-2}+\cdots+x_{1}2^{0}$ in base 10. We refer to the label of each bit as \emph{bit index}, which runs from $0$ to $2^m-1$ in base 10. Next, we define basis vectors $\mathbf{v}_i \in \mathbb{Z}_2^n$ where $i \in [m] = \{1,\dots,m\}$ as follows. $\mathbf{v}_i$ is a binary vector in which the bits labeled by $x_{m}\cdots x_{i+1} 1 x_{i-1} \cdots x_{1}$ have value one, and the bits labeled by $x_{m}\cdots x_{i+1} 0 x_{i-1} \cdots x_{1}$ have value zero, where $x_{j \neq i} \in \{0,1\}$. We refer to the subscript $i$ of such $\mathbf{v}_i$ as \emph{basis vector index}. We also define $\mathbf{1}$ to be the vector of all ones. Adding $\mathbf{1}$ to any vector will flip the locations of zeros and ones. For example, $\mathbf{v}_i+\mathbf{1}$ has ones at bits $x_{m}\cdots x_{i+1} 0 x_{i-1} \cdots x_{1}$ and has zeros at bits $x_{m}\cdots x_{i+1} 1 x_{i-1} \cdots x_{1}$ (where $x_{j \neq i} \in \{0,1\}$).

Let $\mathbf{w}_1,\mathbf{w}_2 \in \mathbb{Z}_2^n$ be two binary vectors. The wedge product $\mathbf{w}_1\wedge \mathbf{w}_2$ (also denoted by $\mathbf{w}_1\mathbf{w}_2$) is obtained by applying bitwise AND between $\mathbf{w}_1$ and $\mathbf{w}_2$. For example, let $\mathbf{v}_{i_1},\mathbf{v}_{i_2}$ be two basis vectors with $i_1 < i_2$. Then, $\mathbf{v}_{i_1}\wedge\mathbf{v}_{i_2}$ is the vector that has ones at bits $x_{m}\cdots x_{i_2+1} 1 x_{i_2-1} \cdots x_{i_1+1} 1 x_{i_1-1} \cdots x_{1}$ (where $x_{j} \in \{0,1\}, j \neq i_1,i_2$) and zeros at the other bits.

We can also define the wedge product of multiple basis vectors given a set of basis vector indices as follows. 
\begin{definition} \label{def:vector_from_set}
    Let $A = \{i_1,...,i_c\} \subseteq [m]$ be a set of basis vector indices. The vector $\mathbf{v}_A \in \mathbb{Z}_2^{2^m}$ is defined as
\begin{equation}
    \mathbf{v}_A = \bigwedge_{a \in A} \mathbf{v}_{a} = \mathbf{v}_{i_1} \wedge \mathbf{v}_{i_2} \wedge \cdots \wedge \mathbf{v}_{i_c}, \label{eq:vector_on_set}
\end{equation}
and we define $\mathbf{v}_\emptyset = \mathbf{1}$, where $\emptyset$ is the empty set.
\end{definition}

Let an $[n,k_\up{c},d_\up{c}]$ code denote a classical binary linear code that encodes $k_\up{c}$ logical bits into $n$ physical bits and has code distance $d_\up{c}$. A classical Reed--Muller code \cite{MS77} can be defined as follows.
\begin{definition} \label{def:CRM}
    Let $r,m$ be integers satisfying $0 \leq r \leq m$. A \emph{classical Reed--Muller code} of order $r$ and block length $n=2^m$, denoted by $\mathrm{RM}(r,m)$, is the $\left[2^m, \sum_{i=0}^{r} {m \choose i}, 2^{m-r}\right]$ code generated by all vectors $\mathbf{v}_A$ such that $0 \leq |A| \leq r$.
\end{definition}

Although the definition of classical Reed--Muller codes in this work is not conventional, it is equivalent to the standard definition such as the one in Ref. \cite{MS77}. In particular, let $\mathbf{x}=(x_1,\dots,x_m)\in\mathbb{Z}_2^m$. The vector $\mathbf{v}_i$ corresponds to the truth table of the Boolean function $f(\mathbf{x})=x_i$, $\mathbf{v}_A$ corresponds to the truth table of $f(\mathbf{x})=\Pi_{j\in A}(x_j)$, and $\mathbf{v}_\emptyset$ corresponds to the truth table of $f(\mathbf{x})=1$. Therefore, the classical Reed--Muller code $\mathrm{RM}(r,m)$ is the set of all vectors that correspond to the truth table of a Boolean function $f(\mathbf{x})$ which is a polynomial of degree at most $r$.

\begin{example}
    Let $m=4$. All possible $\mathbf{v}_{A}$ with $0 \leq |A|\leq 4$ are listed in \cref{tab:all_vA_m=4}.
    The classical Reed--Muller codes $\textrm{RM}(r,4)$ are generated by the following:
    \begin{equation}
        \begin{tabular}{| c | c |}
            \hline
            Code & Generators \\
            \hline
            $\textrm{RM}(0,4)$ & Row 1 \\
            \hline
            $\textrm{RM}(1,4)$ & Rows \numrange{1}{5} \\
            \hline
            $\textrm{RM}(2,4)$ & Rows \numrange{1}{11} \\
            \hline
            $\textrm{RM}(3,4)$ & Rows \numrange{1}{15} \\
            \hline
            $\textrm{RM}(4,4)$ & Rows \numrange{1}{16} \\
            \hline
        \end{tabular} \nonumber
    \end{equation}
    
    \begin{table}[htbp]
    \centering        
    \begin{tabular}{| c | c | c  c  c  c  c  c  c  c  c  c  c  c  c  c  c  c |}
         \hline
         Row & $A$ & \multicolumn{16}{|c|}{$\mathbf{v}_{A}$} \\
         \hline
         1 & $\emptyset$ & $1$ & $1$ & $1$ & $1$ & $1$ & $1$ & $1$ & $1$ & $1$ & $1$ & $1$ & $1$ & $1$ & $1$ & $1$ & $1$ \\
         \hline
         2 & $\{1\}$ & $0$ & $1$ & $0$ & $1$ & $0$ & $1$ & $0$ & $1$ & $0$ & $1$ & $0$ & $1$ & $0$ & $1$ & $0$ & $1$ \\
         \hline
         3 & $\{2\}$ & $0$ & $0$ & $1$ & $1$ & $0$ & $0$ & $1$ & $1$ & $0$ & $0$ & $1$ & $1$ & $0$ & $0$ & $1$ & $1$ \\
         \hline
         4 & $\{3\}$ & $0$ & $0$ & $0$ & $0$ & $1$ & $1$ & $1$ & $1$ & $0$ & $0$ & $0$ & $0$ & $1$ & $1$ & $1$ & $1$ \\
         \hline
         5 & $\{4\}$ & $0$ & $0$ & $0$ & $0$ & $0$ & $0$ & $0$ & $0$ & $1$ & $1$ & $1$ & $1$ & $1$ & $1$ & $1$ & $1$ \\
         \hline
         6 & $\{1,2\}$ & $0$ & $0$ & $0$ & $1$ & $0$ & $0$ & $0$ & $1$ & $0$ & $0$ & $0$ & $1$ & $0$ & $0$ & $0$ & $1$ \\
         \hline
         7 & $\{1,3\}$ & $0$ & $0$ & $0$ & $0$ & $0$ & $1$ & $0$ & $1$ & $0$ & $0$ & $0$ & $0$ & $0$ & $1$ & $0$ & $1$ \\
         \hline
         8 & $\{1,4\}$ & $0$ & $0$ & $0$ & $0$ & $0$ & $0$ & $0$ & $0$ & $0$ & $1$ & $0$ & $1$ & $0$ & $1$ & $0$ & $1$ \\
         \hline
         9 & $\{2,3\}$ & $0$ & $0$ & $0$ & $0$ & $0$ & $0$ & $1$ & $1$ & $0$ & $0$ & $0$ & $0$ & $0$ & $0$ & $1$ & $1$ \\
         \hline
         10 & $\{2,4\}$ & $0$ & $0$ & $0$ & $0$ & $0$ & $0$ & $0$ & $0$ & $0$ & $0$ & $1$ & $1$ & $0$ & $0$ & $1$ & $1$ \\
         \hline
         11 & $\{3,4\}$ & $0$ & $0$ & $0$ & $0$ & $0$ & $0$ & $0$ & $0$ & $0$ & $0$ & $0$ & $0$ & $1$ & $1$ & $1$ & $1$ \\
         \hline
         12 & $\{1,2,3\}$ & $0$ & $0$ & $0$ & $0$ & $0$ & $0$ & $0$ & $1$ & $0$ & $0$ & $0$ & $0$ & $0$ & $0$ & $0$ & $1$ \\
         \hline
         13 & $\{1,2,4\}$ & $0$ & $0$ & $0$ & $0$ & $0$ & $0$ & $0$ & $0$ & $0$ & $0$ & $0$ & $1$ & $0$ & $0$ & $0$ & $1$ \\
         \hline
         14 & $\{1,3,4\}$ & $0$ & $0$ & $0$ & $0$ & $0$ & $0$ & $0$ & $0$ & $0$ & $0$ & $0$ & $0$ & $0$ & $1$ & $0$ & $1$ \\
         \hline
         15 & $\{2,3,4\}$ & $0$ & $0$ & $0$ & $0$ & $0$ & $0$ & $0$ & $0$ & $0$ & $0$ & $0$ & $0$ & $0$ & $0$ & $1$ & $1$ \\
         \hline
         16 & $\{1,2,3,4\}$ & $0$ & $0$ & $0$ & $0$ & $0$ & $0$ & $0$ & $0$ & $0$ & $0$ & $0$ & $0$ & $0$ & $0$ & $0$ & $1$ \\
         \hline
         \parbox[t]{3.5mm}{\multirow{4}{*}{\rotatebox[origin=c]{90}{\textrm{bit indices}}}} & $x_4$ & \multicolumn{1}{|c|}{$0$} & \multicolumn{1}{|c|}{$0$} & \multicolumn{1}{|c|}{$0$} & \multicolumn{1}{|c|}{$0$} & \multicolumn{1}{|c|}{$0$} & \multicolumn{1}{|c|}{$0$} & \multicolumn{1}{|c|}{$0$} & \multicolumn{1}{|c|}{$0$} & \multicolumn{1}{|c|}{$1$} & \multicolumn{1}{|c|}{$1$} & \multicolumn{1}{|c|}{$1$} & \multicolumn{1}{|c|}{$1$} & \multicolumn{1}{|c|}{$1$} & \multicolumn{1}{|c|}{$1$} & \multicolumn{1}{|c|}{$1$} & \multicolumn{1}{|c|}{$1$} \\
         & $x_3$ & \multicolumn{1}{|c|}{$0$} & \multicolumn{1}{|c|}{$0$} & \multicolumn{1}{|c|}{$0$} & \multicolumn{1}{|c|}{$0$} & \multicolumn{1}{|c|}{$1$} & \multicolumn{1}{|c|}{$1$} & \multicolumn{1}{|c|}{$1$} & \multicolumn{1}{|c|}{$1$} & \multicolumn{1}{|c|}{$0$} & \multicolumn{1}{|c|}{$0$} & \multicolumn{1}{|c|}{$0$} & \multicolumn{1}{|c|}{$0$} & \multicolumn{1}{|c|}{$1$} & \multicolumn{1}{|c|}{$1$} & \multicolumn{1}{|c|}{$1$} & \multicolumn{1}{|c|}{$1$} \\
         & $x_2$ & \multicolumn{1}{|c|}{$0$} & \multicolumn{1}{|c|}{$0$} & \multicolumn{1}{|c|}{$1$} & \multicolumn{1}{|c|}{$1$} & \multicolumn{1}{|c|}{$0$} & \multicolumn{1}{|c|}{$0$} & \multicolumn{1}{|c|}{$1$} & \multicolumn{1}{|c|}{$1$} & \multicolumn{1}{|c|}{$0$} & \multicolumn{1}{|c|}{$0$} & \multicolumn{1}{|c|}{$1$} & \multicolumn{1}{|c|}{$1$} & \multicolumn{1}{|c|}{$0$} & \multicolumn{1}{|c|}{$0$} & \multicolumn{1}{|c|}{$1$} & \multicolumn{1}{|c|}{$1$} \\
         & $x_1$ & \multicolumn{1}{|c|}{$0$} & \multicolumn{1}{|c|}{$1$} & \multicolumn{1}{|c|}{$0$} & \multicolumn{1}{|c|}{$1$} & \multicolumn{1}{|c|}{$0$} & \multicolumn{1}{|c|}{$1$} & \multicolumn{1}{|c|}{$0$} & \multicolumn{1}{|c|}{$1$} & \multicolumn{1}{|c|}{$0$} & \multicolumn{1}{|c|}{$1$} & \multicolumn{1}{|c|}{$0$} & \multicolumn{1}{|c|}{$1$} & \multicolumn{1}{|c|}{$0$} & \multicolumn{1}{|c|}{$1$} & \multicolumn{1}{|c|}{$0$} & \multicolumn{1}{|c|}{$1$} \\
         \hline
    \end{tabular}
    \caption{The list of all binary vectors $\mathbf{v}_{A}$ defined over each subset $A \subseteq \{1,2,3,4\}$. Classical Reed--Muller codes $\textrm{RM}(r,4)$ where $r\in \{0,1,2,3,4\}$ can be generated by these vectors.}
        \label{tab:all_vA_m=4}
    \end{table}
\end{example}

The following properties of classical Reed--Muller codes \cite{MS77} will be useful in the construction of quantum Reed--Muller codes:
\begin{enumerate}
    \item Let $r_1,r_2 \leq m$ be non-negative integers. $\mathrm{RM}(r_1,m) \subseteq \mathrm{RM}(r_2,m)$ if and only if $r_1 \leq r_2$.
    \item The dual code of $\mathrm{RM}(r,m)$ is $\mathrm{RM}(r,m)^{\perp} = \mathrm{RM}(m-r-1,m)$.
\end{enumerate}

Let the Hamming weight $\mathrm{wt}(\mathbf{w})$ of a binary vector $\mathbf{w}$ be the number of bits of $\mathbf{w}$ with value one,
and define the dot product between two binary vectors $\mathbf{w} = (w_0,\dots,w_{n-1})$ and $\mathbf{z} = (z_0,\dots,z_{n-1})$ as follows:
\begin{equation}
    \mathbf{w}\cdot\mathbf{z} = (w_0z_0+\dots +w_{n-1}z_{n-1})\textrm{\;mod\;} 2.
\end{equation}
The following proposition can be obtained.
\begin{proposition} \label{prop:Hamming_weight}
    Let $A,B \subseteq [m]$ be sets of basis vector indices.
    \begin{enumerate}
        \item For any $A,B$ such that $A \cap B = \emptyset$, the Hamming weight of $\left(\bigwedge_{a \in A} \mathbf{v}_a \right) \wedge \left(\bigwedge_{b \in B} (\mathbf{v}_b+\mathbf{1}) \right)$ is $2^{m-|A|-|B|}$;
        in the special case where $B=\emptyset$, we have that $\mathrm{wt}(\mathbf{v}_A)=2^{m-|A|}$;
        \item For any $A,B$, $\mathbf{v}_A \cdot \mathbf{v}_B = \mathrm{wt}(\mathbf{v}_{A \cup B})\textrm{\;mod\;} 2$;
        \item For any $A,B$, $\mathbf{v}_A \cdot \mathbf{v}_B = 1$ if and only if $A\cup B = [m]$.
    \end{enumerate}
\end{proposition}
\begin{proof}
    (1) Suppose that $A \cap B = \emptyset$. $\left(\bigwedge_{a \in A} \mathbf{v}_a \right) \wedge \left(\bigwedge_{b \in B} (\mathbf{v}_b+\mathbf{1}) \right)$ is the vector with ones only at the bits whose label $x_m \cdots x_1$ satisfies the following conditions:
    \begin{enumerate}
        \item $x_a=1$ for all $a \in A$,
        \item $x_b=0$ for all $b \in B$,
        \item $x_j$ is unrestricted (can be either 0 or 1) if $j \notin A\cup B$.
    \end{enumerate}
    There are $2^{m-|A|-|B|}$ such bits since the number of unrestricted $x_j$ is $m-|A|-|B|$. Thus, the Hamming weight of $\left(\bigwedge_{a \in A} \mathbf{v}_a \right) \wedge \left(\bigwedge_{b \in B} (\mathbf{v}_b+\mathbf{1}) \right)$ is $2^{m-|A|-|B|}$.
    
    (2) By the definition of the dot product, $\mathbf{w}\cdot\mathbf{z}$ is equal to the number of bits that have value one on both $\mathbf{w}$ and $\mathbf{z}$ modulo 2, which is also equal to $\mathrm{wt}(\mathbf{w}\wedge\mathbf{z}) \mod 2$. Observe that for any $\mathbf{v}_i$, $\mathbf{v}_i \wedge \mathbf{v}_i = \mathbf{v}_i$. Therefore, $\mathbf{v}_A \wedge \mathbf{v}_B = \mathbf{v}_{A\cup B}$, which implies that $\mathbf{v}_A \cdot \mathbf{v}_B = \mathrm{wt}(\mathbf{v}_{A \cup B})\textrm{\;mod\;} 2$.

    (3) If $A\cup B = [m]$, then $\mathrm{wt}(\mathbf{v}_{A \cup B}) = 2^{m-m}=1$ by (1), which implies that $\mathbf{v}_A \cdot \mathbf{v}_B = 1$ by (2). In contrast, if $A\cup B \subsetneq [m]$, then the size of $A\cup B$ is smaller than $m$. By (1), $\mathrm{wt}(\mathbf{v}_{A \cup B})=2^{m-|A\cup B|}$ is some even number, which implies that $\mathbf{v}_A \cdot \mathbf{v}_B = 0$ by (2).
\end{proof}

\subsection{Automorphisms of classical Reed--Muller codes} \label{subsec:CRM_aut}

Let $\pi: \{0,\dots,2^m-1\} \rightarrow \{0,\dots,2^m-1\}$ be a permutation of $n=2^m$ bits (where $n$ is the block length) and let $M_\pi$ be the $n \times n$ permutation matrix of $\pi$. The elements of $M_\pi$ are as follows.
\begin{equation}
    (M_\pi)_{q,p} = \begin{cases} 
   1 & \text{if } q = \pi(p), \\
   0 & \text{otherwise}.
  \end{cases} \label{eq:perm_matrix}
\end{equation}
Let $\mathbf{w}$ be an $n$-bit binary column vector. Under the permutation $\pi$, $\mathbf{w}$ is mapped to $M_\pi\mathbf{w}$; the $\pi(p)$-th element of $M_\pi\mathbf{w}$ is the same as the $p$-th element of $\mathbf{w}$.

Given a classical binary code $\mathcal{C} = \left\langle \mathbf{w}_1, \dots, \mathbf{w}_c \right\rangle$, we say that $\pi$ is an \emph{automorphism} of $\mathcal{C}$ if $\mathcal{C}=\left\langle M_\pi\mathbf{w}_1, \dots, M_\pi\mathbf{w}_c\right\rangle$; i.e., the code space of $\mathcal{C}$ is invariant under $\pi$. 
In this work, we focus on the following automorphisms of classical Reed--Muller codes.
\begin{definition} \label{def:PQ_def}
    Consider a classical Reed--Muller code $\textrm{RM}(r,m)$ where $r,m$ are integers satisfying $0 \leq r \leq m$. Let $i,j \in [m]$, $i \neq j$. \emph{Transposition} $P(i,j)$ and \emph{transvection} $Q(i,j)$, which are automorphisms of $\textrm{RM}(r,m)$, are defined as follows:
    \begin{enumerate}
    \item $P(i,j)$ maps $\mathbf{v}_i$ to $\mathbf{v}_j$, maps $\mathbf{v}_j$ to $\mathbf{v}_i$, and maps $\mathbf{v}_l$ to itself if $l \neq i,j$. 
    \item $Q(i,j)$ maps $\mathbf{v}_i$ to $\mathbf{v}_i+\mathbf{v}_j$, and maps $\mathbf{v}_l$ to itself if $l \neq i$ (including $\mathbf{v}_j$). 
\end{enumerate}
\end{definition}

Recall that any codeword of a classical Reed--Muller code $\textrm{RM}(r,m)$ can be described by a Boolean function $f(\mathbf{x})$ of degree at most $r$, where $\mathbf{x}=(x_1,\dots,x_m)\in\mathbb{Z}_2^m$. It is also possible to describe $P(i,j)$ and $Q(i,j)$ as operations at the monomial level. In particular, $P(i,j)$ corresponds to the transformation,
\begin{equation}
    x_i \mapsto x_j, \quad x_j \mapsto x_i, \quad x_l \mapsto x_l \quad\textrm{for any}\quad l \neq i,j.
\end{equation}
Meanwhile, $Q(i,j)$ corresponds to the transformation,
\begin{equation}
    x_i \mapsto x_i+x_j, \quad x_l \mapsto x_l\quad\textrm{for any}\quad l \neq i.
\end{equation}

Note that $P(i,j) = P(j,i)$, but $Q(i,j) \neq Q(j,i)$. Both $P(i,j)$ and $Q(i,j)$ satisfy $P(i,j)^2 = Q(i,j)^2 = e$ for any $i,j$, where $e$ is the trivial permutation. Therefore, when $\pi$ is either $P(i,j)$ or $Q(i,j)$, its action can be physically realized by applying a physical swap gate to each pair of bits $\{p,\pi(p)\}$ satisfying $p < \pi(p)$, and applying nothing to any bit $q$ satisfying $q=\pi(q)$. A quantum analogue of bit permutation, where qubits are permuted according to a permutation $\pi$ that satisfies $\pi^2=e$, can be realized by a swap-type fold-transversal gate defined later in \cref{subsec:fold_from_aut}.\footnote{Note that in the quantum analogue, the unitary matrix that describes the operation of physical swap gates is a $2^n \times 2^n$ matrix and is not the same as the permutation matrix $M_{\pi}$ of a permutation $\pi$ (which is an $n \times n$ matrix).} 

The actions of $P(i,j)$ and $Q(i,j)$ can be explicitly described in terms of bit indices. In what follows, we assume that $i<j$ for simplicity, but a similar argument is also applicable in the case $i>j$. Observe the following facts.
\begin{enumerate}
    \item $\mathbf{v}_i$ has ones at bits $x_{m}\cdots x_{i+1} 1 x_{i-1} \cdots x_{1}$ and has zeros at bits $x_{m}\cdots x_{i+1} 0 x_{i-1} \cdots x_{1}$.
    \item $\mathbf{v}_j$ has ones at bits $x_{m}\cdots x_{j+1} 1 x_{j-1} \cdots x_{1}$ and has zeros at bits $x_{m}\cdots x_{j+1} 0 x_{j-1} \cdots x_{1}$.
    \item $\mathbf{v}_i+\mathbf{v}_j$ has ones at bits $x_{m}\cdots x_{j+1} 0 x_{j-1} \cdots x_{i+1} 1 x_{i-1} \cdots x_{1}$ and bits $x_{m}\cdots x_{j+1} 1 x_{j-1} \cdots x_{i+1} 0 x_{i-1} \cdots x_{1}$, and has zeros at bits $x_{m}\cdots x_{j+1} 0 x_{j-1} \cdots x_{i+1} 0 x_{i-1} \cdots x_{1}$ and bits $x_{m}\cdots x_{j+1} 1 x_{j-1} \cdots x_{i+1} 1 x_{i-1} \cdots x_{1}$.
\end{enumerate}
(When unspecified, the value of $x_l$ where $l \in [m]$ can be either 0 or 1.) The actions of $P(i,j)$ and $Q(i,j)$ are as follows.
\begin{enumerate}
    \item $P(i,j)$ swaps each pair of bits $x_{m}\cdots x_{j+1} 0 x_{j-1} \cdots x_{i+1} 1 x_{i-1} \cdots x_{1}$ and $x_{m}\cdots x_{j+1} 1 x_{j-1} \cdots x_{i+1} 0 x_{i-1} \cdots x_{1}$ with the same $x_l$ where $l \neq i,j$. It acts trivially on bits $x_{m}\cdots x_{j+1} 0 x_{j-1} \cdots x_{i+1} 0 x_{i-1} \cdots x_{1}$ and bits $x_{m}\cdots x_{j+1} 1 x_{j-1} \cdots x_{i+1} 1 x_{i-1} \cdots x_{1}$.
    \item $Q(i,j)$ swaps each pair of bits $x_{m}\cdots x_{j+1} 1 x_{j-1} \cdots x_{i+1} 0 x_{i-1} \cdots x_{1}$ and $x_{m}\cdots x_{j+1} 1 x_{j-1} \cdots x_{i+1} 1 x_{i-1} \cdots x_{1}$ with the same $x_l$ where $l \neq i,j$. It acts trivially on bits $x_{m}\cdots x_{j+1} 0 x_{j-1} \cdots x_{1}$.
\end{enumerate}
The permutation matrix of $P(i,j)$ or $Q(i,j)$ can be constructed from the descriptions above by converting bit indices from base 2 to base 10 and applying \cref{eq:perm_matrix}.

\begin{example}
    Let $m=4$. The operations of automorphisms $P(1,2)$, $P(3,4)$, $Q(1,2)$, and $Q(3,4)$ are as follows.
    \begin{equation}
        \begin{tabular}{ c c c c c c }
             $P(1,2)$ & \textrm{swaps} & $0001 \leftrightarrow 0010$, & $0101 \leftrightarrow 0110$, & $1001 \leftrightarrow 1010$, & $1101 \leftrightarrow 1110$,\\
             $P(3,4)$ &  \textrm{swaps} & $0100 \leftrightarrow 1000$, & $0101 \leftrightarrow 1001$, & $0110 \leftrightarrow 1010$, & $0111 \leftrightarrow 1011$,\\
             $Q(1,2)$ & \textrm{swaps} & $0010 \leftrightarrow 0011$, & $0110 \leftrightarrow 0111$, & $1010 \leftrightarrow 1011$, & $1110 \leftrightarrow 1111$,\\
             $Q(3,4)$ & \textrm{swaps} & $1000 \leftrightarrow 1100$, & $1001 \leftrightarrow 1101$, & $1010 \leftrightarrow 1110$, & $1011 \leftrightarrow 1111$.\\
        \end{tabular} \nonumber
    \end{equation}
\end{example}

The following proposition determines the action of a permutation $\pi$ on any vector $\mathbf{v}_A$ where $A \subseteq [m]$.
\begin{proposition} \label{prop:permute_vA}
    Let $A = \{i_1,...,i_c\} \subseteq [m]$ be a set of basis vector indices and let $\pi: \{0,\dots,2^m-1\} \rightarrow \{0,\dots,2^m-1\}$ be a permutation of $n=2^m$ bits. Then,
    \begin{equation}
        M_\pi\mathbf{v}_A = \bigwedge_{a \in A} M_\pi\mathbf{v}_{a} = M_\pi\mathbf{v}_{i_1} \wedge M_\pi\mathbf{v}_{i_2} \wedge \cdots \wedge M_\pi\mathbf{v}_{i_c}.
    \end{equation}
\end{proposition}
\begin{proof}
    The statement above follows from the facts that for any $p \in \{0,\dots,2^m-1\}$, (i) the $p$-th bit of $\mathbf{v}_A$ is the product of the $p$-th bits of $\mathbf{v}_{i_1},\dots,\mathbf{v}_{i_c}$, (ii) the $p$-th bit of $\mathbf{v}_A$ is the same as the $\pi(p)$-th bit of $M_{\pi}\mathbf{v}_A$, and (iii) the $p$-th bit of $\mathbf{v}_i$ is the same as the $\pi(p)$-th bit of $M_{\pi}\mathbf{v}_i$ for any $i \in [m]$. 
\end{proof}

It should be noted that the automorphism group of any classical Reed--Muller code $\textrm{RM}(r,m)$ with $1 \leq r \leq m-2$ is the affine general linear group \cite{MS77},
\begin{equation}
    \textrm{Aut}(\textrm{RM}(r,m)) = \textrm{AGL}(m,2).
\end{equation}
That is, suppose that a codeword of the code $\textrm{RM}(r,m)$ ($1 \leq r \leq m-2$) is described by a Boolean function $f(\mathbf{x})$ where $\mathbf{x}=(x_1,\dots,x_m)\in\mathbb{Z}_2^m$. Then, for any $m \times m$ invertible binary matrix $A \in \textrm{GL}(m,2)$ and any $m$-bit string $\mathbf{b} \in \mathbb{Z}^n_m$, an affine transformation of the form
\begin{equation}
    \mathbf{x} \mapsto A\mathbf{x}+\mathbf{b}
\end{equation}
is an automorphism of the code $\textrm{RM}(r,m)$.

One interesting observation is that the general linear group $\textrm{GL}(m,2)$, which is a subgroup of $\textrm{AGL}(m,2)$, can be generated by the automorphisms $Q(i,j)$ alone; i.e.,
\begin{equation}
   \textrm{GL}(m,2) = \langle Q(i,j) : i,j \in [m], i\neq j \rangle.
\end{equation}
This comes from the fact that the special linear group $\textrm{SL}(m,F)$ over a field $F$ is generated by elementary transvections \cite{Rotman95}; i.e.,
\begin{equation}
    \textrm{SL}(m,F) = \langle I+aE_{ij}: a\in F, i,j \in [m], i\neq j\rangle,
\end{equation}
where $I$ is the $m \times m$ identity matrix, and $E_{ij}$ is the $m \times m$ matrix whose $(i,j)$-element is 1 and other elements are 0. When $F = \mathbb{Z}_2$, we have that $I+E_{ij}=Q(i,j)$ for any $i,j\in [m], i\neq j$ and $\textrm{SL}(m,2) = \textrm{GL}(m,2)$. 

Another interesting observation is that the group generated by $P(i,j)$ is a proper subgroup of $\textrm{GL}(m,2)$, and is isomorphic to the symmetric group $S_m$; 
\begin{equation}
S_m \cong \langle P(i,j) : i,j \in [m], i\neq j \rangle < \textrm{GL}(m,2).
\end{equation}
Note that $P(i,j)=Q(i,j)Q(j,i)Q(i,j)$ for any $i,j \in [m]$, $i \neq j$.

\subsection{Definition of quantum Reed--Muller codes} \label{subsec:QRM_def}

A \emph{Calderbank-Shor-Steane (CSS) code} \cite{CS96,Steane96} is a stabilizer code \cite{Gottesman97} in which there exists a choice of stabilizer generators which are purely $\g{X}$- or purely $\g{Z}$-type. Let \codepar{n,k,d} denote a stabilizer code that encodes $k$ logical qubits into $n$ physical qubits and has code distance $d$. A CSS code can be constructed from two classical binary linear codes by the following theorem.
\begin{theorem} [CSS construction] \cite{CS96,Steane96} \label{thm:CSS}
	Let $\mathcal{C}_x$ and $\mathcal{C}_z$ be $[n,k_x,d_x]$ and $[n,k_z,d_z]$ classical linear codes with parity check matrix $H_x$ and $H_z$, respectively. Suppose that $\mathcal{C}_x^\perp\subseteq \mathcal{C}_z$. Let $\mathrm{CSS}\left(\mathcal{C}_x^\perp,\mathcal{C}_z^\perp\right)$ be the CSS code in which $\g{X}$-type ($\g{Z}$-type) stabilizer generators correspond to the rows of $H_x$ ($H_z$), where bit 0 corresponds to the identity operator $\g{I}$ and bit 1 corresponds to the Pauli $\g{X}$ ($\g{Z}$) operator. Then $\mathrm{CSS}\left(\mathcal{C}_x^\perp,\mathcal{C}_z^\perp\right)$ is an \codepar{n,k,d} code with $k=k_x+k_z-n$ and $d\geq\min\{d_x,d_z\}$. Logical $\g{X}$ (logical $\g{Z}$) operators of $\mathrm{CSS}\left(\mathcal{C}_x^\perp,\mathcal{C}_z^\perp\right)$ are described by binary vectors which are coset representatives of $\mathcal{C}_z / \mathcal{C}_x^\perp$ ($\mathcal{C}_x / \mathcal{C}_z^\perp$).
\end{theorem}
\noindent
In the case that $\mathcal{C}_x=\mathcal{C}_z$, the $\g{X}$- and $\g{Z}$-type stabilizer generators of the CSS code obtained from \cref{thm:CSS} are of the same form. Such a code is called \emph{self-dual CSS code}.

A \emph{Quantum Reed--Muller code} \cite{ADP14} is a CSS code which is constructed from two classical Reed--Muller codes through the CSS construction in \cref{thm:CSS}. Suppose that $\mathcal{C}_x^\perp = \mathrm{RM}(r_x,m)$ and $\mathcal{C}_z^\perp = \mathrm{RM}(r_z,m)$, where $r_x,r_z \in \{0,\dots,m\}$. Since $\mathcal{C}_z = \mathrm{RM}(r_z,m)^\perp = \mathrm{RM}(m-r_z-1,m)$, the condition $\mathcal{C}_x^\perp\subseteq \mathcal{C}_z$ is satisfied if and only if $r_x \leq m-r_z-1$, or equivalently, $r_x+r_z+1 \leq m$. In other words, for a given $m$, any code $\mathrm{CSS}\left(\mathrm{RM}(r_x,m),\mathrm{RM}(r_z,m)\right)$ with parameters $r_x,r_z$ that satisfy $r_x+r_z+1 \leq m$ is a valid quantum Reed--Muller code.

In this work, we focus on the following family of self-dual quantum Reed--Muller codes.
\begin{definition} \label{def:QRM}
    Let $m$ be a positive even number. The quantum Reed--Muller code $\mathrm{CSS}\left(\mathrm{RM}(r_x,m),\mathrm{RM}(r_z,m)\right)$ with parameters $r_x=r_z = m/2-1$ is denoted by $\mathrm{QRM}(m)$.
\end{definition}

The code parameters of $\mathrm{QRM}(m)$ is \codepar{n=2^m,k = {m \choose m/2},d=2^{m/2}}. Two well-known codes in this family are the \codepar{4,2,2} code \cite{VGW96} and the \codepar{16,6,4} tesseract code \cite{DR20}, which are $\mathrm{QRM}(2)$ and $\mathrm{QRM}(4)$, respectively. The derivation of the code parameters would be discussed after we define $\mathrm{QRM}(m)$ in stabilizer formalism.

Let $\mathbf{w}=(w_0,\dots,w_{n-1}) \in \mathbb{Z}^n_2$ be an $n$-bit binary vector. We define $\g{X}(\mathbf{w})$ as the $n$-qubit Pauli-$\g{X}$ operator $\g{X}^{w_0}\otimes\cdots\otimes \g{X}^{w_{n-1}}$. The $n$-qubit Pauli-$\g{Z}$ operator $\g{Z}(\mathbf{w})$ is defined analogously. The weight of a Pauli operator is defined as the number of its non-trivial tensor factors, thus the weight of $\g{X}(\mathbf{w})$ or $\g{Z}(\mathbf{w})$ is equal to the Hamming weight of $\mathbf{w}$. 

Following the notations introduced in \cref{subsec:CRM_def}, we can describe the stabilizer generators of $\mathrm{QRM}(m)$ by sets of basis vector indices as follows:
\begin{definition} \label{def:stb_def}
    Let $m$ be a positive even number and let $A \subsetneq [m]$ be a set of basis vector indices of size $0 \leq |A| \leq \frac{m}{2}-1$. The $\g{X}$-type and $\g{Z}$-type Pauli operators $g_x(A)$ and $g_z(A)$ are defined as follows:
\begin{align}
    g_x(A) &= \g{X}\left(\mathbf{v}_A\right) = \g{X}\left(\bigwedge_{a \in A}\mathbf{v}_{a}\right), \label{eq:stb_def1} \\
    g_z(A) &= \g{Z}\left(\mathbf{v}_A\right) = \g{Z}\left(\bigwedge_{a \in A}\mathbf{v}_{a}\right). \label{eq:stb_def2}
\end{align}
We also define $g_x(\emptyset) = \g{X}^{\otimes n}$ and $g_z(\emptyset) = \g{Z}^{\otimes n}$.

\end{definition}

\begin{proposition} \label{prop:QRM_stb}
    For any positive even number $m$, the stabilizer group $\mathcal{S}_{\mathrm{QRM}(m)}$ of $\mathrm{QRM}(m)$ can be written as 
    \begin{equation}
        \mathcal{S}_{\mathrm{QRM}(m)} = \left\langle g_x(A), g_z(A)\right\rangle_{A \subsetneq [m],0 \leq |A| \leq \frac{m}{2}-1}.
    \end{equation} 
\end{proposition}
\begin{proof}
    Recall that $\mathrm{QRM}(m)$ is the CSS code in which both $X$-type and $Z$-type stabilizer generators correspond to the rows of the generator matrix of $\mathrm{RM}(m/2-1,m)$ (or equivalently, the parity check matrix of $\mathrm{RM}(m/2-1,m)^\perp=\mathrm{RM}(m/2,m)$). Since $\mathrm{RM}(m/2-1,m) = \left\langle \mathbf{v}_A \right\rangle_{A \subsetneq [m],0 \leq |A| \leq \frac{m}{2}-1}$ (by \cref{def:CRM}), we have that $\mathcal{S}_{\mathrm{QRM}(m)} = \left\langle g_x(A), g_z(A)\right\rangle_{A \subsetneq [m],0 \leq |A| \leq \frac{m}{2}-1}$.
\end{proof}

Suppose that $A_1$ and $A_2$ satisfy $A_i \subseteq [m]$ and $0 \leq |A_i| \leq \frac{m}{2}-1$. We can readily verify that any pair of $g_x(A_1)$ and $g_z(A_2)$
commute using the facts that (i) $\g{X}(\mathbf{v}_{A_1})$ and $\g{Z}(\mathbf{v}_{A_2})$ commute if and only if $\mathbf{v}_{A_1} \cdot \mathbf{v}_{A_2}=0$, and (ii) $\mathbf{v}_{A_1} \cdot \mathbf{v}_{A_2} = 1$ if and only if $A_1\cup A_2 = [m]$ (by \cref{prop:Hamming_weight}).

We can also define logical $\g{X}$ and logical $\g{Z}$ operators of $\mathrm{QRM}(m)$ by sets of basis vector indices as follows. 
\begin{definition} \label{def:log_def}
Let $m$ be a positive even number and let $B \subsetneq [m]$ be a set of basis vector indices of size $|B| = m/2$. Logical $\g{X}$ and logical $\g{Z}$ operators of $\mathrm{QRM}(m)$ defined over set $B$, denoted by $\logg{X}(B)$ and $\logg{Z}(B)$ respectively, are
\begin{align}
    \logg{X}(B) &= \g{X}\left(\mathbf{v}_B\right) = \g{X}\left(\bigwedge_{b \in B}\mathbf{v}_{b}\right), \label{eq:log_def1} \\
    \logg{Z}(B) &= \g{Z}\left(\mathbf{v}_{B^\up{c}}\right) = \g{Z}\left(\bigwedge_{b' \in B^\up{c}}\mathbf{v}_{b'}\right), \label{eq:log_def2}
\end{align}
where $B^\up{c} = [m]\setminus B$ is the complement of $B$. 
\end{definition}
We can also verify that $\logg{X}(B)$ and $\logg{Z}(B')$ anticommute if and only if $B = B'$ by using the facts that (i) $\g{X}(\mathbf{v}_{B})$ and $\g{Z}(\mathbf{v}_{(B')^\up{c}})$ anticommute if and only if $\mathbf{v}_{B} \cdot \mathbf{v}_{(B')^\up{c}}=1$, and (ii) $\mathbf{v}_{B} \cdot \mathbf{v}_{(B')^\up{c}} = 1$ if and only if $B\cup (B')^\up{c} = [m]$ (which is equivalent to $B=B'$ since $|B|=|B'|=m/2$). 

Following \cref{def:log_def}, we can use a set of basis vector indices $B$ of size $|B| = m/2$ as a \emph{logical qubit index} for a logical qubit defined by $\logg{X}(B)$ and $\logg{Z}(B)$. To assist our explanations, we also label logical qubits of any $\mathrm{QRM}(m)$ using the following convention:
\begin{definition} \label{def:canonical_indices}
For any $\mathrm{QRM}(m)$, the \emph{canonical logical qubit indices} $B_i$ are defined in
lexicographic order for $i \in \{1,\dots,k/2\}$ where $k = {m \choose m/2}$, i.e.,
\begin{align}
B_1&=\{1,2,...,m/2-1,m/2\}, \nonumber \\
B_2&=\{1,2,...,m/2-1,m/2+1\}, \nonumber \\
&\;\;\vdots \nonumber \\
B_{m/2+1}&=\{1,2,...,m/2-1,m\}, \nonumber \\
B_{m/2+2}&=\{1,2,...,m/2,m/2+1\}, \nonumber \\
&\;\;\vdots \nonumber \\
B_{m+1}&=\{1,2,...,m/2,m\}, \nonumber \\
&\;\;\vdots \nonumber \\
B_{k/2-2}&=\{1,m/2+1,\dots,m-2,m-1\}, \nonumber \\
B_{k/2-1}&=\{1,m/2+1,\dots,m-2,m\}, \nonumber \\
B_{k/2}&=\{1,m/2+2,\dots,m-1,m\}. \nonumber
\end{align}
For $i \in \{k/2+1,\dots,k\}$, we define 
$B_i = (B_{i-k/2})^\up{c}$, the complement of $B_i$.
The $i$-th logical qubit of $\mathrm{QRM}(m)$ is the logical qubit defined by $\logg{X}(i)= \logg{X}(B_i)$ and $\logg{Z}(i)= \logg{Z}(B_i)$.
\end{definition}
Note that there are many possible choices of anti-commuting Pauli operators (i.e., symplectic bases) that can be used to define logical Pauli operators of a quantum code. However, throughout this work we will stick with the choice of logical Pauli operators defined by \cref{def:log_def,def:canonical_indices}.

\begin{example} \label{ex:canonical_indices_m4}
Let $m=4$ and consider $\mathrm{QRM}(4)$. Logical qubits 1 to 6 of the code correspond to the canonical logical qubit indices $B_1=\{1,2\}$, $B_2=\{1,3\}$, $B_3=\{1,4\}$, $B_4=\{3,4\}$, $B_5=\{2,4\}$, and $B_6=\{2,3\}$, respectively. Stabilizer generators of $\mathrm{QRM}(4)$ can be constructed from the generators of $\mathrm{RM}(1,4)$. A choice of stabilizer generators corresponding to the generators $\{\mathbf{v}_A: 0\leq|A|\leq 1\}$ of $\mathrm{RM}(1,4)$ is described in \cref{tab:all_g_from_vA}. Logical Pauli operators of $\mathrm{QRM}(4)$ as defined in \cref{def:log_def,def:canonical_indices} are described in \cref{tab:logical_m=4}. Similar examples for $\mathrm{QRM}(6)$ can be found in \cref{sec:ex_m6}.
\end{example}

For any $\mathrm{QRM}(m)$ of block length $n=2^m$, the total number of stabilizer generators of $\mathrm{QRM}(m)$ is $s= 2\left({m \choose 0} +\dots + {m \choose m/2-1} \right)$, and the total number of logical qubits of $\mathrm{QRM}(m)$ is $k={m \choose m/2}$. These numbers satisfy the property $n=k+s$ of any \codepar{n,k,d} stabilizer code. The (quantum) code distance of $\mathrm{QRM}(m)$ is determined by the (classical) code distance of $\mathcal{C}_x=\mathcal{C}_z=\mathrm{RM}(m/2,m)$, which is $2^{m-m/2} = 2^{m/2}$. That is, $\mathrm{QRM}(m)$ is a \codepar{2^m,{m \choose m/2},2^{m/2}} code. When $m$ is large, we can apply Stirling's approximation and show that the code parameters of $\mathrm{QRM}(m)$ are \codepar{n=2^m,k \approx 2^m/\sqrt{\pi m/2} = n/\sqrt{\pi \log_2(n)/2},d=2^{m/2}=\sqrt{n}}. 

It is also possible to show that for any $\mathrm{QRM}(m)$, the weight of all stabilizer generators can be reduced to $2^{m/2+1}=2\sqrt{n}$.

\begin{theorem}[Minimum-weight stabilizer generators of quantum Reed--Muller codes]\label{thm:reduced_weight}
    Let $m$ be a positive even number. For any $\mathrm{QRM}(m)$, there exists a set of stabilizer generators for the stabilizer group of $\mathrm{QRM}(m)$ such that the weight of each stabilizer generator is $2^{m/2+1}=2\sqrt{n}$ (where $n=2^m$), which is minimal.
\end{theorem}
\begin{proof}
    We will show that for any classical Reed--Muller code $\textrm{RM}(m/2-1,m)$, there exists a set of generators $\{\mathbf{h}_A\}$ such that the Hamming weight of each generator is $2^{m/2+1}=2\sqrt{n}$. Recall that $\textrm{RM}(m/2-1,m)$ is generated by all vectors $\mathbf{v}_A$ such that $0\leq|A|\leq m/2-1$. We define a new vector $\mathbf{h}_A$ over a set of vector indices $A$ as follows.
    \begin{enumerate}
        \item If $|A|=m/2-1$, we let $\mathbf{h}_A = \mathbf{v}_A$. By \cref{prop:Hamming_weight}, $\mathbf{h}_A$ has weight $2^{m/2+1}$.
        
        \item If $|A|=m/2-1-c$ where $1 \leq c \leq m/2-2$, we can pick any set of $c$ vector indices $\{i_1,\dots,i_c\} \subsetneq [m] \setminus A$, and define $\mathbf{h}_A=\mathbf{v}_{A}\left(\mathbf{1}+\mathbf{v}_{i_1}\right)\cdots\left(\mathbf{1}+\mathbf{v}_{i_c}\right)$. By \cref{prop:Hamming_weight}, we find that $\mathbf{h}_A$ has weight $2^{m-|A|-c} = 2^{m/2+1}$. We also find that $\mathbf{h}_A= \sum_{L\subseteq \{i_1,\dots,i_c\}} \mathbf{v}_{A \cup L}$,
        which is a linear combination of $\mathbf{v}_{A}$ and some vectors $\mathbf{v}_{A'}$ such that $|A'|>m/2-1-c$. Conversely, $\mathbf{v}_{A}$ is a linear combination of $\mathbf{h}_{A}$ and some vectors $\mathbf{h}_{A'}$ such that $|A'|>m/2-1-c$.
        
        \item If $|A|=0$, we let $\mathbf{h}_\emptyset= \left(\mathbf{1}+\mathbf{v}_{i_1}\right)\cdots\left(\mathbf{1}+\mathbf{v}_{i_{m/2-1}}\right)=\sum_{L\subseteq \{i_1,\dots,i_{m/2-1}\}} \mathbf{v}_{L}$ for some vector indices $i_1,\dots,i_{m/2-1} \in [m]$. Similarly, we can show that $\mathbf{h}_\emptyset$ has weight $2^{m/2+1}$. We find that $\mathbf{h}_\emptyset$ is a linear combination of $\mathbf{v}_\emptyset=\mathbf{1}$ and some vectors $\mathbf{v}_{A'}$ such that $|A'|>0$. Conversely, $\mathbf{v}_\emptyset$ is a linear combination of $\mathbf{h}_\emptyset$ and some vectors $\mathbf{h}_{A'}$ such that $|A'|>0$.
    \end{enumerate}
    With the above definition of $\mathbf{h}_A$, we have that $\left\langle \mathbf{h}_A: 0\leq|A|\leq m/2-1\right\rangle = \left\langle \mathbf{v}_A: 0\leq|A|\leq m/2-1\right\rangle = \textrm{RM}(m/2-1,m)$. Let $h_x(A) = \g{X}(\mathbf{h}_A)$ and $h_z(A) = \g{Z}(\mathbf{h}_A)$. We find that the stabilizer group $\mathcal{S}_{\mathrm{QRM}(m)}$ of $\mathrm{QRM}(m)$ can be written as $\mathcal{S}_{\mathrm{QRM}(m)} = \left\langle h_x(A), h_z(A)\right\rangle_{A \subsetneq [m], 0 \leq |A| \leq \frac{m}{2}-1}$.
    The weight of each stabilizer generator $h_x(A)$ or $h_z(A)$ is equal to the Hamming weight of $\mathbf{h}_A$, which is $2^{m/2+1}=2\sqrt{n}$.
\end{proof}

\begin{table}[!htbp]
    \centering        
    \begin{tabular}{| c | c | c  c  c  c  c  c  c  c  c  c  c  c  c  c  c  c |}
         \hline
         \parbox[t]{3.5mm}{\multirow{6}{*}{\rotatebox[origin=c]{90}{\textrm{$\g{X}$-type}}}} & $A$ & \multicolumn{16}{|c|}{$g_x(A)$} \\
         \cline{2-18}
         & $\emptyset$ & $\g{X}$ & $\g{X}$ & $\g{X}$ & $\g{X}$ & $\g{X}$ & $\g{X}$ & $\g{X}$ & $\g{X}$ & $\g{X}$ & $\g{X}$ & $\g{X}$ & $\g{X}$ & $\g{X}$ & $\g{X}$ & $\g{X}$ & $\g{X}$ \\
         \cline{2-18}
         & $\{1\}$ & $\g{I}$ & $\g{X}$ & $\g{I}$ & $\g{X}$ & $\g{I}$ & $\g{X}$ & $\g{I}$ & $\g{X}$ & $\g{I}$ & $\g{X}$ & $\g{I}$ & $\g{X}$ & $\g{I}$ & $\g{X}$ & $\g{I}$ & $\g{X}$ \\
         \cline{2-18}
         & $\{2\}$ & $\g{I}$ & $\g{I}$ & $\g{X}$ & $\g{X}$ & $\g{I}$ & $\g{I}$ & $\g{X}$ & $\g{X}$ & $\g{I}$ & $\g{I}$ & $\g{X}$ & $\g{X}$ & $\g{I}$ & $\g{I}$ & $\g{X}$ & $\g{X}$ \\
         \cline{2-18}
         & $\{3\}$ & $\g{I}$ & $\g{I}$ & $\g{I}$ & $\g{I}$ & $\g{X}$ & $\g{X}$ & $\g{X}$ & $\g{X}$ & $\g{I}$ & $\g{I}$ & $\g{I}$ & $\g{I}$ & $\g{X}$ & $\g{X}$ & $\g{X}$ & $\g{X}$ \\
         \cline{2-18}
         & $\{4\}$ & $\g{I}$ & $\g{I}$ & $\g{I}$ & $\g{I}$ & $\g{I}$ & $\g{I}$ & $\g{I}$ & $\g{I}$ & $\g{X}$ & $\g{X}$ & $\g{X}$ & $\g{X}$ & $\g{X}$ & $\g{X}$ & $\g{X}$ & $\g{X}$ \\
         \hline
         \parbox[t]{3.5mm}{\multirow{6}{*}{\rotatebox[origin=c]{90}{\textrm{$\g{Z}$-type}}}} & $A$ & \multicolumn{16}{|c|}{$g_z(A)$} \\
         \cline{2-18}
         & $\emptyset$ & $\g{Z}$ & $\g{Z}$ & $\g{Z}$ & $\g{Z}$ & $\g{Z}$ & $\g{Z}$ & $\g{Z}$ & $\g{Z}$ & $\g{Z}$ & $\g{Z}$ & $\g{Z}$ & $\g{Z}$ & $\g{Z}$ & $\g{Z}$ & $\g{Z}$ & $\g{Z}$ \\
         \cline{2-18}
         & $\{1\}$ & $\g{I}$ & $\g{Z}$ & $\g{I}$ & $\g{Z}$ & $\g{I}$ & $\g{Z}$ & $\g{I}$ & $\g{Z}$ & $\g{I}$ & $\g{Z}$ & $\g{I}$ & $\g{Z}$ & $\g{I}$ & $\g{Z}$ & $\g{I}$ & $\g{Z}$ \\
         \cline{2-18}
         & $\{2\}$ & $\g{I}$ & $\g{I}$ & $\g{Z}$ & $\g{Z}$ & $\g{I}$ & $\g{I}$ & $\g{Z}$ & $\g{Z}$ & $\g{I}$ & $\g{I}$ & $\g{Z}$ & $\g{Z}$ & $\g{I}$ & $\g{I}$ & $\g{Z}$ & $\g{Z}$ \\
         \cline{2-18}
         & $\{3\}$ & $\g{I}$ & $\g{I}$ & $\g{I}$ & $\g{I}$ & $\g{Z}$ & $\g{Z}$ & $\g{Z}$ & $\g{Z}$ & $\g{I}$ & $\g{I}$ & $\g{I}$ & $\g{I}$ & $\g{Z}$ & $\g{Z}$ & $\g{Z}$ & $\g{Z}$ \\
         \cline{2-18}
         & $\{4\}$ & $\g{I}$ & $\g{I}$ & $\g{I}$ & $\g{I}$ & $\g{I}$ & $\g{I}$ & $\g{I}$ & $\g{I}$ & $\g{Z}$ & $\g{Z}$ & $\g{Z}$ & $\g{Z}$ & $\g{Z}$ & $\g{Z}$ & $\g{Z}$ & $\g{Z}$ \\
         \hline
         \parbox[t]{3.5mm}{\multirow{4}{*}{\rotatebox[origin=c]{90}{\textrm{bit indices}}}} & $x_4$ & \multicolumn{1}{|c|}{$0$} & \multicolumn{1}{|c|}{$0$} & \multicolumn{1}{|c|}{$0$} & \multicolumn{1}{|c|}{$0$} & \multicolumn{1}{|c|}{$0$} & \multicolumn{1}{|c|}{$0$} & \multicolumn{1}{|c|}{$0$} & \multicolumn{1}{|c|}{$0$} & \multicolumn{1}{|c|}{$1$} & \multicolumn{1}{|c|}{$1$} & \multicolumn{1}{|c|}{$1$} & \multicolumn{1}{|c|}{$1$} & \multicolumn{1}{|c|}{$1$} & \multicolumn{1}{|c|}{$1$} & \multicolumn{1}{|c|}{$1$} & \multicolumn{1}{|c|}{$1$} \\
         & $x_3$ & \multicolumn{1}{|c|}{$0$} & \multicolumn{1}{|c|}{$0$} & \multicolumn{1}{|c|}{$0$} & \multicolumn{1}{|c|}{$0$} & \multicolumn{1}{|c|}{$1$} & \multicolumn{1}{|c|}{$1$} & \multicolumn{1}{|c|}{$1$} & \multicolumn{1}{|c|}{$1$} & \multicolumn{1}{|c|}{$0$} & \multicolumn{1}{|c|}{$0$} & \multicolumn{1}{|c|}{$0$} & \multicolumn{1}{|c|}{$0$} & \multicolumn{1}{|c|}{$1$} & \multicolumn{1}{|c|}{$1$} & \multicolumn{1}{|c|}{$1$} & \multicolumn{1}{|c|}{$1$} \\
         & $x_2$ & \multicolumn{1}{|c|}{$0$} & \multicolumn{1}{|c|}{$0$} & \multicolumn{1}{|c|}{$1$} & \multicolumn{1}{|c|}{$1$} & \multicolumn{1}{|c|}{$0$} & \multicolumn{1}{|c|}{$0$} & \multicolumn{1}{|c|}{$1$} & \multicolumn{1}{|c|}{$1$} & \multicolumn{1}{|c|}{$0$} & \multicolumn{1}{|c|}{$0$} & \multicolumn{1}{|c|}{$1$} & \multicolumn{1}{|c|}{$1$} & \multicolumn{1}{|c|}{$0$} & \multicolumn{1}{|c|}{$0$} & \multicolumn{1}{|c|}{$1$} & \multicolumn{1}{|c|}{$1$} \\
         & $x_1$ & \multicolumn{1}{|c|}{$0$} & \multicolumn{1}{|c|}{$1$} & \multicolumn{1}{|c|}{$0$} & \multicolumn{1}{|c|}{$1$} & \multicolumn{1}{|c|}{$0$} & \multicolumn{1}{|c|}{$1$} & \multicolumn{1}{|c|}{$0$} & \multicolumn{1}{|c|}{$1$} & \multicolumn{1}{|c|}{$0$} & \multicolumn{1}{|c|}{$1$} & \multicolumn{1}{|c|}{$0$} & \multicolumn{1}{|c|}{$1$} & \multicolumn{1}{|c|}{$0$} & \multicolumn{1}{|c|}{$1$} & \multicolumn{1}{|c|}{$0$} & \multicolumn{1}{|c|}{$1$} \\
         \hline
    \end{tabular}
    \caption{A choice of stabilizer generators of the quantum Reed--Muller code $\mathrm{QRM}(4)$ according to \cref{prop:QRM_stb}, which correspond to the generators $\{\mathbf{v}_A: 0\leq|A|\leq 1\}$ of the classical Reed--Muller code $\mathrm{RM}(1,4)$.}
        \label{tab:all_g_from_vA}
\end{table}

\begin{table}[!htbp]
    \centering        
    \begin{tabular}{| c | c | c  c  c  c  c  c  c  c  c  c  c  c  c  c  c  c |}
         \hline
         \parbox[t]{3.5mm}{\multirow{7}{*}{\rotatebox[origin=c]{90}{\textrm{$\g{X}$-type}}}} & $B$ & \multicolumn{16}{|c|}{$\logg{X}(B)$} \\
         \cline{2-18}
         & $B_1=\{1,2\}$ & $\g{I}$ & $\g{I}$ & $\g{I}$ & $\g{X}$ & $\g{I}$ & $\g{I}$ & $\g{I}$ & $\g{X}$ & $\g{I}$ & $\g{I}$ & $\g{I}$ & $\g{X}$ & $\g{I}$ & $\g{I}$ & $\g{I}$ & $\g{X}$ \\
         \cline{2-18}
         & $B_2=\{1,3\}$ & $\g{I}$ & $\g{I}$ & $\g{I}$ & $\g{I}$ & $\g{I}$ & $\g{X}$ & $\g{I}$ & $\g{X}$ & $\g{I}$ & $\g{I}$ & $\g{I}$ & $\g{I}$ & $\g{I}$ & $\g{X}$ & $\g{I}$ & $\g{X}$ \\
         \cline{2-18}
         & $B_3=\{1,4\}$ & $\g{I}$ & $\g{I}$ & $\g{I}$ & $\g{I}$ & $\g{I}$ & $\g{I}$ & $\g{I}$ & $\g{I}$ & $\g{I}$ & $\g{X}$ & $\g{I}$ & $\g{X}$ & $\g{I}$ & $\g{X}$ & $\g{I}$ & $\g{X}$ \\
         \cline{2-18}
         & $B_4=\{3,4\}$ & $\g{I}$ & $\g{I}$ & $\g{I}$ & $\g{I}$ & $\g{I}$ & $\g{I}$ & $\g{I}$ & $\g{I}$ & $\g{I}$ & $\g{I}$ & $\g{I}$ & $\g{I}$ & $\g{X}$ & $\g{X}$ & $\g{X}$ & $\g{X}$ \\
         \cline{2-18}
         & $B_5=\{2,4\}$ & $\g{I}$ & $\g{I}$ & $\g{I}$ & $\g{I}$ & $\g{I}$ & $\g{I}$ & $\g{I}$ & $\g{I}$ & $\g{I}$ & $\g{I}$ & $\g{X}$ & $\g{X}$ & $\g{I}$ & $\g{I}$ & $\g{X}$ & $\g{X}$ \\
         \cline{2-18}
         & $B_6=\{2,3\}$ & $\g{I}$ & $\g{I}$ & $\g{I}$ & $\g{I}$ & $\g{I}$ & $\g{I}$ & $\g{X}$ & $\g{X}$ & $\g{I}$ & $\g{I}$ & $\g{I}$ & $\g{I}$ & $\g{I}$ & $\g{I}$ & $\g{X}$ & $\g{X}$ \\
         \hline
         \parbox[t]{3.5mm}{\multirow{7}{*}{\rotatebox[origin=c]{90}{\textrm{$\g{Z}$-type}}}} & $B$ & \multicolumn{16}{|c|}{$\logg{Z}(B)$} \\
         \cline{2-18}
         & $B_1=\{1,2\}$ & $\g{I}$ & $\g{I}$ & $\g{I}$ & $\g{I}$ & $\g{I}$ & $\g{I}$ & $\g{I}$ & $\g{I}$ & $\g{I}$ & $\g{I}$ & $\g{I}$ & $\g{I}$ & $\g{Z}$ & $\g{Z}$ & $\g{Z}$ & $\g{Z}$ \\
         \cline{2-18}
         & $B_2=\{1,3\}$ & $\g{I}$ & $\g{I}$ & $\g{I}$ & $\g{I}$ & $\g{I}$ & $\g{I}$ & $\g{I}$ & $\g{I}$ & $\g{I}$ & $\g{I}$ & $\g{Z}$ & $\g{Z}$ & $\g{I}$ & $\g{I}$ & $\g{Z}$ & $\g{Z}$ \\
         \cline{2-18}
         & $B_3=\{1,4\}$ & $\g{I}$ & $\g{I}$ & $\g{I}$ & $\g{I}$ & $\g{I}$ & $\g{I}$ & $\g{Z}$ & $\g{Z}$ & $\g{I}$ & $\g{I}$ & $\g{I}$ & $\g{I}$ & $\g{I}$ & $\g{I}$ & $\g{Z}$ & $\g{Z}$ \\
         \cline{2-18}
         & $B_4=\{3,4\}$ & $\g{I}$ & $\g{I}$ & $\g{I}$ & $\g{Z}$ & $\g{I}$ & $\g{I}$ & $\g{I}$ & $\g{Z}$ & $\g{I}$ & $\g{I}$ & $\g{I}$ & $\g{Z}$ & $\g{I}$ & $\g{I}$ & $\g{I}$ & $\g{Z}$ \\
         \cline{2-18}
         & $B_5=\{2,4\}$ & $\g{I}$ & $\g{I}$ & $\g{I}$ & $\g{I}$ & $\g{I}$ & $\g{Z}$ & $\g{I}$ & $\g{Z}$ & $\g{I}$ & $\g{I}$ & $\g{I}$ & $\g{I}$ & $\g{I}$ & $\g{Z}$ & $\g{I}$ & $\g{Z}$ \\
         \cline{2-18}
         & $B_6=\{2,3\}$ & $\g{I}$ & $\g{I}$ & $\g{I}$ & $\g{I}$ & $\g{I}$ & $\g{I}$ & $\g{I}$ & $\g{I}$ & $\g{I}$ & $\g{Z}$ & $\g{I}$ & $\g{Z}$ & $\g{I}$ & $\g{Z}$ & $\g{I}$ & $\g{Z}$ \\
         \hline
         \parbox[t]{3.5mm}{\multirow{4}{*}{\rotatebox[origin=c]{90}{\textrm{bit indices}}}} & $x_4$ & \multicolumn{1}{|c|}{$0$} & \multicolumn{1}{|c|}{$0$} & \multicolumn{1}{|c|}{$0$} & \multicolumn{1}{|c|}{$0$} & \multicolumn{1}{|c|}{$0$} & \multicolumn{1}{|c|}{$0$} & \multicolumn{1}{|c|}{$0$} & \multicolumn{1}{|c|}{$0$} & \multicolumn{1}{|c|}{$1$} & \multicolumn{1}{|c|}{$1$} & \multicolumn{1}{|c|}{$1$} & \multicolumn{1}{|c|}{$1$} & \multicolumn{1}{|c|}{$1$} & \multicolumn{1}{|c|}{$1$} & \multicolumn{1}{|c|}{$1$} & \multicolumn{1}{|c|}{$1$} \\
         & $x_3$ & \multicolumn{1}{|c|}{$0$} & \multicolumn{1}{|c|}{$0$} & \multicolumn{1}{|c|}{$0$} & \multicolumn{1}{|c|}{$0$} & \multicolumn{1}{|c|}{$1$} & \multicolumn{1}{|c|}{$1$} & \multicolumn{1}{|c|}{$1$} & \multicolumn{1}{|c|}{$1$} & \multicolumn{1}{|c|}{$0$} & \multicolumn{1}{|c|}{$0$} & \multicolumn{1}{|c|}{$0$} & \multicolumn{1}{|c|}{$0$} & \multicolumn{1}{|c|}{$1$} & \multicolumn{1}{|c|}{$1$} & \multicolumn{1}{|c|}{$1$} & \multicolumn{1}{|c|}{$1$} \\
         & $x_2$ & \multicolumn{1}{|c|}{$0$} & \multicolumn{1}{|c|}{$0$} & \multicolumn{1}{|c|}{$1$} & \multicolumn{1}{|c|}{$1$} & \multicolumn{1}{|c|}{$0$} & \multicolumn{1}{|c|}{$0$} & \multicolumn{1}{|c|}{$1$} & \multicolumn{1}{|c|}{$1$} & \multicolumn{1}{|c|}{$0$} & \multicolumn{1}{|c|}{$0$} & \multicolumn{1}{|c|}{$1$} & \multicolumn{1}{|c|}{$1$} & \multicolumn{1}{|c|}{$0$} & \multicolumn{1}{|c|}{$0$} & \multicolumn{1}{|c|}{$1$} & \multicolumn{1}{|c|}{$1$} \\
         & $x_1$ & \multicolumn{1}{|c|}{$0$} & \multicolumn{1}{|c|}{$1$} & \multicolumn{1}{|c|}{$0$} & \multicolumn{1}{|c|}{$1$} & \multicolumn{1}{|c|}{$0$} & \multicolumn{1}{|c|}{$1$} & \multicolumn{1}{|c|}{$0$} & \multicolumn{1}{|c|}{$1$} & \multicolumn{1}{|c|}{$0$} & \multicolumn{1}{|c|}{$1$} & \multicolumn{1}{|c|}{$0$} & \multicolumn{1}{|c|}{$1$} & \multicolumn{1}{|c|}{$0$} & \multicolumn{1}{|c|}{$1$} & \multicolumn{1}{|c|}{$0$} & \multicolumn{1}{|c|}{$1$} \\
         \hline
    \end{tabular}
    \caption{Logical Pauli operators of the quantum Reed--Muller code $\mathrm{QRM}(4)$ according to \cref{def:log_def,def:canonical_indices}.}
        \label{tab:logical_m=4}
\end{table}

It is also known that any classical Reed--Muller code $\textrm{RM}(r,m)$ can be generated by its minimum-weight codewords; see Ref. \cite{MS77} Chapter 13, Theorems 8 and 12. Our theorem above can be viewed as a special case for $\textrm{RM}(m/2-1,m)$.

\begin{example} 
    The following is a choice of stabilizer generators of $\mathrm{QRM}(4)$ alternative to the one in \cref{ex:canonical_indices_m4}, which is constructed from \cref{thm:reduced_weight}. From the generators $\{\mathbf{v}_A: 0\leq|A|\leq 1\}$ of $\mathrm{RM}(1,4)$, we can let $\mathbf{h}_A=\mathbf{v}_A$ when $|A|=1$, and let $\mathbf{h}_\emptyset = \mathbf{1}+\mathbf{v}_4$. Stabilizer generators corresponding to the generators $\{\mathbf{h}_A: 0\leq|A|\leq 1\}$ of $\mathrm{RM}(1,4)$ are described in \cref{tab:all_g_from_hA}. For this choice, all stabilizer generators have weight $2\sqrt{n}=8$. A similar example for $\mathrm{QRM}(6)$ can be found in \cref{sec:ex_m6}.

    \begin{table}[htbp]
    \centering        
    \begin{tabular}{| c | c | c  c  c  c  c  c  c  c  c  c  c  c  c  c  c  c |}
         \hline
         \parbox[t]{3.5mm}{\multirow{6}{*}{\rotatebox[origin=c]{90}{\textrm{$\g{X}$-type}}}} & $A$ & \multicolumn{16}{|c|}{$h_x(A)$} \\
         \cline{2-18}
         & $\emptyset$ & $\g{X}$ & $\g{X}$ & $\g{X}$ & $\g{X}$ & $\g{X}$ & $\g{X}$ & $\g{X}$ & $\g{X}$ & $\g{I}$ & $\g{I}$ & $\g{I}$ & $\g{I}$ & $\g{I}$ & $\g{I}$ & $\g{I}$ & $\g{I}$ \\
         \cline{2-18}
         & $\{1\}$ & $\g{I}$ & $\g{X}$ & $\g{I}$ & $\g{X}$ & $\g{I}$ & $\g{X}$ & $\g{I}$ & $\g{X}$ & $\g{I}$ & $\g{X}$ & $\g{I}$ & $\g{X}$ & $\g{I}$ & $\g{X}$ & $\g{I}$ & $\g{X}$ \\
         \cline{2-18}
         & $\{2\}$ & $\g{I}$ & $\g{I}$ & $\g{X}$ & $\g{X}$ & $\g{I}$ & $\g{I}$ & $\g{X}$ & $\g{X}$ & $\g{I}$ & $\g{I}$ & $\g{X}$ & $\g{X}$ & $\g{I}$ & $\g{I}$ & $\g{X}$ & $\g{X}$ \\
         \cline{2-18}
         & $\{3\}$ & $\g{I}$ & $\g{I}$ & $\g{I}$ & $\g{I}$ & $\g{X}$ & $\g{X}$ & $\g{X}$ & $\g{X}$ & $\g{I}$ & $\g{I}$ & $\g{I}$ & $\g{I}$ & $\g{X}$ & $\g{X}$ & $\g{X}$ & $\g{X}$ \\
         \cline{2-18}
         & $\{4\}$ & $\g{I}$ & $\g{I}$ & $\g{I}$ & $\g{I}$ & $\g{I}$ & $\g{I}$ & $\g{I}$ & $\g{I}$ & $\g{X}$ & $\g{X}$ & $\g{X}$ & $\g{X}$ & $\g{X}$ & $\g{X}$ & $\g{X}$ & $\g{X}$ \\
         \hline
         \parbox[t]{3.5mm}{\multirow{6}{*}{\rotatebox[origin=c]{90}{\textrm{$\g{Z}$-type}}}} & $A$ & \multicolumn{16}{|c|}{$h_z(A)$} \\
         \cline{2-18}
         & $\emptyset$ & $\g{Z}$ & $\g{Z}$ & $\g{Z}$ & $\g{Z}$ & $\g{Z}$ & $\g{Z}$ & $\g{Z}$ & $\g{Z}$ & $\g{I}$ & $\g{I}$ & $\g{I}$ & $\g{I}$ & $\g{I}$ & $\g{I}$ & $\g{I}$ & $\g{I}$ \\
         \cline{2-18}
         & $\{1\}$ & $\g{I}$ & $\g{Z}$ & $\g{I}$ & $\g{Z}$ & $\g{I}$ & $\g{Z}$ & $\g{I}$ & $\g{Z}$ & $\g{I}$ & $\g{Z}$ & $\g{I}$ & $\g{Z}$ & $\g{I}$ & $\g{Z}$ & $\g{I}$ & $\g{Z}$ \\
         \cline{2-18}
         & $\{2\}$ & $\g{I}$ & $\g{I}$ & $\g{Z}$ & $\g{Z}$ & $\g{I}$ & $\g{I}$ & $\g{Z}$ & $\g{Z}$ & $\g{I}$ & $\g{I}$ & $\g{Z}$ & $\g{Z}$ & $\g{I}$ & $\g{I}$ & $\g{Z}$ & $\g{Z}$ \\
         \cline{2-18}
         & $\{3\}$ & $\g{I}$ & $\g{I}$ & $\g{I}$ & $\g{I}$ & $\g{Z}$ & $\g{Z}$ & $\g{Z}$ & $\g{Z}$ & $\g{I}$ & $\g{I}$ & $\g{I}$ & $\g{I}$ & $\g{Z}$ & $\g{Z}$ & $\g{Z}$ & $\g{Z}$ \\
         \cline{2-18}
         & $\{4\}$ & $\g{I}$ & $\g{I}$ & $\g{I}$ & $\g{I}$ & $\g{I}$ & $\g{I}$ & $\g{I}$ & $\g{I}$ & $\g{Z}$ & $\g{Z}$ & $\g{Z}$ & $\g{Z}$ & $\g{Z}$ & $\g{Z}$ & $\g{Z}$ & $\g{Z}$ \\
         \hline
         \parbox[t]{3.5mm}{\multirow{4}{*}{\rotatebox[origin=c]{90}{\textrm{bit indices}}}} & $x_4$ & \multicolumn{1}{|c|}{$0$} & \multicolumn{1}{|c|}{$0$} & \multicolumn{1}{|c|}{$0$} & \multicolumn{1}{|c|}{$0$} & \multicolumn{1}{|c|}{$0$} & \multicolumn{1}{|c|}{$0$} & \multicolumn{1}{|c|}{$0$} & \multicolumn{1}{|c|}{$0$} & \multicolumn{1}{|c|}{$1$} & \multicolumn{1}{|c|}{$1$} & \multicolumn{1}{|c|}{$1$} & \multicolumn{1}{|c|}{$1$} & \multicolumn{1}{|c|}{$1$} & \multicolumn{1}{|c|}{$1$} & \multicolumn{1}{|c|}{$1$} & \multicolumn{1}{|c|}{$1$} \\
         & $x_3$ & \multicolumn{1}{|c|}{$0$} & \multicolumn{1}{|c|}{$0$} & \multicolumn{1}{|c|}{$0$} & \multicolumn{1}{|c|}{$0$} & \multicolumn{1}{|c|}{$1$} & \multicolumn{1}{|c|}{$1$} & \multicolumn{1}{|c|}{$1$} & \multicolumn{1}{|c|}{$1$} & \multicolumn{1}{|c|}{$0$} & \multicolumn{1}{|c|}{$0$} & \multicolumn{1}{|c|}{$0$} & \multicolumn{1}{|c|}{$0$} & \multicolumn{1}{|c|}{$1$} & \multicolumn{1}{|c|}{$1$} & \multicolumn{1}{|c|}{$1$} & \multicolumn{1}{|c|}{$1$} \\
         & $x_2$ & \multicolumn{1}{|c|}{$0$} & \multicolumn{1}{|c|}{$0$} & \multicolumn{1}{|c|}{$1$} & \multicolumn{1}{|c|}{$1$} & \multicolumn{1}{|c|}{$0$} & \multicolumn{1}{|c|}{$0$} & \multicolumn{1}{|c|}{$1$} & \multicolumn{1}{|c|}{$1$} & \multicolumn{1}{|c|}{$0$} & \multicolumn{1}{|c|}{$0$} & \multicolumn{1}{|c|}{$1$} & \multicolumn{1}{|c|}{$1$} & \multicolumn{1}{|c|}{$0$} & \multicolumn{1}{|c|}{$0$} & \multicolumn{1}{|c|}{$1$} & \multicolumn{1}{|c|}{$1$} \\
         & $x_1$ & \multicolumn{1}{|c|}{$0$} & \multicolumn{1}{|c|}{$1$} & \multicolumn{1}{|c|}{$0$} & \multicolumn{1}{|c|}{$1$} & \multicolumn{1}{|c|}{$0$} & \multicolumn{1}{|c|}{$1$} & \multicolumn{1}{|c|}{$0$} & \multicolumn{1}{|c|}{$1$} & \multicolumn{1}{|c|}{$0$} & \multicolumn{1}{|c|}{$1$} & \multicolumn{1}{|c|}{$0$} & \multicolumn{1}{|c|}{$1$} & \multicolumn{1}{|c|}{$0$} & \multicolumn{1}{|c|}{$1$} & \multicolumn{1}{|c|}{$0$} & \multicolumn{1}{|c|}{$1$} \\
         \hline
    \end{tabular}
    \caption{Minimum-weight stabilizer generators of the quantum Reed--Muller code $\mathrm{QRM}(4)$ constructed from \cref{thm:reduced_weight}.}
        \label{tab:all_g_from_hA}
    \end{table}

\end{example}

Quantum Reed--Muller codes are not quantum LDPC codes as the stabilizer weights increase with the block length. \cref{thm:reduced_weight} can help us construct a minimum-weight generating set for the quantum Reed--Muller codes, which can be very useful if the FTEC schemes to be applied are Shor \cite{Shor96,DA07,TPB23} or flag \cite{CR18,CR20} schemes as their circuits scale with the weight of the stabilizer generators to be measured. The stabilizer weights might not affect the scaling of the circuits for Steane \cite{Steane97,Steane04} or Knill \cite{Knill05} schemes directly since their circuits involve transversal gates and are independent of stabilizer weight. However, the stabilizer weights could affect the preparation of the ancilla blocks required for such schemes.
    
\section{Logical Clifford operators of quantum Reed--Muller codes} \label{sec:logical_for_QRM}

In this section, we first provide constructions of fold-transversal gates for the quantum Reed--Muller code $\mathrm{QRM}(m)$ from automorphisms of its corresponding classical Reed--Muller code $\mathrm{RM}(m/2-1,m)$. A sequence of such fold-transversal gates can then be used to construct addressable $\logg{S}$ gates on any logical qubit and addressable controlled-$\g{Z}$ gates on specific logical qubit pairs. With the help of a transversal gate $\g{H}^{\otimes n}$, we later show that addressable $\logg{H}$, addressable $\logg{SW}$, and addressable controlled-$\g{Z}$ gates on any logical qubit (or any pair of logical qubits) can also be constructed. These gates are sufficient to generate the full logical Clifford group for any $\mathrm{QRM}(m)$.

\subsection{Fold-transversal gates constructed from code automorphisms} \label{subsec:fold_from_aut}

Transversal gates \cite{Gottesman97} are known to provide an efficient way to perform fault-tolerant quantum computation on a code space since the implementation has circuit-depth one and no ancilla qubits are required. A transversal gate applied to a single code block involves only single-qubit gates, thus does not spread errors within the code block. For any stabilizer code, any logical Pauli operation can be implemented transversally since there exists a choice of symplectic basis such that the logical Pauli $\g{X}$ ($\g{Z}$) operator on each logical qubit can be written as a tensor product of physical $\g{X}$ ($\g{Z}$) and $\g{I}$ operators \cite{BDH06,Wilde09b}. For other logical operations such as logical Clifford or non-Clifford operations, whether their transversal implementation exists depends on the code symmetry. One well-known example of a stabilizer code that admits
a transversal implementation of several Clifford gates
is the \codepar{7,1,3} Steane code \cite{Steane96}, which is the smallest 2D color code \cite{BM06}. For this code, logical Pauli operators are $\logg{X}=\g{X}^{\otimes 7}$ and $\logg{Z}=\g{Z}^{\otimes 7}$, and $\logg{H}$ and $\logg{S}$ gates can be implemented transversally by $\g{H}^{\otimes 7}$ and $\left(\g{S}^\dagger\right)^{\otimes 7}$, respectively.\footnote{The full logical Clifford group of the Steane code defined on a single code block is $\overline{C}_1=\left\langle\logg{H},\logg{S}\right\rangle$ which can be achieved using only transversal gates, but this is not a high-rate code since $k=1$.}

Fold-transversal gates \cite{Moussa16,BB24} provide an alternative way to implement logical gates with a constant-depth circuit and without ancilla qubits; for a general CSS code, the $\g{X}$- and $\g{Z}$-type stabilizer generators may not overlap, in which case a transversal gate such as $\g{H}^{\otimes n}$ would not be a valid logical operator since it would not preserve the stabilizer group. However, on some CSS codes, it is possible to find a permutation of physical qubits that maps the support of each $\g{X}$-type generator to that of some $\g{Z}$-type generator (and vice versa). Such permutations are called \emph{$\g{ZX}$-dualities}, and a CSS code that possesses $\g{ZX}$-dualities is called \emph{self-$\g{ZX}$-dual} \cite{BB24}. The surface code is an example of a self-$\g{ZX}$-dual CSS code: for the unrotated variant, one can reflect the code along its diagonal line to swap the $\g{X}$- and $\g{Z}$-type generators. For self-$\g{ZX}$-dual CSS codes, some logical gates can be implemented by a constant-depth circuit that involves only single-qubit gates and two-qubit gates with non-overlapping supports; such a logical gate is called fold-transversal. As an example, the $\logg{H}$ gate of an unrotated surface code can be implemented by applying $\g{H}^{\otimes n}$ and swapping each physical qubit with its reflection along the diagonal line. Fold-transversal gates were originally developed for surface codes in \cite{Moussa16}, and were generalized to any self-$\g{ZX}$-dual CSS code in \cite{BB24}. Note that self-dual CSS codes are a special case of self-$\g{ZX}$-dual CSS codes in which the trivial permutation serves as a $\g{ZX}$-duality.

Although the definition of fold-transversal gates in Ref. \cite{BB24} covers transversal gates, we make a distinction between transversal and fold-transversal gates in this work. In particular, a \emph{transversal gate} refers to a logical gate whose implementation comprises only single-qubit gates, i.e., the gate is strictly transversal. In contrast, a \emph{fold-transversal gate} refers to a logical gate which is \emph{not} strictly transversal, whose implementation comprises two-qubit gates with non-overlapping supports and single-qubit gates. In this work, we are interested only in operations inside a single code block, and not those involving two or more blocks.

The transversal and fold-transversal gates developed in this work are related to the Hadamard gate $\g{H}$, the phase gate $\g{S}$, the swap gate $\g{SW}$, the controlled-NOT gate $\g{CX}$, and two types of controlled-$\g{Z}$ gates; namely, the zero-zero-controlled-$\g{Z}$ gate $\g{C_{00}Z}$ and the one-one-controlled-$\g{Z}$ gate $\g{C_{11}Z}$. The actions of these gates in the computational basis can be described by the following matrices:
\begin{equation}
    \g{H} = \frac{1}{\sqrt{2}}\begin{pmatrix*}[r]
    1 & 1 \\
    1 & -1
\end{pmatrix*},\quad
\g{S} = \begin{pmatrix}
    1 & 0 \\
    0 & \up{i}
\end{pmatrix},\quad
\g{SW} = \begin{pmatrix}
    1 & 0 & 0 & 0 \\
    0 & 0 & 1 & 0 \\
    0 & 1 & 0 & 0 \\
    0 & 0 & 0 & 1
\end{pmatrix}, \nonumber
\end{equation}
\begin{equation}
\g{CX} = \begin{pmatrix*}[r]
    1 & 0 & 0 & 0 \\
    0 & 1 & 0 & 0 \\
    0 & 0 & 0 & 1 \\
    0 & 0 & 1 & 0
\end{pmatrix*},\quad
    \g{C_{00}Z} = \begin{pmatrix*}[r]
    -1 & 0 & 0 & 0 \\
    0 & 1 & 0 & 0 \\
    0 & 0 & 1 & 0 \\
    0 & 0 & 0 & 1
\end{pmatrix*},\quad
\g{C_{11}Z} = \begin{pmatrix*}[r]
    1 & 0 & 0 & 0 \\
    0 & 1 & 0 & 0 \\
    0 & 0 & 1 & 0 \\
    0 & 0 & 0 & -1
\end{pmatrix*}, \nonumber
\end{equation}
where $\up{i} = \sqrt{-1}$. Note that $\g{C_{11}Z}$ is the conventional controlled-$\g{Z}$ gate in the literature, and $\g{C_{00}Z} = (\g{X}\otimes \g{X})\g{C_{11}Z}(\g{X}\otimes \g{X})= \g{C_{11}Z}(\g{Z}\otimes \g{Z})$ up to some global phase. Illustrations of $\g{H}$, $\g{S}$, $\g{SW}$, $\g{CX}$, $\g{C_{00}Z}$, and $\g{C_{11}Z}$ gates in a quantum circuit are shown in \cref{fig:basic_gates}.

\begin{figure}[htbp]
	\centering
    \begin{tikzpicture}
\begin{yquant}
nobit n;
qubit {} q[2];

hspace {1.2mm} n;
text {$\g{H}$} n;
box {$\g{H}$} q[0];
barrier (-);

hspace {1mm} n;
text {$\g{S}$} n;
box {$\g{S}$} q[0];
barrier (-);

text {$\g{SW}$} n;
hspace {1mm} q[0]-q[1];
swap (q);
barrier (-);

text {$\g{CX}$} n;
cnot q[1] | q[0];
barrier (-);

text {$\g{C_{00}Z}$} n;
hspace {2mm} q[0]-q[1];
cnot ~q;
barrier (-);

text {$\g{C_{11}Z}$} n;
hspace {2mm} q[0]-q[1];
cnot |q;

\end{yquant}
\end{tikzpicture}
	\caption{Circuit diagrams for the Hadamard gate $\g{H}$, the phase gate $\g{S}$, the swap gate $\g{SW}$, the controlled-NOT gate $\g{CX}$, and two types of controlled-$\g{Z}$ gates $\g{C_{00}Z}$ and $\g{C_{11}Z}$.}
	\label{fig:basic_gates}%
\end{figure}

$\g{H}$, $\g{S}$, $\g{SW}$, $\g{CX}$, $\g{C_{00}Z}$ and $\g{C_{11}Z}$ transform the Pauli $\g{X}$ and $\g{Z}$ operators as follows:
\begin{equation}
    \begin{tabular}{r r c l c r c l}
        $\g{H}$: & $\g{X}$ & $\mapsto$ & $\g{Z}$, & $\quad$ & $\g{Z}$ & $\mapsto$ & $\g{X}$, \\
        $\g{S}$: & $\g{X}$ & $\mapsto$ & $\up{i}\g{XZ}$, & $\quad$ & $\g{Z}$ & $\mapsto$ & $\g{Z}$, \\
        $\g{SW}$: & $\g{X}\otimes\g{I}$ & $\mapsto$ & $\g{I}\otimes\g{X}$, & $\quad$ & $\g{Z}\otimes\g{I}$ & $\mapsto$ & $\g{I}\otimes\g{Z}$, \\
        $\g{CX}$: & $\g{X}\otimes\g{I}$ & $\mapsto$ & $\g{X}\otimes\g{X}$, & $\quad$ & $\g{Z}\otimes\g{I}$ & $\mapsto$ & $\g{Z}\otimes\g{I}$, \\
        & $\g{I}\otimes\g{X}$ & $\mapsto$ & $\g{I}\otimes\g{X}$, & $\quad$ & $\g{I}\otimes\g{Z}$ & $\mapsto$ & $\g{Z}\otimes\g{Z}$, \\
        $\g{C_{00}Z}$: & $\g{X}\otimes\g{I}$ & $\mapsto$ & $-\g{X}\otimes\g{Z}$, & $\quad$ & $\g{Z}\otimes\g{I}$ & $\mapsto$ & $\g{Z}\otimes\g{I}$, \\
        & $\g{I}\otimes\g{X}$ & $\mapsto$ & $-\g{Z}\otimes\g{X}$, & $\quad$ & $\g{I}\otimes\g{Z}$ & $\mapsto$ & $\g{I}\otimes\g{Z}$, \\
        $\g{C_{11}Z}$: & $\g{X}\otimes\g{I}$ & $\mapsto$ & $\g{X}\otimes\g{Z}$, & $\quad$ & $\g{Z}\otimes\g{I}$ & $\mapsto$ & $\g{Z}\otimes\g{I}$, \\
        & $\g{I}\otimes\g{X}$ & $\mapsto$ & $\g{Z}\otimes\g{X}$, & $\quad$ & $\g{I}\otimes\g{Z}$ & $\mapsto$ & $\g{I}\otimes\g{Z}$, \\
    \end{tabular} \nonumber
\end{equation}
where the first and second qubits that the $\g{CX}$ gate acts on are the control and target qubits, respectively.

In this work, we define two types of fold-transversal gates as follows.\footnote{Note that the types of fold-transversal gates considered in this work are slightly different from the ones defined in \cite{BB24}.}
\begin{definition} \label{def:SW_PH_def}
    Suppose that $\pi: \{0,\dots,2^m-1\} \rightarrow \{0,\dots,2^m-1\} $ is a permutation of $2^m$ bits that satisfies $\pi^2 = e$. The \emph{swap-type fold-transversal gate} $\g{U_{SW}}(\pi)$ is constructed by 
    \begin{enumerate}
        \item applying physical $\g{SW}$ to each pair of physical qubits $\{p,\pi(p)\}$ satisfying $p < \pi(p)$, and
        \item applying nothing ($\g{I}$) to any qubit $q$ satisfying $q = \pi(q)$.
    \end{enumerate}
    The \emph{phase-type fold-transversal gate} $\g{U_{PH}}(\pi)$ is constructed by
    \begin{enumerate}
        \item applying physical $\g{C_{11}Z}$ to each pair of physical qubits $\{p,\pi(p)\}$ satisfying $p < \pi(p)$, and
        \item applying physical $\g{S}$ to any qubit $q$ satisfying $q = \pi(q)$.
    \end{enumerate}
    (It is required that the permutation $\pi$ has order 2 to ensure that both $\g{U_{SW}}(\pi)$ and $\g{U_{PH}}(\pi)$ are well-defined. However, this does not guarantee that both $\g{U_{SW}}(\pi)$ and $\g{U_{PH}}(\pi)$ are valid logical operators.)
\end{definition}

\begin{example}
    Consider $\mathrm{QRM}(4)$. Swap-type fold-transversal gates $\g{U_{SW}}(e)$, $\g{U_{SW}}(P(1,2))$, $\g{U_{SW}}(P(3,4))$, $\g{U_{SW}}(Q(1,2))$, $\g{U_{SW}}(Q(3,4))$, and $\g{U_{SW}}(Q(1,2)Q(3,4))$ can be implemented by the circuits in \cref{fig:ex_SW-type}, while phase-type fold-transversal gates $\g{U_{PH}}(e)$, $\g{U_{PH}}(P(1,2))$, $\g{U_{PH}}(P(3,4))$, $\g{U_{PH}}(Q(1,2))$, $\g{U_{PH}}(Q(3,4))$, and $\g{U_{PH}}(Q(1,2)Q(3,4))$ can be implemented by the circuits in \cref{fig:ex_phase-type}.
\end{example}

\begin{figure}[htbp]
	\centering
    \begin{tikzpicture}
\begin{yquant}
nobit n;
qubit {0000} q0;
qubit {0001} q1;
qubit {0010} q2;
qubit {0011} q3;
qubit {0100} q4;
qubit {0101} q5;
qubit {0110} q6;
qubit {0111} q7;
qubit {1000} q8;
qubit {1001} q9;
qubit {1010} q10;
qubit {1011} q11;
qubit {1100} q12;
qubit {1101} q13;
qubit {1110} q14;
qubit {1111} q15;

text {$\g{U_{SW}}(e)$} n;
barrier (-);

text {$\g{U_{SW}}(P(1, 2))$} n;
hspace {9mm} q0-q15;
swap (q1, q2), (q5, q6), (q9, q10), (q13, q14);
barrier (-);

text {$\g{U_{SW}}(P(3, 4))$} n;
hspace {4mm} q0-q15;
swap (q8, q4);
swap (q9, q5);
swap (q10, q6);
swap (q11, q7);
barrier (-);

text {$\g{U_{SW}}(Q(1, 2))$} n;
hspace {9mm} q0-q15;
swap (q6, q7), (q2, q3), (q10, q11), (q14, q15);
barrier (-);

text {$\g{U_{SW}}(Q(3, 4))$} n;
hspace {4mm} q0-q15;
swap (q8, q12);
swap (q9, q13);
swap (q10, q14);
swap (q11, q15);
barrier (-);

text {$\g{U_{SW}}(Q(1, 2)Q(3, 4))$} n;
hspace {10mm} q0-q15;
hspace {4.3mm} q0-q7;
swap (q2, q3), (q6, q7);
swap (q8, q12);
swap (q9, q13);
swap (q10, q15);
swap (q11, q14);

\end{yquant}
\end{tikzpicture}
	\caption{Examples of swap-type fold-transversal gates for the quantum Reed--Muller code $\mathrm{QRM}(4)$.}
	\label{fig:ex_SW-type}%
\end{figure}

\begin{figure}[htbp]
	\centering
    \begin{tikzpicture}
\begin{yquant}
nobit n;
qubit {0000} q0;
qubit {0001} q1;
qubit {0010} q2;
qubit {0011} q3;
qubit {0100} q4;
qubit {0101} q5;
qubit {0110} q6;
qubit {0111} q7;
qubit {1000} q8;
qubit {1001} q9;
qubit {1010} q10;
qubit {1011} q11;
qubit {1100} q12;
qubit {1101} q13;
qubit {1110} q14;
qubit {1111} q15;

hspace {1mm} -;
text {$\g{U_{PH}}(e)$} n;
hspace {3mm} q0-q15;
box {$\g{S}$} q0-q15;
barrier (-);

text {$\g{U_{PH}}(P(1, 2))$} n;
hspace {7.5mm} q0-q15;
hspace {0.8mm} q1, q2, q5, q6, q9, q10, q13, q14;
zz (q1, q2), (q5, q6), (q9, q10), (q13, q14);
box {$\g{S}$} q0, q3, q4, q7, q8, q11, q12, q15;
barrier (-);

text {$\g{U_{PH}}(P(3, 4))$} n;
hspace {4mm} q0-q15;
hspace {3.5mm} q0-q3, q12-q15;
zz (q8, q4);
zz (q9, q5);
zz (q10, q6);
zz (q11, q7);
box {$\g{S}$} q0-q3, q12-q15;
barrier (-);

text {$\g{U_{PH}}(Q(1, 2))$} n;
hspace {7.5mm} q0-q15;
hspace {0.8mm} q2, q3, q6, q7, q10, q11, q14, q15;
zz (q2, q3);
zz (q6, q7);
zz (q10, q11);
zz (q14, q15);
box {$\g{S}$} q0, q1, q4, q5, q8, q9, q12, q13;
barrier (-);

text {$\g{U_{PH}}(Q(3, 4))$} n;
hspace {4mm} q0-q15;
hspace {3.5mm} q0-q7;
zz (q8, q12);
zz (q9, q13);
zz (q10, q14);
zz (q11, q15);
box {$\g{S}$} q0-q7;
barrier (-);

text {$\g{U_{PH}}(Q(1, 2)Q(3, 4))$} n;
hspace {10mm} q0-q15;
hspace {3.5mm} q0-q7;
hspace {0.8mm} q2, q3, q6, q7;
zz (q2, q3), (q6, q7);
zz (q8, q12);
zz (q9, q13);
zz (q10, q15);
zz (q11, q14);
box {$\g{S}$} q0, q1, q4, q5;

\end{yquant}
\end{tikzpicture}
	\caption{Examples of phase-type fold-transversal gates for the quantum Reed--Muller code $\mathrm{QRM}(4)$.}
	\label{fig:ex_phase-type}%
\end{figure}

Swap-type fold-transversal gates $\g{U_{SW}}(\pi)$ and phase-type fold transversal gates $\g{U_{PH}}(\pi)$ transform any Pauli-$\g{X}$ or Pauli-$\g{Z}$ operator as follows:
\begin{proposition} \label{prop:SW_PH_transform}
    Let $\mathbf{w} \in \mathbb{Z}_2^{2^m}$ be a binary vector and suppose that the permutation $\pi$ satisfies $\pi^2=e$. Then, 
    \begin{enumerate}
        \item $\g{U_{SW}}(\pi)$ transforms $\g{X}(\mathbf{w}) \mapsto \g{X}(M_{\pi}\mathbf{w})$ and $\g{Z}(\mathbf{w}) \mapsto \g{Z}(M_{\pi}\mathbf{w})$;
        \item $\g{U_{PH}}(\pi)$ transforms $\g{X}(\mathbf{w}) \mapsto (\up{i})^c \g{X}(\mathbf{w})\g{Z}(M_{\pi}\mathbf{w})$ and $\g{Z}(\mathbf{w}) \mapsto \g{Z}(\mathbf{w})$, where $c = \mathrm{wt}(\mathbf{w} \wedge M_{\pi}\mathbf{w})$.
    \end{enumerate}
\end{proposition}

\begin{proof}
    In what follows, a single-qubit physical gate $\g{V}$ applied to the $p$-th physical qubit is denoted by $\g{V}_p$, and a two-qubit physical gate $\g{W}$ applied to the $p$-th and $q$-th physical qubits is denoted by $\g{W}_{p,q}$.
    \begin{enumerate}[label=(\arabic*)]
        \item From the definition of $\g{U_{SW}}(\pi)$, physical qubits are permuted by swap gates similar to how the physical bits of $\mathbf{w}$ are permuted by $\pi$. Therefore, $\g{X}(\mathbf{w})$ (or $\g{Z}(\mathbf{w})$) is mapped to $\g{X}(M_{\pi}\mathbf{w})$ (or $\g{Z}(M_{\pi}\mathbf{w})$).
        \item All physical gates in $\g{U_{PH}}(\pi)$ act trivially on Pauli-$\g{Z}$ operators, so $\g{Z}(\mathbf{w})$ is mapped to itself. Next, let $\g{X}_p$ (or $\g{Z}_p$ or $\g{S}_p$) denote an $\g{X}$ (or a $\g{Z}$ or an $\g{S}$) operator acting on the $p$-th qubit of the code block. Consider the $p$-th bit of $\mathbf{w}$ in which the bit value is 1. The bit corresponds to the tensor factor $\g{X}_p$ of the operator $\g{X}(\mathbf{w})$. If $p$ satisfies $p=\pi(p)$, then $\g{U_{PH}}(\pi)$ applies $\g{S}_p$, so $\g{X}_p$ is mapped to $\up{i}\g{X}_p\g{Z}_p =\up{i}\g{X}_p\g{Z}_{\pi(p)}$. If $p$ satisfies $p < \pi(p)$, then $\g{U_{PH}}(\pi)$ applies $(\g{C_{11}Z})_{p,\pi(p)}$, so $\g{X}_p$ is mapped to $\g{X}_p\g{Z}_{\pi(p)}$, and $\g{X}_{\pi(p)}$ is mapped to $\g{X}_{\pi(p)}\g{Z}_{p}$. Therefore, $\g{U_{PH}}(\pi)$ maps $\g{X}(\mathbf{w})$ to $(\up{i})^\alpha \g{X}(\mathbf{w})\g{Z}(M_{\pi}\mathbf{w})$ where $(\up{i})^\alpha$ is some phase factor and $\alpha$ is some integer. Recall that a physical $\g{S}$ maps $\g{X}$ to $\g{Y}$, and a physical $\g{CZ}$ gate maps $\g{X}\otimes \g{I}$ to $\g{X}\otimes \g{Z}$, $\g{I}\otimes \g{X}$ to $\g{Z}\otimes \g{X}$, and $\g{X}\otimes \g{X}$ to $\g{Y}\otimes \g{Y}$. That is, under the operation of $\g{U_{PH}}\left(\pi\right)$, $\g{X}(\mathbf{w})$ is mapped to some Pauli operator $\g{P}$ which is a tensor product of $\g{I},\g{X},\g{Y},\g{Z}$ \emph{without} phase factor. Let $\g{P}_X$ and $\g{P}_Z$ be the $\g{X}$-part and the $\g{Z}$-part of $\g{P}$, respectively. Because $\g{Y}=\up{i}\g{XZ}$, we can write $\g{P} = (\up{i})^c \g{P}_X\g{P}_Z$ where $c$ is the number of $\g{Y}$'s in $\g{P}$. Since $\g{P}_X = \g{X}(\mathbf{w})$ and $\g{P}_Z=\g{Z}(M_{\pi}\mathbf{w})$, we have that $(\up{i})^\alpha = (\up{i})^c$ where $c$ is the number of bits where $\mathbf{w}$ and $M_{\pi}\mathbf{w}$ overlap, given by $\mathrm{wt}(\mathbf{w} \wedge M_{\pi}\mathbf{w})$. That is, $\g{U_{PH}}(\pi)$ maps $\g{X}(\mathbf{w})$ to $(\up{i})^c \g{X}(\mathbf{w})\g{Z}(M_{\pi}\mathbf{w})$ where $c=\mathrm{wt}(\mathbf{w} \wedge M_{\pi}\mathbf{w})$.
        \qedhere
    \end{enumerate}
\end{proof}
It should be noted that whether $\g{U_{SW}}(\pi)$ and $\g{U_{PH}}(\pi)$ are valid logical operators for $\mathrm{QRM}(m)$ or not may depend on the permutation $\pi$. Also, swap-type and phase-type fold-transversal gates according to \cref{def:SW_PH_def} may not be valid logical operators for a general self-$\g{ZX}$-dual CSS code.

\begin{example} \label{ex:Q12_on_QRM4}
    Consider the gate $\g{U_{PH}}(Q(1,2))$ on $\mathrm{QRM}(4)$. Let us first observe its action on Pauli operators $\g{X}(\mathbf{v}_\emptyset)$, $\g{X}(\mathbf{v}_1)$ and $\g{X}(\mathbf{v}_l)$, $l\neq 1$. $\g{U_{PH}}(Q(1,2))$ gives the following transformation:
\begin{equation}
    \begin{tabular}{r r c l}
        $\g{U_{PH}}(Q(1,2))$: & $\g{X}(\mathbf{v}_\emptyset)$ & $\mapsto$ & $\g{X}(\mathbf{v}_\emptyset)\g{Z}(\mathbf{v}_\emptyset)$, \\
        & $\g{X}(\mathbf{v}_1)$ & $\mapsto$ & $\g{X}(\mathbf{v}_1)\g{Z}(\mathbf{v}_1+\mathbf{v}_2)=\g{X}(\mathbf{v}_1)\g{Z}(\mathbf{v}_1)\g{Z}(\mathbf{v}_2)$, \\
        & $\g{X}(\mathbf{v}_l)$ & $\mapsto$ & $\g{X}(\mathbf{v}_l)\g{Z}(\mathbf{v}_l)\quad\textrm{for any}\quad l\neq 1$.
    \end{tabular} \nonumber
\end{equation}
$\g{U_{PH}}(Q(1,2))$ also transforms $\g{Z}(\mathbf{v}_\emptyset)$ or any $\g{Z}(\mathbf{v}_i)$ to itself. Since the stabilizer group of $\mathrm{QRM}(4)$ is preserved under this transformation, $\g{U_{PH}}(Q(1,2))$ is a valid logical operator.

Next, consider the action of $\g{U_{PH}}(Q(1,2))$ on logical $\g{X}$ operators. We find the following:
\begin{equation}
    \begin{tabular}{r r c l}
        $\g{U_{PH}}(Q(1,2))$: & $\logg{X}(1)=\g{X}(\mathbf{v}_1\mathbf{v}_2)$ & $\mapsto$ & $\phantom{+}\g{X}(\mathbf{v}_1\mathbf{v}_2)\g{Z}((\mathbf{v}_1+\mathbf{v}_2)\mathbf{v}_2)=\phantom{+}\g{X}(\mathbf{v}_1\mathbf{v}_2)\g{Z}(\mathbf{v}_1\mathbf{v}_2)\g{Z}(\mathbf{v}_2)$,\\
        & $\logg{X}(2)=\g{X}(\mathbf{v}_1\mathbf{v}_3)$ & $\mapsto$ & $-\g{X}(\mathbf{v}_1\mathbf{v}_3)\g{Z}((\mathbf{v}_1+\mathbf{v}_2)\mathbf{v}_3)=-\g{X}(\mathbf{v}_1\mathbf{v}_3)\g{Z}(\mathbf{v}_1\mathbf{v}_3)\g{Z}(\mathbf{v}_2\mathbf{v}_3)$, \\
        & $\logg{X}(3)=\g{X}(\mathbf{v}_1\mathbf{v}_4)$ & $\mapsto$ & $-\g{X}(\mathbf{v}_1\mathbf{v}_4)\g{Z}((\mathbf{v}_1+\mathbf{v}_2)\mathbf{v}_4)=-\g{X}(\mathbf{v}_1\mathbf{v}_4)\g{Z}(\mathbf{v}_1\mathbf{v}_4)\g{Z}(\mathbf{v}_2\mathbf{v}_4)$, \\
        & $\logg{X}(4)=\g{X}(\mathbf{v}_3\mathbf{v}_4)$ & $\mapsto$ & $\phantom{+}\g{X}(\mathbf{v}_3\mathbf{v}_4)\g{Z}(\mathbf{v}_3\mathbf{v}_4)$, \\
        & $\logg{X}(5)=\g{X}(\mathbf{v}_2\mathbf{v}_4)$ & $\mapsto$ & $\phantom{+}\g{X}(\mathbf{v}_2\mathbf{v}_4)\g{Z}(\mathbf{v}_2\mathbf{v}_4)$,  \\
        & $\logg{X}(6)=\g{X}(\mathbf{v}_2\mathbf{v}_3)$ & $\mapsto$ & $\phantom{+}\g{X}(\mathbf{v}_2\mathbf{v}_3)\g{Z}(\mathbf{v}_2\mathbf{v}_3)$.
    \end{tabular} \nonumber
\end{equation}
$\g{U_{PH}}(Q(1,2))$ also transforms any logical $\g{Z}$ operator to itself. The logical action of $\g{U_{PH}}(Q(1,2))$ is as follows:
\begin{equation}
    \begin{tabular}{r r c l r c l c r c l r c l}
        $\g{U_{PH}}(Q(1,2))$: & $\logg{X}(1)$ & $\mapsto$ & $\phantom{+}\logg{X}(1)\logg{Z}(4)$,  & $\logg{Z}(1)$ & $\mapsto$ & $\logg{Z}(1)$, & \quad & $\logg{X}(4)$ & $\mapsto$ & $\logg{X}(4)\logg{Z}(1)$,  & $\logg{Z}(4)$ & $\mapsto$ & $\logg{Z}(4)$,\\
        & $\logg{X}(2)$ & $\mapsto$ & $-\logg{X}(2)\logg{Z}(3)\logg{Z}(5)$,  & $\logg{Z}(2)$ & $\mapsto$ & $\logg{Z}(2)$, & \quad & $\logg{X}(5)$ & $\mapsto$ & $\logg{X}(5)\logg{Z}(2)$,  & $\logg{Z}(5)$ & $\mapsto$ & $\logg{Z}(5)$,\\
        & $\logg{X}(3)$ & $\mapsto$ & $-\logg{X}(3)\logg{Z}(2)\logg{Z}(6)$,  & $\logg{Z}(3)$ & $\mapsto$ & $\logg{Z}(3)$, & \quad & $\logg{X}(6)$ & $\mapsto$ & $\logg{X}(6)\logg{Z}(3)$,  & $\logg{Z}(6)$ & $\mapsto$ & $\logg{Z}(6)$.
    \end{tabular} \nonumber
\end{equation}
Therefore, $\g{U_{PH}}(Q(1,2))$ acts as $\logg{C_{11}Z}(1,4)\logg{C_{11}Z}(2,5)\logg{C_{11}Z}(3,6)\logg{C_{00}Z}(2,3)$ on $\mathrm{QRM}(4)$.
\end{example}

Next, we prove that when the permutation $\pi$ is either $P(i,j)$ or $Q(i,j)$ as defined in \cref{subsec:CRM_aut}, both $\g{U_{SW}}(\pi)$ and $\g{U_{PH}}(\pi)$ preserve the stabilizer group of $\mathrm{QRM}(m)$, thus are valid logical operators.
\begin{theorem}[Code-space preservation of fold-transversal gates from $P(i,j)$ and $Q(i,j)$] \label{thm:PnQ_preserve_stb}
    Let $m$ be a positive even number, $\mathrm{QRM}(m)$ be the quantum Reed--Muller code defined in \cref{def:QRM}, and $P(i,j)$ and $Q(i,j)$ be permutations defined in \cref{def:PQ_def} where $i,j \in [m]$, $i\neq j$. Then, the stabilizer group of $\mathrm{QRM}(m)$ is preserved under the operation of $\g{U_{SW}}\left(P(i,j)\right)$, $\g{U_{PH}}\left(P(i,j)\right)$, $\g{U_{SW}}\left(Q(i,j)\right)$, or $\g{U_{PH}}\left(Q(i,j)\right)$.
\end{theorem} 
\begin{proof}
    We start by considering $\pi=P(i,j)$. Recall that $P(i,j)$ maps $\mathbf{v}_i$ to $\mathbf{v}_j$, maps $\mathbf{v}_j$ to $\mathbf{v}_i$, and maps $\mathbf{v}_l$ to itself if $l \neq i,j$. Any $\g{X}$-type ($\g{Z}$-type) stabilizer generator of $\mathrm{QRM}(m)$ is of the form $g_x(A) = \g{X}\left(\mathbf{v}_A\right)$ ($g_z(A) = \g{Z}\left(\mathbf{v}_A\right)$) where $A \subsetneq [m], 0\leq|A|\leq m/2-1$. By \cref{prop:SW_PH_transform,prop:permute_vA}, $\g{U_{SW}}\left(P(i,j)\right)$ maps $g_x(A)$ as follows.
    \begin{enumerate}
        \item If either $\{i,j\} \subsetneq A$ or $\{i,j\} \cap A =\emptyset$, then $g_x(A) \mapsto g_x(A)$.
        \item If $i \in A$ but $j \notin A$, then $g_x(A) \mapsto g_x(A')$ where $A'=(A\cup\{j\})\setminus\{i\}$.
        \item If $j \in A$ but $i \notin A$, then $g_x(A) \mapsto g_x(A')$ where $A'=(A\cup\{i\})\setminus\{j\}$.
    \end{enumerate}
    Note that if $\g{U_{SW}}\left(P(i,j)\right)$ maps $g_x(A)$ to $g_x(A')$, it also maps $g_x(A')$ to $g_x(A)$. Similar arguments are applicable to the transformation of $g_z(A)$ by $\g{U_{SW}}\left(P(i,j)\right)$. Thus, the stabilizer group of $\mathrm{QRM}(m)$ which is $\left\langle g_x(A), g_z(A)\right\rangle_{A \subsetneq [m],0 \leq |A| \leq \frac{m}{2}-1}$ is preserved under the operation of $\g{U_{SW}}\left(P(i,j)\right)$.

    For any $A \subsetneq [m], 0\leq|A|\leq m/2-1$, $\g{U_{PH}}\left(P(i,j)\right)$ maps $g_z(A)$ to itself and maps $g_x(A)$ as follows.
    \begin{enumerate}
        \item If either $\{i,j\} \subsetneq A$ or $\{i,j\} \cap A =\emptyset$, then $g_x(A) \mapsto (\up{i})^c g_x(A)g_z(A)$ where $c = \mathrm{wt}(\mathbf{v}_A \wedge \mathbf{v}_A) = \mathrm{wt}(\mathbf{v}_A)$. By \cref{prop:Hamming_weight}, $c = 2^{m-|A|}$. Because $0\leq|A|\leq m/2-1$ and $m \geq 2$, $c$ is always divisible by 4. Thus, $g_x(A)$ is mapped to $g_x(A)g_z(A)$.
        \item If $i \in A$ but $j \notin A$ (which can occur only when $m \geq 4$), then $g_x(A) \mapsto (\up{i})^c g_x(A)g_z(A')$ where $A'=(A\cup\{j\})\setminus\{i\}$ and $c = \mathrm{wt}(\mathbf{v}_A \wedge \mathbf{v}_{A'}) = \mathrm{wt}(\mathbf{v}_{A\cup\{j\}})$. By \cref{prop:Hamming_weight}, $c = 2^{m-|A|-1}$. Because $0\leq|A|\leq m/2-1$ and $m \geq 4$, $c$ is always divisible by 4. Thus, $g_x(A)$ is mapped to $g_x(A)g_z(A')$.
        \item If $j \in A$ but $i \notin A$ (which can occur only when $m \geq 4$), then $g_x(A) \mapsto (\up{i})^c g_x(A)g_z(A')$ where $A'=(A\cup\{i\})\setminus\{j\}$ and $c = \mathrm{wt}(\mathbf{v}_A \wedge \mathbf{v}_{A'}) = \mathrm{wt}(\mathbf{v}_{A\cup\{i\}})$. Similarly to the previous case, $c = 2^{m-|A|-1}$ is always divisible by 4. Thus, $g_x(A)$ is mapped to $g_x(A)g_z(A')$.
    \end{enumerate}
    Therefore, the stabilizer group of $\mathrm{QRM}(m)$ is preserved under the operation of $\g{U_{PH}}\left(P(i,j)\right)$.
    
    Next, we consider $\pi=Q(i,j)$. Recall that $Q(i,j)$ maps $\mathbf{v}_i$ to $\mathbf{v}_i+\mathbf{v}_j$, and maps $\mathbf{v}_l$ to itself if $l \neq i$. For any $A \subsetneq [m], 0\leq|A|\leq m/2-1$, $\g{U_{SW}}\left(Q(i,j)\right)$ maps $g_x(A)$ as follows.
    \begin{enumerate}
        \item If $i \notin A$, then $g_x(A) \mapsto g_x(A)$.
        \item If $i \in A$, then $g_x(A) \mapsto g_x(A)g_x(A')$ where $A'=(A\cup\{j\})\setminus\{i\}$.
    \end{enumerate}
    Similar arguments are applicable to the transformation of $g_z(A)$ by $\g{U_{SW}}\left(Q(i,j)\right)$. Thus, the stabilizer group of $\mathrm{QRM}(m)$ is preserved under the operation of $\g{U_{SW}}\left(Q(i,j)\right)$.

    For any $A \subsetneq [m], 0\leq|A|\leq m/2-1$, $\g{U_{PH}}\left(Q(i,j)\right)$ maps $g_z(A)$ to itself and maps $g_x(A)$ as follows.
    \begin{enumerate}
        \item If $i \notin A$, then $g_x(A) \mapsto (\up{i})^c g_x(A)g_z(A)$ where $c = \mathrm{wt}(\mathbf{v}_A \wedge \mathbf{v}_A)= \mathrm{wt}(\mathbf{v}_A) = 2^{m-|A|}$ (by \cref{prop:Hamming_weight}). Because $0\leq|A|\leq m/2-1$ and $m \geq 2$, $c$ is always divisible by 4. Thus, $g_x(A)$ is mapped to $g_x(A)g_z(A)$.
        \item If $i \in A$, then $g_x(A) \mapsto (\up{i})^c g_x(A)g_z(A)g_z(A')$ where $A'=(A\cup\{j\})\setminus\{i\}$ and $c = \mathrm{wt}(\mathbf{v}_A \wedge (\mathbf{v}_A+\mathbf{v}_{A'}))$. 
        \begin{enumerate}
            \item If $j \in A$, then $\mathbf{v}_A \wedge (\mathbf{v}_A+\mathbf{v}_{A'}) = \mathbf{v}_A+\mathbf{v}_A = \mathbf{0}$, thus $c=0$.
            \item If $j \notin A$ (which can occur only when $m \geq 4$), then $\mathbf{v}_A \wedge (\mathbf{v}_A+\mathbf{v}_{A'}) = \mathbf{v}_A+\mathbf{v}_{A\cup\{j\}} = \mathbf{v}_A\wedge(\mathbf{1}+\mathbf{v}_j)$. Thus, $c = 2^{m-|A|-1}$ by \cref{prop:Hamming_weight}. Since $0\leq|A|\leq m/2-1$ and $m \geq 4$, $c$ is always divisible by 4.
        \end{enumerate}
        In both cases, $g_x(A)$ is mapped to $g_x(A)g_z(A)g_z(A')$.
    \end{enumerate}
    Therefore, the stabilizer group of $\mathrm{QRM}(m)$ is preserved under the operation of $\g{U_{PH}}\left(Q(i,j)\right)$.
\end{proof}

A product of automorphisms is an automorphism, so it is also possible to construct fold-transversal gates from a product of automorphisms. In this work, we are particularly interested in phase-type fold-transversal gates constructed from products of $Q(i,j)$, which will be later used as ingredients to construct addressable gates for $\mathrm{QRM}(m)$. We introduce the following notation to describe such products of automorphisms.
\begin{definition} \label{def:K}
    Let $m$ be a positive even number and let $K=\{(i_1,i_2),\dots,(i_{2c-1},i_{2c})\}$ be a set of $c$ ordered pairs such that $i_1,\dots,i_{2c}\in[m]$ are all distinct. We define $Q(K)$ to be the product $Q(i_1,i_2)\cdots Q(i_{2c-1},i_{2c})$, define $F_1(K) = \{i_1,i_3,\dots,i_{2c-1}\}$ to be the set of the first elements from each ordered pair in $K$, and define $F_2(K) = \{i_2,i_4,\dots,i_{2c}\}$ to be the set of the second arguments of $K$. We also define $Q(\emptyset)=e$ (the trivial permutation), with $F_1(\emptyset)=F_2(\emptyset)=\emptyset$.
\end{definition}
Because we require the vector indices $i_1,\dots,i_{2c}\in[m]$ that appear in $K$ to all be distinct, the order of the permutations in the product does not matter. For example, we have that $Q(K)=Q(i_1,i_2)\cdots Q(i_{2c-1},i_{2c})=Q(i_{2c-1},i_{2c})\cdots Q(i_1,i_2)$. The commutation relations are also applicable to their corresponding permutation matrices; $M_{Q(K)}=M_{Q(i_1,i_2)}\cdots M_{Q(i_{2c-1},i_{2c})}=M_{Q(i_{2c-1},i_{2c})}\cdots M_{Q(i_1,i_2)}$. The same requirement of $i_1,\dots,i_{2c}$ also limits the maximum size of $K$ to $m/2$.

\begin{example}
    Consider the gates $\g{U_{PH}}(Q(1,2))$ and $\g{U_{PH}}(Q(1,2)Q(3,4))$ on $\mathrm{QRM}(4)$. $\g{U_{PH}}(Q(1,2))$ can be written as $\g{U_{PH}}(K)$ with $K=\{(1,2)\}$, in which $F_1(K)=\{1\}$ and $F_2(K)=\{2\}$. Meanwhile, $\g{U_{PH}}(Q(1,2)Q(3,4))$ can be written as $\g{U_{PH}}(K)$ with $K=\{(1,2),(3,4)\}$, in which $F_1(K)=\{1,3\}$ and $F_2(K)=\{2,4\}$.
\end{example}

For any $K$, it is possible to show that $\g{U_{PH}}\left(Q(K)\right)$ is a valid logical operator.
\begin{theorem}[Code-space preservation of phase-type fold-transversal gates from $Q(K)$]  \label{thm:QK_preserve_stb}
    Let $m$ be a positive even number, $\mathrm{QRM}(m)$ be the quantum Reed--Muller code defined in \cref{def:QRM}, and $K$ be a set of ordered pairs defined in \cref{def:K} of size $0 \leq |K| \leq m/2$. Then, for any $K$, the stabilizer group of $\mathrm{QRM}(m)$ is preserved under the operation of $\g{U_{PH}}\left(Q(K)\right)$.
\end{theorem}
A proof of \cref{thm:QK_preserve_stb} is provided in \cref{sec:proof_of_thms}. 

\subsection{Addressable gates from transversal and fold-transversal gates} \label{subsec:addr_gates}
\subsubsection{Addressable $\logg{S}$, $\logg{S}^\dagger$, and $\logg{C_{00}Z}$ gates}

Recall that for our choice of logical Pauli operators, we can describe any logical qubit using its logical qubit index $B$, which is a set of basis vector indices of size $|B|=m/2$ (\cref{def:log_def}). Similarly, we can describe any logical operator on one or two logical qubits using the following notations.
\begin{definition}
    Let $\g{V}$ (or $\g{W}$) be any single-qubit (or two-qubit) unitary operator. The corresponding logical operator that acts on the logical qubit with logical qubit index $B$ (or two logical qubits with logical qubit indices $B$, $B'$) is denoted by $\logg{V}(B)$ (or $\logg{W}(B,B')$). When $B=B_i$ and $B'=B_j$ are canonical logical qubit indices according to \cref{def:canonical_indices}, we may use a shorthand notation $\logg{V}(i)=\logg{V}(B_i)$ (or $\logg{W}(i,j)=\logg{W}(B_i,B_j)$).
\end{definition}
Examples of such logical operators are the logical Hadamard gate $\logg{H}(B)$, the logical phase gate $\logg{S}(B)$, the logical swap gate $\logg{SW}(B,B')$, and two types of logical controlled-$\g{Z}$ gates $\logg{C_{00}Z}(B,B')$ and $\logg{C_{11}Z}(B,B')$.

Next, we consider the logical operation of the phase-type fold-transversal gate $\g{U_{PH}}\left(Q(K)\right)$ corresponding to an automorphism $Q(K)$. On our choice of symplectic basis (\cref{def:log_def,def:canonical_indices}), $\g{U_{PH}}\left(Q(K)\right)$ generally does not act as an addressable gate on the logical subspace; it acts as a product of $\logg{S}$, $\logg{C_{00}Z}$, and $\logg{C_{11}Z}$ on multiple logical qubits. The logical operation of any $\g{U_{PH}}\left(Q(K)\right)$ can be found using the following theorem.

\begin{theorem}[Logical operation of a phase-type fold-transversal gate]\label{thm:logical_from_fold}
    Let $m$ be a positive even number, $\mathrm{QRM}(m)$ be the quantum Reed--Muller code defined in \cref{def:QRM}, and $K$ be a set of ordered pairs defined in \cref{def:K} of size $0 \leq |K| \leq m/2$. The following logical gates are implemented by the action of $\g{U_{PH}}\left(Q(K)\right)$: 
    \begin{enumerate}
        \item $\logg{C_{11}Z}(B,B')$ is applied to any unordered pair of logical qubits $\{B,B'\}$ satisfying $F_1(L) \subseteq B$, $F_2(L)\cap B = \emptyset$, and
        \begin{equation}
            B' =  \left(B^\up{c}\cup F_1(L)\right)\setminus F_2(L),
        \end{equation} for all $L \subseteq K$ such that $|L| \leq m/2-2$.
        \item Additionally, if $|K| \geq m/2-1$, $\logg{C_{00}Z}(D,D')$ is applied to any unordered pair of logical qubits $\{D,D'\}$ satisfying $F_1(L) \subseteq D$, $F_2(L)\cap D = \emptyset$, and
        \begin{equation}
            D' =  \left(D^\up{c}\cup F_1(L)\right)\setminus F_2(L),
        \end{equation} for all $L \subseteq K$ such that $|L| = m/2-1$.
        \item Additionally, if $|K| = m/2$, $\logg{S}(F_1(K))$ is applied when $m/2$ is even (or $\logg{S}^{\dagger}(F_1(K))$ is applied when $m/2$ is odd).
    \end{enumerate}
\end{theorem}

A proof of \cref{thm:logical_from_fold} is provided in \cref{sec:proof_of_thms}.

\begin{example} \label{ex:logical_from_fold}
    Consider $\mathrm{QRM}(4)$ with the canonical logical qubit indices, i.e., $B_1=\{1,2\}$, $B_2=\{1,3\}$, $B_3=\{1,4\}$, $B_4=\{3,4\}$, $B_5=\{2,4\}$, and $B_6=\{2,3\}$. Let us first observe the logical operation of $\g{U_{PH}}(Q(1,2))$ according to \cref{thm:logical_from_fold}. In this case, $K=\{(1,2)\}$. Since $|K|=1=m/2-1$, Statements 1 and 2 of \cref{thm:logical_from_fold} are applicable. 
\begin{enumerate}
    \item There is only one subset $L \subseteq K$ of size $|L| \leq m/2-2=0$, which is $L = \emptyset$. Here we have that $F_1(L)=F_2(L)=\emptyset$. There are three unordered pairs $\{B,B'\}$ that satisfy the conditions $\emptyset \subseteq B$, $\emptyset \cap B = \emptyset$, and $B' =  \left(B^\up{c}\cup \emptyset \right)\setminus \emptyset = B^\up{c}$. Such pairs are $\{B_1,B_4\}$, $\{B_2,B_5\}$, and $\{B_3,B_6\}$. Thus, Statement 1 implies that $\logg{C_{11}Z}(1,4)\logg{C_{11}Z}(2,5)\logg{C_{11}Z}(3,6)$ is a part of the logical operation of $\g{U_{PH}}(Q(1,2))$.
    \item There is only one subset $L \subseteq K$ of size $|L| = m/2-1=1$, which is $L = \{(1,2)\}$. Here we have that $F_1(L)=\{1\}$ and $F_2(L)=\{2\}$. There is only one unordered pair $\{D,D'\}$ that satisfies the conditions $\{1\} \subseteq D$, $\{2\} \cap D = \emptyset$, and $D' =  \left(D^\up{c}\cup \{1\} \right)\setminus \{2\}$. The pair is $\{B_2,B_3\}$. Thus, Statement 2 implies that $\logg{C_{00}Z}(2,3)$ is a part of the logical operation of $\g{U_{PH}}(Q(1,2))$.
\end{enumerate}
Combining both statements, the logical operation of $\g{U_{PH}}(Q(1,2))$ is $\logg{C_{11}Z}(1,4)\logg{C_{11}Z}(2,5)\logg{C_{11}Z}(3,6)\logg{C_{00}Z}(2,3)$, which matches our derivation in \cref{ex:Q12_on_QRM4}.

We can also extend the application of \cref{thm:logical_from_fold} to the other gates such as $\g{U_{PH}}(e)$, $\g{U_{PH}}(Q(3,4))$, and $\g{U_{PH}}(Q(1,2)Q(3,4))$. The logical operations of these gates are shown below.
    \begin{equation}
        \begin{tabular}{| c | c | c |}
            \hline
             $K$ & $\g{U_{PH}}(Q(K))$ & \textrm{Logical operation} \\
             \hline
             $\emptyset$ & $\g{U_{PH}}(e)$ & $\logg{C_{11}Z}(1,4)\logg{C_{11}Z}(2,5)\logg{C_{11}Z}(3,6)$\\
             \hline
             $\{(1,2)\}$ & $\g{U_{PH}}(Q(1,2))$ & $\logg{C_{11}Z}(1,4)\logg{C_{11}Z}(2,5)\logg{C_{11}Z}(3,6)\logg{C_{00}Z}(2,3)$\\
             \hline
             $\{(3,4)\}$ & $\g{U_{PH}}(Q(3,4))$ & $\logg{C_{11}Z}(1,4)\logg{C_{11}Z}(2,5)\logg{C_{11}Z}(3,6)\logg{C_{00}Z}(2,6)$\\
             \hline
             $\{(1,2),(3,4)\}$ & $\g{U_{PH}}(Q(1,2)Q(3,4))$ & $\logg{C_{11}Z}(1,4)\logg{C_{11}Z}(2,5)\logg{C_{11}Z}(3,6)\logg{C_{00}Z}(2,3)\logg{C_{00}Z}(2,6)\logg{S}(2)$ \\
             \hline
        \end{tabular} \nonumber
    \end{equation}
    These logical operations can be illustrated by the circuit in \cref{fig:ex_phase-type_logical}, in which each line represents a logical qubit and each gate represents a logical gate on $\mathrm{QRM}(4)$.
\end{example}

\begin{figure}[htbp]
	\centering
    \begin{tikzpicture}
\begin{yquant}
nobit n;
qubit {$B_1$} q1;
qubit {$B_2$} q2;
qubit {$B_3$} q3;
qubit {$B_4$} q4;
qubit {$B_5$} q5;
qubit {$B_6$} q6;

hspace {1mm} -;
text {$\g{U_{PH}}(e)$} n;
zz (q1, q4);
zz (q2, q5);
zz (q3, q6);
barrier (-);

text {$\g{U_{PH}}(Q(1, 2))$} n;
hspace {4mm} q1-q6;
zz (q1, q4);
zz (q2, q5);
zz (q3, q6);
cnot ~ q2, q3;
barrier (-);

text {$\g{U_{PH}}(Q(3, 4))$} n;
hspace {4mm} q1-q6;
zz (q1, q4);
zz (q2, q5);
zz (q3, q6);
cnot ~ q2, q6;
barrier (-);

text {$\g{U_{PH}}(Q(1, 2)Q(3, 4))$} n;
hspace {5mm} q1-q6;
zz (q1, q4);
zz (q2, q5);
zz (q3, q6);
cnot ~ q2, q3;
cnot ~ q2, q6;
box {$\logg{S}$} q2;

\end{yquant}
\end{tikzpicture}
	\caption{Logical operations on $\mathrm{QRM}(4)$ which are implemented by the gates $\g{U_{PH}}(e)$, $\g{U_{PH}}(Q(1,2))$, $\g{U_{PH}}(Q(3,4))$, and $\g{U_{PH}}(Q(1,2)Q(3,4))$.}
	\label{fig:ex_phase-type_logical}%
\end{figure}

Note that none of the fold-transversal gates constructed from the automorphisms shown in \cref{ex:logical_from_fold}
directly give an addressable gate. Nevertheless, it is possible to construct an addressable gate from a product of certain phase-type fold-transversal gates; for example, $\g{U_{PH}}(e)\g{U_{PH}}(Q(1,2))=\logg{C_{00}Z}(2,3)$. This motivates our development of the following theorem.

\begin{theorem}[Logical operation of a product of phase-type fold-transversal gates] \label{thm:fold_product}
    Let $m$ be a positive even number, $\mathrm{QRM}(m)$ be the quantum Reed--Muller code defined in \cref{def:QRM}, and $K$ be a set of ordered pairs defined in \cref{def:K} of size $0 \leq |K| \leq m/2$. The following logical gates are implemented by the action of $\prod_{L \subseteq K}\g{U_{PH}}\left(Q(L)\right)$:
    \begin{enumerate}
        \item \label[statement]{statement:fold_product_cz}
        If $|K|\leq m/2-2$, then $\logg{C_{11}Z}(B,B')$ is applied to any unordered pair of logical qubits $\{B,B'\}$ satisfying $F_1(K) \subseteq B$, $F_2(K)\cap B = \emptyset$, and
        \begin{equation}
            B' = \left(B^\up{c} \cup F_1(K)\right)\setminus F_2(K).
        \end{equation}
        \item \label[statement]{statement:fold_product_czxx}
        If $|K|= m/2-1$, then $\logg{C_{00}Z}(B,B')$ is applied to the logical qubits $B,B'$ that satisfy $F_1(K) \subseteq B$, $F_2(K)\cap B = \emptyset$, and
        \begin{equation}
            B' = \left(B^\up{c} \cup F_1(K)\right)\setminus F_2(K).
        \end{equation}
        \item \label[statement]{statement:fold_product_s}
        If $|K|= m/2$, then $\logg{S}(F_1(K))$ is implemented if $m/2$ is even (or $\logg{S}^{\dagger}(F_1(K))$ is implemented if $m/2$ is odd).
    \end{enumerate}
\end{theorem}

A proof of \cref{thm:fold_product} is provided in \cref{sec:proof_of_thms}.

\begin{example} \label{ex:fold_product}
    Consider $\mathrm{QRM}(4)$ with the canonical logical qubit indices, i.e., $B_1=\{1,2\}$, $B_2=\{1,3\}$, $B_3=\{1,4\}$, $B_4=\{3,4\}$, $B_5=\{2,4\}$, and $B_6=\{2,3\}$. 
    Let us first observe the logical operation of $\prod_{L\subseteq K}\g{U_{PH}}(Q(L))$ with $K=\{(1,2)\}$ according to \cref{thm:fold_product}. Explicitly, the gate is $\g{U_{PH}}(e)\g{U_{PH}}(Q(1,2))$. 
    Since $|K|=1=m/2-1$, Statement 2 of \cref{thm:fold_product} is applicable. Here we have that $F_1(K)=\{1\}$ and $F_2(K)=\{2\}$. There is only one unordered pair $\{B,B'\}$ that satisfies the conditions $\{1\} \subseteq B$, $\{2\} \cap B = \emptyset$, and $B' =  \left(B^\up{c}\cup \{1\} \right)\setminus \{2\}$, which is the pair $\{B_2,B_3\}$. Therefore, Statement 2 of \cref{thm:fold_product} implies that $\g{U_{PH}}(e)\g{U_{PH}}(Q(1,2))$ implements $\logg{C_{00}Z}(2,3)$. This matches our previous observation.

    Extending the application of \cref{thm:fold_product} to $\prod_{L\subseteq K}\g{U_{PH}}(Q(L))$ with $K=\emptyset$, $K=\{(3,4)\}$, and $K=\{(1,2),(3,4)\}$, we find the following logical operations.
    \begin{equation}
        \begin{tabular}{| c | c | c |}
            \hline
             $K$ & $\prod_{L\subseteq K}\g{U_{PH}}(Q(L))$ & \textrm{Logical operation} \\
             \hline
             $\emptyset$ & $\g{U_{PH}}(e)$ & $\logg{C_{11}Z}(1,4)\logg{C_{11}Z}(2,5)\logg{C_{11}Z}(3,6)$\\
             \hline
             $\{(1,2)\}$ & $\g{U_{PH}}(e)\g{U_{PH}}(Q(1,2))$ & $\logg{C_{00}Z}(2,3)$\\
             \hline
             $\{(3,4)\}$ & $\g{U_{PH}}(e)\g{U_{PH}}(Q(3,4))$ & $\logg{C_{00}Z}(2,6)$\\
             \hline
             $\{(1,2),(3,4)\}$ & $\g{U_{PH}}(e)\g{U_{PH}}(Q(1,2))\g{U_{PH}}(Q(3,4))\g{U_{PH}}(Q(1,2)Q(3,4))$ & $\logg{S}(2)$ \\
             \hline
        \end{tabular} \nonumber
    \end{equation}
    These logical operations can be illustrated by the circuit in \cref{fig:ex_phase-type_logical_addressable_1}, with their corresponding physical operations shown in \cref{fig:ex_phase-type_logical_addressable_2}.
\end{example}

With \cref{thm:fold_product}, it is possible to construct an addressable $\logg{S}$ gate on any logical qubit and an addressable $\logg{C_{00}Z}$ gate on certain pairs of logical qubits by choosing $K$ that satisfies certain conditions on $F_1(K)$ and $F_2(K)$. The construction of such gates can be summarized in the following corollary.

\begin{corollary} [addressable $\logg{S}$ and addressable $\logg{C_{00}Z}$ gates] \label{cor:addr_S_CZ}
    Let $m$ be a positive even number, $\mathrm{QRM}(m)$ be the quantum Reed--Muller code defined in \cref{def:QRM}, and $K$ be a set of ordered pairs defined in \cref{def:K}.
    \begin{enumerate}
        \item \label[statement]{statement:cor_s}
        For any logical qubit index $B$, the addressable gate $\logg{S}(B)$ when $m/2$ is even (or $\logg{S}^{\dagger}(B)$ when $m/2$ is odd) can be implemented by choosing any $K$ such that $F_1(K)=B$ then applying \cref{statement:fold_product_s} of \cref{thm:fold_product}. The gate $\logg{S}^\dagger(B)$ when $m/2$ is even (or $\logg{S}(B)$ when $m/2$ is odd) can be implemented by the inverse operation of $\logg{S}(B)$ (or $\logg{S}^\dagger(B)$). The implementation of $\logg{S}(B)$ or $\logg{S}^{\dagger}(B)$ has depth $2^{m/2} = \sqrt{n}$.
        \item \label[statement]{statement:cor_czxx}
        For any two logical qubit indices $B,B'$ that differ by exactly one basis vector (i.e., $|B \cap B'|=m/2-1$), the addressable gate $\logg{C_{00}Z}(B,B')$ can be implemented by choosing $K$ such that $F_1(K)=B\cap B'$ and $F_2(K)=[m]\setminus(B\cup B')$ then applying \cref{statement:fold_product_czxx} of \cref{thm:fold_product}. The implementation of $\logg{C_{00}Z}(B,B')$ has depth $2^{m/2-1} = \sqrt{n}/2$.
    \end{enumerate}
\end{corollary}

\begin{figure}[!htbp]
	\centering
    \begin{tikzpicture}
\begin{yquant}
nobit n0;
nobit n1;
qubit {$B_1$} q1;
qubit {$B_2$} q2;
qubit {$B_3$} q3;
qubit {$B_4$} q4;
qubit {$B_5$} q5;
qubit {$B_6$} q6;
hspace {1mm} -;

text {\hspace{2.8cm} $\prod_{L\subseteq K}\g{U_{PH}}(Q(L))$} n0;

text {$K=\emptyset$} n1;
zz (q1, q4);
zz (q2, q5);
zz (q3, q6);
barrier n1-q6;

text {$K=\{(1,2)\}$} n1;
hspace {8mm} q1-q6;
cnot ~ q2, q3;
barrier n1-q6;

text {$K=\{(3,4)\}$} n1;
hspace {8mm} q1-q6;
cnot ~ q2, q6;
barrier n1-q6;

text {$K=\{(1,2),(3,4)\}$} n1;
hspace {12mm} q1-q6;
box {$\logg{S}$} q2;


\end{yquant}
\end{tikzpicture}
	\caption{Logical operations on $\mathrm{QRM}(4)$ which are implemented by the gates $\prod_{L\subseteq K}\g{U_{PH}}(Q(L))$ with $K=\emptyset$, $K=\{(1,2)\}$, $K=\{(3,4)\}$, and $K=\{(1,2),(3,4)\}$.}
	\label{fig:ex_phase-type_logical_addressable_1}%
\end{figure}

\begin{figure}[!htbp]
	\centering
    \begin{tikzpicture}
\begin{yquant}
nobit n0;
nobit n1;
qubit {0000} q0;
qubit {0001} q1;
qubit {0010} q2;
qubit {0011} q3;
qubit {0100} q4;
qubit {0101} q5;
qubit {0110} q6;
qubit {0111} q7;
qubit {1000} q8;
qubit {1001} q9;
qubit {1010} q10;
qubit {1011} q11;
qubit {1100} q12;
qubit {1101} q13;
qubit {1110} q14;
qubit {1111} q15;
hspace {1mm} -;

text {\hspace{2.8cm} $\prod_{L\subseteq K}\g{U_{PH}}(Q(L))$} n0;

text {$K=\emptyset$} n1;
hspace {2mm} q0-q15;
box {$\g{S}$} q0-q15;
barrier n1-q15;

text {$K=\{(1,2)\}$} n1;
hspace {5mm} q0-q15;
box {$\g{S}$} q0-q15;
hspace {0.8mm} q2, q3, q6, q7, q10, q11, q14, q15;
zz (q2, q3);
zz (q6, q7);
zz (q10, q11);
zz (q14, q15);
box {$\g{S}$} q0, q1, q4, q5, q8, q9, q12, q13;
barrier n1-q15;

text {$K=\{(3,4)\}$} n1;
hspace {1mm} q0-q15;
box {$\g{S}$} q0-q15;
hspace {1mm} q0-q15;
hspace {3.5mm} q0-q7;
zz (q8, q12);
zz (q9, q13);
zz (q10, q14);
zz (q11, q15);
box {$\g{S}$} q0-q7;
barrier n1-q15;

hspace {3mm} n1;
text {$K=\{(1,2),(3,4)\}$} n1;
hspace {1mm} q0-q15;
box {$\g{S}$} q0-q15;
hspace {1mm} q0-q15;
hspace {0.8mm} q2, q3, q6, q7, q10, q11, q14, q15;
zz (q2, q3);
zz (q6, q7);
zz (q10, q11);
zz (q14, q15);
box {$\g{S}$} q0, q1, q4, q5, q8, q9, q12, q13;
hspace {1mm} q0-q15;
hspace {3.5mm} q0-q7;
zz (q8, q12);
zz (q9, q13);
zz (q10, q14);
zz (q11, q15);
box {$\g{S}$} q0-q7;
hspace {1mm} q0-q15;
hspace {3.5mm} q0-q7;
hspace {0.8mm} q2, q3, q6, q7;
zz (q2, q3), (q6, q7);
zz (q8, q12);
zz (q9, q13);
zz (q10, q15);
zz (q11, q14);
box {$\g{S}$} q0, q1, q4, q5;
hspace {1mm} q0-q15;

\end{yquant}
\end{tikzpicture}
	\caption{Physical operations corresponding to the gates $\prod_{L\subseteq K}\g{U_{PH}}(Q(L))$ with $K=\emptyset$, $K=\{(1,2)\}$, $K=\{(3,4)\}$, and $K=\{(1,2),(3,4)\}$ on $\mathrm{QRM}(4)$.}
	\label{fig:ex_phase-type_logical_addressable_2}%
\end{figure}

\begin{proof}
    (1) A set of ordered pairs $K$ such that $F_1(K)=B$ always exists. By \cref{statement:fold_product_s} of \cref{thm:fold_product}, $\prod_{L \subseteq K}\g{U_{PH}}\left(Q(L)\right)$ implements $\logg{S}(F_1(K))$ if $m/2$ is even (or $\logg{S}^{\dagger}(F_1(K))$ if $m/2$ is odd). The depth of $\prod_{L \subseteq K}\g{U_{PH}}\left(Q(L)\right)$ is $2^{|K|} = 2^{m/2} = \sqrt{n}$.
    
    (2) Suppose that $B$ and $B'$ are logical qubit indices that differ by exactly one basis vector, i.e., $|B \cap B'| = m/2-1$. We can write $B = (B\cap B')\cup\{i_1\}$ and $B' = (B\cap B')\cup\{i_2\}$ for some basis vector indices $i_1,i_2$. Let $K$ be a set of ordered pairs such that $F_1(K)=B\cap B'$ and $F_2(K)=[m]\setminus(B\cup B')$ (so that $F_1(K)\cup F_2(K)$ does not contain $i_1,i_2$). Such $K$ always exists since $|F_2(K)|=m-|B\cup B'|=m-(|B|+|B'|-|B\cap B'|)=|B\cap B'| = |F_1(K)|$. Note that $|K|=m/2-1$. By \cref{statement:fold_product_czxx} of \cref{thm:fold_product}, $\prod_{L \subseteq K}\g{U_{PH}}\left(Q(L)\right)$ implements $\logg{C_{00}Z}(D,D')$ to any unordered pair of qubits $\{D,D'\}$ satisfying $F_1(K) \subseteq D$, $F_2(K)\cap D = \emptyset$, and $D' = \left(D^\up{c} \cup F_1(K)\right)\setminus F_2(K)$. Choices of $D$ satisfying $F_1(K) \subseteq D$ include (i) $D = F_1(K)\cup\{i_1\}=B$, (ii) $D = F_1(K)\cup\{i_2\}=B'$, and (iii) $D=F_1(K)\cup\{j\}$ for some $j \in F_2(K)$. However, (iii) does not satisfy $F_2(K)\cap D = \emptyset$. Also the choices of $D$ in (i) and (ii), which are $B$ and $B'$ respectively, are related by $B' = \left(B^\up{c} \cup F_1(K)\right)\setminus F_2(K)$ (as well as, $B = \left((B')^\up{c} \cup F_1(K)\right)\setminus F_2(K)$). Therefore, by choosing $K$ as previously described, the product of all $\g{U_{PH}}\left(Q(L)\right)$ such that $L \subseteq K$ implements the addressable gate $\logg{C_{00}Z}(B,B')$. The depth of $\prod_{L \subseteq K}\g{U_{PH}}\left(Q(L)\right)$ is $2^{|K|} = 2^{m/2-1} = \sqrt{n}/2$.
\end{proof}
Note that for each of the statements in \cref{cor:addr_S_CZ}, both $F_1(K)$ and $F_2(K)$ are unique. However, the choice of $K$ that gives such $F_1(K)$ and $F_2(K)$ might not be unique.

\begin{example}
    Consider $\mathrm{QRM}(4)$ with the canonical logical qubit indices, i.e., $B_1=\{1,2\}$, $B_2=\{1,3\}$, $B_3=\{1,4\}$, $B_4=\{3,4\}$, $B_5=\{2,4\}$, and $B_6=\{2,3\}$. The logical operation $\logg{S}(2)$ can be implemented by applying \cref{statement:cor_s} of \cref{cor:addr_S_CZ} to either $K=\{(1,2),(3,4)\}$ or $K=\{(1,4),(3,2)\}$. Meanwhile, the logical operation $\logg{C_{00}Z}(2,3)$ can be implemented by applying \cref{statement:cor_czxx} of \cref{cor:addr_S_CZ} to $K=\{(1,2)\}$; in fact, this is the only $K$ that can lead to $\logg{C_{00}Z}(2,3)$. 
\end{example}

We point out that while all addressable $\logg{S}$ gates can be obtained through \cref{cor:addr_S_CZ}, not all addressable $\logg{C_{00}Z}$ gates can be obtained from this construction. Fortunately, the missing addressable $\logg{C_{00}Z}$ gates, as well as other Clifford gates, can be achieved if addressable $\logg{H}$ and addressable $\logg{SW}$ gates can be constructed. The construction of such gates will be discussed in the next section.

\subsubsection{Addressable $\logg{H}$, $\logg{SW}$, and $\logg{C_{00}Z}$ gates} 

So far, we have constructed addressable $\logg{S}$ gates on any logical qubit and addressable $\logg{C_{00}Z}$ gates on certain pairs of logical qubits from sequences of fold-transversal gates. In this section, we will combine these gates with a transversal Hadamard gate $\g{H}^{\otimes n}$ to construct addressable $\logg{H}$ and $\logg{SW}$ gates, as well as addressable $\logg{C_{00}Z}$ gates that cannot be attained by \cref{cor:addr_S_CZ}.

We start by considering the logical operation of $\g{H}^{\otimes n}$ and construct an addressable $\logg{H}$ gate from $\g{H}^{\otimes n}$ and addressable $\logg{S}$ gates.
\begin{theorem}[addressable $\logg{H}$ gates]\label{thm:addr_H}
    Let $m$ be a positive even number, and $\mathrm{QRM}(m)$ be the \codepar{n,k,d} quantum Reed--Muller code defined in \cref{def:QRM}. Then,
    \begin{enumerate}
        \item $\g{H}^{\otimes n}$ applies $\logg{H}^{\otimes k}$ followed by $\prod_{i=1}^{k/2}\logg{SW}(i,i+k/2)$;
        \item for any logical qubit $B$, the addressable gate $\logg{H}(B)$ can be constructed from $\g{H}^{\otimes n}$, $\logg{S}(B)$, and $\logg{S}(B^\up{c})$. The implementation of $\logg{H}(B)$ has depth $O(\sqrt{n})$.
    \end{enumerate}
\end{theorem}
\begin{proof}
    (1) From our choice of logical Pauli operators defined in \cref{def:log_def,def:canonical_indices}, for any logical qubit $B$, $\g{H}^{\otimes n}$ maps $\logg{X}(B)=\g{X}(\mathbf{v}_B)$ to $\logg{Z}(B^\up{c})=\g{Z}(\mathbf{v}_B)$. This is equivalent to applying $\logg{H}^{\otimes k}$ followed by $\prod_{i=1}^{k/2}\logg{SW}(i,i+k/2)$.

    (2) Using the fact that $\g{HSHSHS}=\g{I}$ up to some global phase, for any logical qubit $B$, $\logg{H}(B)$ can be constructed from the following logical operations. 
    \begin{equation}
        \includegraphics[scale=1.15]{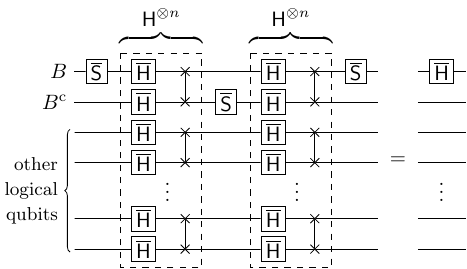} \nonumber
    \end{equation}
    Note that the product of $\logg{H}^{\otimes k}$ and $\prod_{i=1}^{k/2}\logg{SW}(i,i+k/2)$ can be constructed by (1), and $\logg{S}(B)$ and $\logg{S}(B^\up{c})$ can be constructed by \cref{cor:addr_S_CZ}. Since each addressable $\logg{S}$ gate has depth $\sqrt{n}$, the overall circuit for $\logg{H}(B)$ has depth $O(\sqrt{n})$.
\end{proof}

Next, we show how addressable $\logg{SW}$ gates on certain pairs of logical qubits can be constructed from addressable $\logg{C_{00}Z}$ and addressable $\logg{H}$ gates which are obtainable from \cref{cor:addr_S_CZ,thm:addr_H}. Such addressable $\logg{SW}$ gates can be combined to make an addressable $\logg{SW}$ gate on any pair of logical qubits. We can also construct addressable $\logg{C_{00}Z}$ gates which are not obtainable from \cref{cor:addr_S_CZ} using similar ideas.

\begin{theorem}[addressable $\logg{SW}$ and $\logg{C_{00}Z}$ gates]\label{thm:addr_SW_CZ}
    Let $m$ be a positive even number, and $\mathrm{QRM}(m)$ be the \codepar{n,k,d} quantum Reed--Muller code defined in \cref{def:QRM}. Then,
    \begin{enumerate}
        \item for any two logical qubit indices $B,B'$ that differ by exactly one basis vector (i.e., $|B\cap B'|=m/2-1$), the addressable gate $\logg{SW}(B,B')$ can be constructed from $\logg{C_{00}Z}(B,B')$, $\logg{H}(B)$, and $\logg{H}(B')$. The implementation of such a gate has depth $O(\sqrt{n})$.
        \item for any two logical qubits $B,B'$, the addressable gate $\logg{SW}(B,B')$ can be constructed from a product of addressable $\logg{SW}$ gates from (1). The implementation of any $\logg{SW}(B,B')$ has depth $O(\sqrt{n}\log{n})$.
        \item for any two logical qubits $B,B'$, the addressable gate $\logg{C_{00}Z}(B,B')$ can be constructed from an addressable $\logg{C_{00}Z}$ gate from \cref{cor:addr_S_CZ} and a product of addressable $\logg{SW}$ gates from (1). The implementation of any $\logg{C_{00}Z}(B,B')$ has depth $O(\sqrt{n}\log{n})$.
    \end{enumerate}
\end{theorem}
\begin{proof} 
    (1) For any two logical qubits $B,B'$ such that $|B\cap B'|=m/2-1$, $\logg{C_{00}Z}(B,B')$ can be constructed by \cref{cor:addr_S_CZ}, and $\logg{H}(B)$, and $\logg{H}(B')$ can be constructed from \cref{thm:addr_H}. The gate $\logg{SW}(B,B')$ can be constructed from $\logg{C_{00}Z}(B,B')$, $\logg{H}(B)$, and $\logg{H}(B')$ as follows.
    \begin{equation}
        \begin{tikzpicture}
\begin{yquantgroup}

\registers{
qubit {} b;
qubit {} bprime;
}

\circuit{
init {$B$} b;
init {$B'$} bprime;

hspace {1mm} b;
cnot ~ b, bprime;
box {$\logg{H}$} b, bprime;
cnot ~ b, bprime;
box {$\logg{H}$} b, bprime;
cnot ~ b, bprime;
box {$\logg{H}$} b, bprime;
hspace {1mm} b;

}

\equals

\circuit{
swap (b, bprime);
}

\end{yquantgroup}
\end{tikzpicture} \nonumber
    \end{equation}
    Since each addressable $\logg{H}$ gate has depth $O(\sqrt{n})$ and each addressable $\logg{C_{00}Z}$ gate has depth $\sqrt{n}/2$, the overall circuit for $\logg{SW}(B,B')$ such that $|B\cap B'|=m/2-1$ has depth $O(\sqrt{n})$.
    
    (2) Suppose that $B$ and $B'$ differ by $c$ vectors (i.e., $|B \cap B'|=m/2-c$). There exist logical qubits $B^{(1)},B^{(2)},\dots,B^{(c-2)},B^{(c-1)}$ such that $|B \cap B^{(1)}| = |B^{(1)} \cap B^{(2)}|=\dots=|B^{(c-2)} \cap B^{(c-1)}|=|B^{(c-1)} \cap B'|=m/2-1$. Addressable swap gates $\logg{SW}(B,B^{(1)})$, $\logg{SW}(B^{(1)},B^{(2)})$, \dots, $\logg{SW}(B^{(c-2)},B^{(c-1)})$, and $\logg{SW}(B^{(c-1)},B')$ can be constructed from (1).
    Using these gates, $\logg{SW}(B,B')$ can be constructed as follows.
    \begin{equation}
        \begin{tikzpicture}
\begin{yquantgroup}

\registers{
qubit {} b;
qubit {} b1;
qubit {} b2;
nobit ellipsis;
qubit {} bc2;
qubit {} bc1;
qubit {} bprime;
}

\circuit{
init {$B$} b;
init {$B_1$} b1;
init {$B_2$} b2;
init {$B_{c-2}$} bc2;
init {$B_{c-1}$} bc1;
init {$B'$} bprime;

swap (b, b1);
swap (b1, b2);
align -;
dots ellipsis;
align -;
swap (bc2, bc1);
swap (bc1, bprime);
swap (bc2, bc1);
align -;
dots ellipsis;
align -;
swap (b1, b2);
swap (b, b1);
}

\equals

\circuit{
dots ellipsis;
swap (b, bprime);
}

\end{yquantgroup}
\end{tikzpicture} \nonumber
    \end{equation}
    The overall circuit for $\logg{SW}(B,B')$ such that $|B \cap B'|=m/2-c$ consists of $2c-1$ addressable swap gates from (1). In the worst case, $c= m/2 = (\log_2{n})/2$. Therefore, the circuit for any $\logg{SW}(B,B')$ has depth $O(\sqrt{n}\log{n})$.
    
    (3) Suppose that $B$ and $B'$ differ by $c$ vectors ($|B \cap B'|=m/2-c$). We can find  $B^{(1)}$, $B^{(2)}$, \dots, $B^{(c-2)}$, $B^{(c-1)}$ similar to the previous case and construct addressable swap gates $\logg{SW}(B,B^{(1)})$, $\logg{SW}(B^{(1)},B^{(2)})$, \dots, $\logg{SW}(B^{(c-2)},B^{(c-1)})$ from (1). We can also construct $\logg{C_{00}Z}(B^{(c-1)},B')$ from \cref{cor:addr_S_CZ}. Using these gates, $\logg{C_{00}Z}(B,B')$ can be constructed as follows.
    \begin{equation}
        \begin{tikzpicture}
\begin{yquantgroup}

\registers{
qubit {} b;
qubit {} b1;
qubit {} b2;
nobit ellipsis;
qubit {} bc2;
qubit {} bc1;
qubit {} bprime;
}

\circuit{
init {$B$} b;
init {$B_1$} b1;
init {$B_2$} b2;
init {$B_{c-2}$} bc2;
init {$B_{c-1}$} bc1;
init {$B'$} bprime;

swap (b, b1);
swap (b1, b2);
align -;
dots ellipsis;
align -;
swap (bc2, bc1);
cnot ~bc1, bprime;
swap (bc2, bc1);
align -;
dots ellipsis;
align -;
swap (b1, b2);
swap (b, b1);
}

\equals

\circuit{
dots ellipsis;
cnot ~b, bprime;
}

\end{yquantgroup}
\end{tikzpicture} \nonumber
    \end{equation}
    Similar to (2), the circuit for any $\logg{C_{00}Z}(B,B')$ has depth $O(\sqrt{n}\log{n})$.
\end{proof}

\subsubsection{The full logical Clifford group} \label{subsubsec:full_Clifford}

It is possible to show that the full logical Clifford group of any $\mathrm{QRM}(m)$ can be generated by addressable $\logg{S}$, $\logg{H}$, and $\logg{C_{00}Z}$ gates, whose constructions from transversal and fold-transversal gates were described in the previous sections. This implies that any logical Clifford gate in the full logical Clifford group can be implemented by a sequence of transversal and fold-transversal gates. The result can be summarized in the following theorem.

\begin{theorem}[The full logical Clifford group of quantum Reed--Muller codes] \label{thm:full_Clifford}
    Let $m$ be a positive even number, $\mathrm{QRM}(m)$ be the \codepar{n,k,d} quantum Reed--Muller code defined in \cref{def:QRM}, and $K$ be a set of ordered pairs defined in \cref{def:K}. Then for any $m$, the full logical Clifford group $\overline{C}_k$ of $\mathrm{QRM}(m)$ can be generated by $\g{H}^{\otimes n}$ and $\g{U_{PH}}\left(Q(K)\right)$ for all $K$ of size $0 \leq |K| \leq m/2$; i.e.,
    \begin{equation}
        \overline{C}_k = \left\langle \g{H}^{\otimes n},\g{U_{PH}}\left(Q(K)\right)\right\rangle_{0 \leq |K| \leq m/2}.
    \end{equation}
\end{theorem}
\begin{proof}
    For any $\mathrm{QRM}(m)$, addressable $\logg{S}$, $\logg{H}$, and $\logg{C_{00}Z}$ gates on any logical qubit (or any pair of logical qubits) can be constructed by \cref{cor:addr_S_CZ,thm:addr_H,thm:addr_SW_CZ}. Observe that (i) $\logg{C_{00}Z} = (\logg{C_{11}Z})(\logg{Z}\otimes \logg{Z})$ up to some global phase, (ii) $\logg{Z} = \logg{S}^2$, and (iii) the full logical Clifford group can be generated by $\logg{H}$, $\logg{S}$, and $\logg{C_{11}Z}$ gates on any logical qubit (or any pair of logical qubits). Thus, for any $\mathrm{QRM}(m)$, the addressable gates from \cref{cor:addr_S_CZ,thm:addr_H,thm:addr_SW_CZ} are sufficient to generate the full logical Clifford group $\overline{C}_k$. Since any addressable gate from \cref{cor:addr_S_CZ}, \cref{thm:addr_H}, or \cref{thm:addr_SW_CZ} can be implemented by a sequence of $\g{H}^{\otimes n}$ and $\g{U_{PH}}\left(Q(K)\right)$ where $0 \leq |K| \leq m/2$, we have that $\overline{C}_k = \left\langle \g{H}^{\otimes n},\g{U_{PH}}\left(Q(K)\right)\right\rangle_{0 \leq |K| \leq m/2}$.
\end{proof}

It should be noted that the logical operations of the gate sequences proposed in \cref{cor:addr_S_CZ,thm:addr_H,thm:addr_SW_CZ} are specific to the symplectic basis determined by \cref{def:log_def,def:canonical_indices}. Nevertheless, the result in \cref{thm:full_Clifford} is \emph{independent} of the choice of symplectic basis; any symplectic basis can be converted to another by the conjugation of some logical Clifford operation (see \cite{TTF25} for example), and our construction can generate \emph{any} logical Clifford gate in the full logical Clifford group. Therefore, even if the gate sequences proposed in \cref{cor:addr_S_CZ,thm:addr_H,thm:addr_SW_CZ} provide different logical operations in a different symplectic basis, $\{ \g{H}^{\otimes n},\g{U_{PH}}\left(Q(K)\right)\}_{0 \leq |K| \leq m/2}$ still generates the full logical Clifford group.

\subsection{Other logical gates} \label{subsec:other_gates}

Besides $\g{U_{PH}}(Q(K))$, other fold-transversal gates such as $\g{U_{SW}}(P(i,j))$, $\g{U_{PH}}(P(i,j))$, and $\g{U_{SW}}(Q(i,j))$ for any $i,j \in [m]$ are also valid logical operators since they preserve the stabilizer group (as shown by \cref{thm:PnQ_preserve_stb}). In this section, we consider some examples of $\g{U_{SW}}(P(i,j))$, $\g{U_{PH}}(P(i,j))$, and $\g{U_{SW}}(Q(i,j))$ for the code $\mathrm{QRM}(4)$ and observe their logical operations.

We start by considering the permutation $P(1,2)$ which swaps $\mathbf{v}_1$ and $\mathbf{v}_2$ and maps other $\mathbf{v}_i$ to itself. Under $P(1,2)$, the vectors of the form $\mathbf{v}_i\mathbf{v}_j$ where $i,j \in \{1,2,3,4\}$ and $i \neq j$ are mapped as follows.
\begin{equation}
    \begin{tabular}{r c l c r c l}
        $\mathbf{v}_1\mathbf{v}_2$ &  $\mapsto$ & $\mathbf{v}_1\mathbf{v}_2$, & $\quad$ & 
        $\mathbf{v}_3\mathbf{v}_4$ & $\mapsto$ & $\mathbf{v}_3\mathbf{v}_4$, \\
        $\mathbf{v}_1\mathbf{v}_3$ &  $\mapsto$ & $\mathbf{v}_2\mathbf{v}_3$, & $\quad$ & 
        $\mathbf{v}_2\mathbf{v}_4$ & $\mapsto$ & $\mathbf{v}_1\mathbf{v}_4$, \\
        $\mathbf{v}_1\mathbf{v}_4$ &  $\mapsto$ & $\mathbf{v}_2\mathbf{v}_4$, & $\quad$ & 
        $\mathbf{v}_2\mathbf{v}_3$ & $\mapsto$ & $\mathbf{v}_1\mathbf{v}_3$,
    \end{tabular} \nonumber
\end{equation}
Recall that for the code $\mathrm{QRM}(4)$, the canonical logical qubit indices are $B_1=\{1,2\}$, $B_2=\{1,3\}$, $B_3=\{1,4\}$, $B_4=\{3,4\}$, $B_5=\{2,4\}$, and $B_6=\{2,3\}$, and the logical Pauli operators are $\logg{X}(i)=\logg{X}(B_i)=\g{X}(\mathbf{v}_{B_i})$ and $\logg{Z}(i)=\logg{Z}(B_i)=\g{Z}(\mathbf{v}_{B_i^\up{c}})$. By \cref{prop:SW_PH_transform}, the fold-transversal gates $\g{U_{SW}}(P(1,2))$ and $\g{U_{PH}}(P(1,2))$ map the logical Pauli operators as follows.
\begin{equation}
    \begin{tabular}{r r c l r c l c r c l r c l}
        $\g{U_{SW}}(P(1,2))$: & $\logg{X}(1)$ & $\mapsto$ & $\logg{X}(1)$,  & $\logg{Z}(1)$ & $\mapsto$ & $\logg{Z}(1)$, & \quad & $\logg{X}(4)$ & $\mapsto$ & $\logg{X}(4)$,  & $\logg{Z}(4)$ & $\mapsto$ & $\logg{Z}(4)$,\\
        & $\logg{X}(2)$ & $\mapsto$ & $\logg{X}(6)$,  & $\logg{Z}(2)$ & $\mapsto$ & $\logg{Z}(6)$, & \quad & $\logg{X}(5)$ & $\mapsto$ & $\logg{X}(3)$,  & $\logg{Z}(5)$ & $\mapsto$ & $\logg{Z}(3)$,\\
        & $\logg{X}(3)$ & $\mapsto$ & $\logg{X}(5)$,  & $\logg{Z}(3)$ & $\mapsto$ & $\logg{Z}(5)$, & \quad & $\logg{X}(6)$ & $\mapsto$ & $\logg{X}(2)$,  & $\logg{Z}(6)$ & $\mapsto$ & $\logg{Z}(2)$,
    \end{tabular} \nonumber
\end{equation}
\begin{equation}
    \begin{tabular}{r r c r r c l c r c r r c l}
        $\g{U_{PH}}(P(1,2))$: & $\logg{X}(1)$ & $\mapsto$ & $\logg{X}(1)\logg{Z}(4)$,  & $\logg{Z}(1)$ & $\mapsto$ & $\logg{Z}(1)$, & \quad & $\logg{X}(4)$ & $\mapsto$ & $\logg{X}(4)\logg{Z}(1)$,  & $\logg{Z}(4)$ & $\mapsto$ & $\logg{Z}(4)$,\\
        & $\logg{X}(2)$ & $\mapsto$ & $-\logg{X}(2)\logg{Z}(3)$,  & $\logg{Z}(2)$ & $\mapsto$ & $\logg{Z}(2)$, & \quad & $\logg{X}(5)$ & $\mapsto$ & $-\logg{X}(5)\logg{Z}(6)$,  & $\logg{Z}(5)$ & $\mapsto$ & $\logg{Z}(5)$,\\
        & $\logg{X}(3)$ & $\mapsto$ & $-\logg{X}(3)\logg{Z}(2)$,  & $\logg{Z}(3)$ & $\mapsto$ & $\logg{Z}(3)$, & \quad & $\logg{X}(6)$ & $\mapsto$ & $-\logg{X}(6)\logg{Z}(5)$,  & $\logg{Z}(6)$ & $\mapsto$ & $\logg{Z}(6)$.
    \end{tabular} \nonumber
\end{equation}
Therefore, $\g{U_{SW}}(P(1,2)) = \logg{SW}(2,6)\logg{SW}(3,5)$ and $\g{U_{PH}}(P(1,2))=\logg{C_{11}Z}(1,4)\logg{C_{00}Z}(2,3)\logg{C_{00}Z}(5,6)$. Similarly, for any $i,j \in \{1,2,3,4\}$ where $i \neq j$, $\g{U_{SW}}(P(i,j))$ acts as a product of logical swap gates on some pairs of logical qubits, and $\g{U_{PH}}(P(i,j))$ acts as a product of logical controlled-$\g{Z}$ gates on some pairs of logical qubits.

Next, we consider the permutation $Q(1,2)$ which maps $\mathbf{v}_1$ to $\mathbf{v}_1+\mathbf{v}_2$ and maps other $\mathbf{v}_i$ to itself. Under $Q(1,2)$, the vectors of the form $\mathbf{v}_i\mathbf{v}_j$ where $i,j \in \{1,2,3,4\}$ and $i \neq j$ are mapped as follows.
\begin{equation}
    \begin{tabular}{r c l c r c l}
        $\mathbf{v}_1\mathbf{v}_2$ &  $\mapsto$ & $\mathbf{v}_1\mathbf{v}_2+\mathbf{v}_2$, & $\quad$ & 
        $\mathbf{v}_3\mathbf{v}_4$ & $\mapsto$ & $\mathbf{v}_3\mathbf{v}_4$, \\
        $\mathbf{v}_1\mathbf{v}_3$ &  $\mapsto$ & $\mathbf{v}_1\mathbf{v}_3+\mathbf{v}_2\mathbf{v}_3$, & $\quad$ & 
        $\mathbf{v}_2\mathbf{v}_4$ & $\mapsto$ & $\mathbf{v}_2\mathbf{v}_4$, \\
        $\mathbf{v}_1\mathbf{v}_4$ &  $\mapsto$ & $\mathbf{v}_1\mathbf{v}_4+\mathbf{v}_2\mathbf{v}_4$, & $\quad$ & 
        $\mathbf{v}_2\mathbf{v}_3$ & $\mapsto$ & $\mathbf{v}_2\mathbf{v}_3$.
    \end{tabular} \nonumber
\end{equation}
By \cref{prop:SW_PH_transform}, the fold-transversal gate $\g{U_{SW}}(Q(1,2))$ maps the logical Pauli operators as follows (up to a multiplication of some stabilizer).
\begin{equation}
    \begin{tabular}{r r c l r c l c r c l r c l}
        $\g{U_{SW}}(Q(1,2))$: & $\logg{X}(1)$ & $\mapsto$ & $\logg{X}(1)$,  & $\logg{Z}(1)$ & $\mapsto$ & $\logg{Z}(1)$, & \quad & $\logg{X}(4)$ & $\mapsto$ & $\logg{X}(4)$,  & $\logg{Z}(4)$ & $\mapsto$ & $\logg{Z}(4)$,\\
        & $\logg{X}(2)$ & $\mapsto$ & $\logg{X}(2)\logg{X}(6)$,  & $\logg{Z}(2)$ & $\mapsto$ & $\logg{Z}(2)$, & \quad & $\logg{X}(5)$ & $\mapsto$ & $\logg{X}(5)$,  & $\logg{Z}(5)$ & $\mapsto$ & $\logg{Z}(5)\logg{Z}(3)$,\\
        & $\logg{X}(3)$ & $\mapsto$ & $\logg{X}(3)\logg{X}(5)$,  & $\logg{Z}(3)$ & $\mapsto$ & $\logg{Z}(3)$, & \quad & $\logg{X}(6)$ & $\mapsto$ & $\logg{X}(6)$,  & $\logg{Z}(6)$ & $\mapsto$ & $\logg{Z}(6)\logg{Z}(2)$,
    \end{tabular} \nonumber
\end{equation}
Therefore, $\g{U_{SW}}(Q(1,2)) = \logg{CX}(2,6)\logg{CX}(3,5)$. Similarly, for any $i,j \in \{1,2,3,4\}$ where $i \neq j$, $\g{U_{SW}}(Q(i,j))$ acts as a product of logical $\g{CX}$ gates on some pairs of logical qubits. For comparison, we also show below the transformation of logical Pauli operators by the fold-transversal gate $\g{U_{PH}}(Q(1,2))=\logg{C_{11}Z}(1,4)\logg{C_{11}Z}(2,5)$ $\logg{C_{11}Z}(3,6)\logg{C_{00}Z}(2,3)$, which was previously discussed in \cref{ex:Q12_on_QRM4}.
\begin{equation}
    \begin{tabular}{r r c l r c l c r c l r c l}
        $\g{U_{PH}}(Q(1,2))$: & $\logg{X}(1)$ & $\mapsto$ & $\phantom{+}\logg{X}(1)\logg{Z}(4)$,  & $\logg{Z}(1)$ & $\mapsto$ & $\logg{Z}(1)$, & \quad & $\logg{X}(4)$ & $\mapsto$ & $\logg{X}(4)\logg{Z}(1)$,  & $\logg{Z}(4)$ & $\mapsto$ & $\logg{Z}(4)$,\\
        & $\logg{X}(2)$ & $\mapsto$ & $-\logg{X}(2)\logg{Z}(3)\logg{Z}(5)$,  & $\logg{Z}(2)$ & $\mapsto$ & $\logg{Z}(2)$, & \quad & $\logg{X}(5)$ & $\mapsto$ & $\logg{X}(5)\logg{Z}(2)$,  & $\logg{Z}(5)$ & $\mapsto$ & $\logg{Z}(5)$,\\
        & $\logg{X}(3)$ & $\mapsto$ & $-\logg{X}(3)\logg{Z}(2)\logg{Z}(6)$,  & $\logg{Z}(3)$ & $\mapsto$ & $\logg{Z}(3)$, & \quad & $\logg{X}(6)$ & $\mapsto$ & $\logg{X}(6)\logg{Z}(3)$,  & $\logg{Z}(6)$ & $\mapsto$ & $\logg{Z}(6)$.
    \end{tabular} \nonumber
\end{equation}

In this work, we did not prove the closed forms for the logical operations of $\g{U_{SW}}(P(i,j))$, $\g{U_{PH}}(P(i,j))$, and $\g{U_{SW}}(Q(i,j))$ for any $\mathrm{QRM}(m)$ since these gates are not necessary for the generation of the full logical Clifford group. Nevertheless, we believe that such gates could be useful for optimizing the depth of the circuit implementing an arbitrary logical Clifford gate. 

\section{Fundamental limitations on circuit depth for realizing an arbitrary logical Clifford gate} \label{sec:limitations}

\cref{thm:full_Clifford} ensures that any logical Clifford gate in $\overline{C}_k$ can be generated by the logical gates, each of which can be implemented by a constant-depth application of physical Clifford gates without requiring ancilla qubits.
Nevertheless, this does not mean that all logical Clifford gates can be implemented in constant depth, as one generally needs to sequentially apply generators of the full Clifford group. 
In fact, our construction implements the standard addressable Clifford gates $\logg{H}$, $\logg{S}$, and $\logg{C_{00}Z}$ on the specific logical qubits in depth 
\bal
 \logg{H}: O(\sqrt{n}),\quad \logg{S}:O(\sqrt{n}),\quad \logg{C_{00}Z}:O(\sqrt{n}\log n) 
 \label{eq:standard Cliffod gates}
\eal
which are not constant in $n$.

One might wonder whether it is possible to come up with a better construction that only requires constant depth to implement an arbitrary logical Clifford gate. The following results suggest that it is not the case---the fact that our construction cannot achieve a constant-depth implementation of some logical Clifford gate is not an artifact of the specific choice of our method, but a consequence of the fundamental limitation.
\begin{theorem}[The circuit depth of an arbitrary logical Clifford gate] \label{thm:depth fully addressable Clifford}
   For a given family of $\codepar{n,k,d}$ codes, let $\overline{C}_k$ be the full logical Clifford group of the code. Also, let $C_l$ be the set of physical Clifford gates that have support on at most $l$ qubits. Suppose that $\overline{C}_k$ can be generated by physical gates in $C_l$ where $l=O(1)$. Then, there exists a logical Clifford gate $\logg{G} \in \overline{C}_k$ such that any implementation of $\logg{G}$ by physical gates in $C_l$ has depth $\Omega(\frac{k^2}{n\log n})$.
\end{theorem}
\begin{proof}
Let $N_{l,n}$ be the number of unitary operators that can be realized by applying gates from $C_l$ on $n$ qubits in one layer without overlap of the support. 
$N_{l,n}$ can be upper-bounded by the number of ways of dividing $\left\lceil\frac{n}{l}\right\rceil l$ qubits into $\left\lceil\frac{n}{l}\right\rceil$ sets of $l$ qubits, multiplied by $|C_l|^{\left\lceil\frac{n}{l}\right\rceil}$, i.e., the number of ways of applying Clifford gates that have support on at most $l$ qubits in each set. Namely, it holds that
\bal
 N_{l,n}\leq \frac{\left(\left\lceil\frac{n}{l}\right\rceil l\right)!}{(l!)^{\left\lceil\frac{n}{l}\right\rceil
 }\left\lceil\frac{n}{l}\right\rceil!}|C_l|^{^{\left\lceil\frac{n}{l}\right\rceil
 }}.
\eal
Therefore, the number of different gates realizable by a depth-$D$ application of gates in $C_l$ is upper bounded by 
\bal
 \left(N_{l,n}\right)^D\leq \left[\frac{\left(\left\lceil\frac{n}{l}\right\rceil l\right)!}{(l!)^{\left\lceil\frac{n}{l}\right\rceil
 }\left\lceil\frac{n}{l}\right\rceil!}|C_l|^{^{\left\lceil\frac{n}{l}\right\rceil
 }}\right]^D.
\eal
To realize all Clifford gates in $\overline{C}_k$, it is necessary that the number of realizable gates after depth $D$ is greater than the number of $k$-qubit Clifford gates. 
This requires 
\bal
|C_k|\leq \left[\frac{\left(\left\lceil\frac{n}{l}\right\rceil l\right)!}{(l!)^{\left\lceil\frac{n}{l}\right\rceil
 }\left\lceil\frac{n}{l}\right\rceil!}|C_l|^{^{\left\lceil\frac{n}{l}\right\rceil
 }}\right]^D,
\eal
which provides a lower bound for the depth as 
\bal
 D\geq \frac{\log |C_k|}{\log \left(\frac{\left(\left\lceil\frac{n}{l}\right\rceil l\right)!}{(l!)^{\left\lceil\frac{n}{l}\right\rceil
 }\left\lceil\frac{n}{l}\right\rceil!}|C_l|^{^{\left\lceil\frac{n}{l}\right\rceil
 }}\right)}.
 \label{eq:depth lower bound explicity}
\eal

Let us now focus on the large $n$ and $k$ regime and see how the lower bound scales with these parameters.
We first note that, using the assumption that $l=O(1)$, we get 
\bal
 \log \left(\frac{\left(\left\lceil\frac{n}{l}\right\rceil l\right)!}{(l!)^{\left\lceil\frac{n}{l}\right\rceil
 }\left\lceil\frac{n}{l}\right\rceil!}\right) = \Theta(n\log n).
\eal
Also, the number of $n$-qubit Clifford gates is known to be~\cite{Koenig_2014_howto}
\bal
 |C_n| = 2^{n^2+2n}\prod_{j=1}^n(4^j-1),
\eal
satisfying $\log |C_n| = \Theta(n^2)$.

Together with \eqref{eq:depth lower bound explicity}, we get 
\bal
 D\geq \Omega\left(\frac{k^2}{n\log n}\right)
\eal
where we used that $|C_l|=O(1)$.
\end{proof}

\cref{thm:depth fully addressable Clifford} particularly means that the constant-depth implementation of an arbitrary logical Clifford gate in $\overline{C}_k$ is impossible for codes with high-encoding rates satisfying $k=\omega(\sqrt{n\log n})$.
This is the case for the quantum Reed--Muller codes focused in this work, which gives the following fundamental restriction. 

\begin{corollary}\label{cor:depth fully addressable Reed--Muller_V2}
 Let $m$ be a positive even number, $\mathrm{QRM}(m)$ be the $\codepar{n=2^m,k=\binom{m}{m/2},d=2^{m/2}}$ quantum Reed--Muller code defined in \cref{def:QRM}, and $\overline{C}_k$ be the full logical Clifford group. Then, there exists a logical Clifford gate $\logg{G} \in \overline{C}_k$ such that for any implementation of $\logg{G}$ that uses only transversal and fold-transversal gates comprising single- and two-qubit physical Clifford gates, the implementation has depth $\Omega\left(\frac{n}{(\log n)^2}\right)$.
\end{corollary}
\begin{proof}
   This follows by applying \cref{thm:depth fully addressable Clifford} to the quantum Reed--Muller codes and noting that $k=\binom{m}{m/2} = \Theta(n/\sqrt{\log n})$.
\end{proof}
\cref{cor:depth fully addressable Reed--Muller_V2} not only shows that some logical Clifford gates cannot be implemented in constant depth on the quantum Reed--Muller codes, 
but also implies that the implementation of the standard logical Clifford gates in \eqref{eq:standard Cliffod gates} has a ``better than the worst'' depth scaling.

To implement a Clifford circuit on the logical level of a quantum Reed--Muller code, one can turn each 1- and 2-qubit Clifford gate in the logical circuit into its corresponding physical implementation from \cref{subsec:addr_gates}. Whether this compilation of an arbitrary Clifford circuit provides an implementation that matches the lower bound in \cref{cor:depth fully addressable Reed--Muller_V2} is still an open question. However, we conjecture that this might not be the case for two reasons. First, \cref{cor:depth fully addressable Reed--Muller_V2} is based on a rather simple counting argument, so the worst-case logical Clifford gate might require a circuit that scales worse than $\Omega\left(\frac{n}{(\log n)^2}\right)$. Second, the compilation discussed above might be suboptimal since some combination of logical gates can be easily implemented using a single layer of fold-transversal gates, while an implementation of each of the addressable gates in the combination requires a much deeper circuit. For example, for any $\mathrm{QRM}(m)$, the logical operation $\logg{C_{11}Z}(1,k/2+1)\logg{C_{11}Z}(2,k/2+2)\cdots\logg{C_{11}Z}(k/2,k)$ can be implemented by the depth-1 circuit $\g{U_{PH}}(e)$ (similarly to the gate $\g{U_{PH}}(e)$ for $\mathrm{QRM}(4)$ in \cref{ex:logical_from_fold,ex:fold_product}), while a compilation that implements each $\logg{C_{11}Z}(i,k/2+i)$ separately requires a circuit of depth $\frac{k}{2} O(\sqrt{n}{\log n}) = O(n^{3/2}\sqrt{\log n})$ in total. We believe that our lower bound could be improved and a more efficient circuit compilation tailored to our gate construction could be developed. We leave these studies to future work.

\section{Discussion and conclusions} \label{sec:discussion}

QECCs that encode multiple logical qubits per code block seem to be good candidates for large-scale quantum computers as the information can be encoded more efficiently than well-known codes such as surface codes or color codes. However, one big question for such high-rate codes is how logical operations on a specific set of logical qubits in the code block can be performed with low overhead. In this work, we are interested particularly in logical operations through transversal and fold-transversal gates as they require constant time overhead with no ancilla qubits. We study the family of \codepar{n=2^m,k={m \choose m/2},d=2^{m/2}} quantum Reed--Muller codes and show that it is possible to construct \emph{any} addressable Clifford gate in the full logical Clifford group using only transversal and fold-transversal gates. To our knowledge, this is the first known construction of the full logical Clifford group using only transversal and fold-transversal gates for a family of codes in which $k$ grows near-linearly in $n$.

Explicit constructions of any addressable $\logg{H}$, $\logg{S}$, $\logg{C_{00}Z}$, or $\logg{SW}$ gate as a sequence of transversal and fold-transversal gates for any code in the studied quantum Reed--Muller code family can be found using our Python package \cite{Chan2026}.
The package can be used to verify all theorems in \cref{sec:logical_for_QRM}, and we have numerically verified \crefrange{thm:PnQ_preserve_stb}{thm:fold_product} up to $m =10$ and \cref{cor:addr_S_CZ,thm:addr_H,thm:addr_SW_CZ} up to $m =6$. Our package can construct and compute the logical action of other fold-transversal gates such as $\g{U_{SW}}(P(i,j))$, $\g{U_{PH}}(P(i,j))$, and $\g{U_{SW}}(Q(i,j))$, assisting further exploration of logical gates for quantum Reed--Muller codes.

We note the differences between our construction and that of Gong and Renes in \cite{GR24}, in which the full logical Clifford group on the same family of quantum Reed--Muller codes can also be achieved using only transversal and fold-transversal gates. In their work, any logical $\g{CX}$-type gate (including intra-block addressable $\logg{CX}$ gates) is constructed using the technique proposed by Grassl and Roetteler in \cite{GR13}, which requires transversal $\g{CX}$ gates and an additional ancilla block encoded in the same code as the data block. Any logical phase-type gate (including addressable $\logg{S}$ gates) is constructed from $\logg{CX}$ gates and a fold-transversal gate involving physical $\g{S}$ and $\g{C_{11}Z}$ gates. Logical $\g{CX}$-type and phase-type gates together with a Hadamard-type gate $\g{H}^{\otimes n}$ generate the full logical Clifford group for the code on the data block. In our work, any addressable $\logg{S}$ and some intra-block $\logg{C_{00}Z}$ gates are constructed from a sequence of fold-transversal gates involving physical $\g{S}$ and $\g{C_{11}Z}$ gates. These gates and $\g{H}^{\otimes n}$ are then used to construct any addressable $\logg{H}$, intra-block $\logg{SW}$, or intra-block $\logg{C_{00}Z}$ gate, which in turn generate the full logical Clifford group. The main difference is that an additional ancilla block of code is not required in our construction. 

Although our construction can save the space overhead by 50\%, the required time overhead may be smaller or larger than the construction in \cite{GR24} depending on the Clifford circuit being implemented. While the construction of \cite{GR24} can be used to derive an upper bound on the circuit depth required to implement a general logical $\g{CX}$-type gate, such a bound may overestimate the depth required for a particular addressable gate which has additional structure. It was observed in \cite{MGFCS+25} for SHYPS codes that, although an addressable intra-block $\logg{CX}$ gate can be implemented without an ancilla block, employing the technique of \cite{GR13} can lead to a lower-depth implementation. However, this depth reduction does not apply to all Clifford circuits. This suggests that the construction of \cite{GR24}, which exploits the same technique, may likewise yield a circuit depth for certain addressable gates that is substantially lower than the one suggested by a generic gate synthesis.

We observe that some logical operations on the studied quantum Reed--Muller code family can be virtually realized by relabeling the physical qubits. Examples are the gates $\g{U_{SW}}(P(1,2)) = \logg{SW}(2,6)\logg{SW}(3,5)$ and $\g{U_{SW}}(Q(1,2)) = \logg{CX}(2,6)\logg{CX}(3,5)$ on the code $\mathrm{QRM}(4)$ (see \cref{subsec:other_gates}). This is similar to the logical operations on the phantom codes introduced by Koh et al.~\cite{KGDTG+26}; by the definition of phantom codes, an addressable $\logg{CX}$ gate on any pair of logical qubits can be implemented by physical qubit permutation. Unfortunately, on the quantum Reed--Muller code family that we study, an addressable $\logg{SW}$ or $\logg{CX}$ gate on a single pair of logical qubits cannot be obtained using the swap-type fold-transversal gates alone. It was shown in \cite{KGDTG+26} that some phantom codes can be derived from general quantum Reed--Muller codes by turning some logical operators into stabilizers; in other words, the addressability of $\logg{CX}$ gates can be obtained by sacrificing some logical qubits. We believe that our construction of addressable $\logg{S}$ and $\logg{C_{00}Z}$ gates should also be applicable for the phantom codes derived from quantum Reed--Muller codes. However, our construction for the other logical Clifford gates may not be applicable; this is because the gate $\g{H}^{\otimes n}$ (required to generate the full logical Clifford group) does not preserve the stabilizer group of phantom quantum Reed--Muller codes. Therefore, additional techniques may be required.

It is interesting to ask what properties of the quantum Reed--Muller codes allow for the construction of the full Clifford group from transversal and fold-transversal gates, and whether there are other code families with similar properties. We believe that there are relations between the fact that the full logical Clifford group of $\mathrm{QRM}(m)$ can be generated by $\g{H}^{\otimes n}$ and $\g{U_{PH}}\left(Q(K)\right), 0 \leq |K| \leq m/2$ (\cref{thm:full_Clifford}) and the fact that the general linear group $\textrm{GL}(m,2)$---which is a subgroup of $\textrm{AGL}(m,2)$, the automorphism group of $\textrm{RM}(m/2-1,m)$---can be generated by $Q(i,j), i,j \in [m], i\neq j$. However, whether our construction is applicable to other quantum codes from classical codes with a similar automorphism group is still not known, and further studies are required.

Suppose that there exists an \codepar{n,k,d} code (that is not a quantum Reed--Muller code) in which the logical Clifford group $\overline{C}_{\tilde{k}}$ ($\tilde{k} \leq k$) can be achieved using only transversal and fold-transversal gates. We believe that there are some fundamental limitations on the encoding rate $k/n$, the number of logical qubits $\tilde{k}$ that allow for any logical Clifford gate, and the circuit depth. In particular, we conjecture that if $k/n$ is greater than a certain number, there exists a logical Clifford gate that cannot be implemented in constant time, and if $k/n$ is even greater, addressable Clifford gates on some logical qubits might not be achievable ($\tilde{k}$ is strictly less than $k$). Our results in \cref{thm:depth fully addressable Clifford,cor:depth fully addressable Reed--Muller_V2} partially answer the first conjecture, but we believe that our bound can be further improved. These limitations could be related to other no-go theorems for addressable gates from transversal gates such as the ones in \cite{GJ25,KZ25}. We leave the investigation of such limitations to future work.

The family of SHYPS codes developed in \cite{MGFCS+25} is another family of high-rate codes in which the full logical Clifford group can be generated using only transversal and fold-transversal gates. However, the number of logical qubits $k$ in that case scales poly-logarithmically with the block length $n$. One key difference between the SHYPS codes and the quantum Reed--Muller codes of our interest is that the SHYPS codes are qLDPC while the quantum Reed--Muller codes are not. For each code family of qLDPC codes, the weight of each stabilizer generator is upper-bounded by some constant independent of the block length, making them desirable for the implementation of FTEC schemes whose circuits depend on the weight of stabilizer generators, such as Shor \cite{Shor96,DA07,TPB23} and flag \cite{CR18,CR20} FTEC schemes. For a family of non-qLDPC codes, in contrast, the weight of each stabilizer generator may scale with the block length, but this does not affect the implementation of FTEC schemes such as Steane \cite{Steane97,Steane04} or Knill \cite{Knill05} schemes, whose circuits are independent of stabilizer weight (assuming that the required ancilla code blocks can be prepared efficiently). Moreover, non-qLDPC codes such as quantum Reed--Muller codes, quantum Hamming codes \cite{Steane96_qHamming}, and quantum Bose–Chaudhuri–Hocquenghem (BCH) codes \cite{GB99} generally have higher encoding rates compared to qLDPC codes of the same distance. Whether qLDPC or non-qLDPC codes are more practical could depend on the platform of the quantum computer, architecture, or the available resources.

We point out that our construction only shows the possibility to construct any logical Clifford gate through transversal and fold-transversal gates, but it might not be the most efficient way to do so. Besides the gates used for our construction of the full Clifford group, there are other logical gates that can be constructed by transversal and fold-transversal gates, such as $\g{U_{SW}}(P(i,j))$, $\g{U_{PH}}(P(i,j))$, and $\g{U_{SW}}(Q(i,j))$. Although these gates are not necessary for the generation of the full Clifford group, they might be useful for circuit depth optimization. Also, some non-addressable Clifford gates such as $\prod_{i=1}^{k/2} \logg{C_{11}Z}(i,k/2+i)$ can be implemented using only a few layers of fold-transversal gates, suggesting that decomposing an arbitrary Clifford circuit into addressable gates and applying our construction on each of them might not be the best solution. To minimize the overall circuit depth, a better way to compile arbitrary Clifford circuits may be required.

Our work could provide an alternative way to think about how a general circuit on a high-rate QECC can be constructed. Conventionally, one might aim to find an addressable gate that can be implemented in constant time for a certain code (or find a code in which such an implementation is possible) then construct a general circuit from such addressable gates. Our approach is different: logical gates that can be implemented on the quantum Reed--Muller codes in constant time are not addressable, but addressable gates can be made by combining such constant-depth logical gates. It has been demonstrated in several works that transversal gates generally implement non-addressable gates on a high-rate code (see \cite{TTF25} for example). For this reason, we believe that our approach to construct addressable gates and a general circuit might be more natural than the conventional approach, but efficient circuit compilation would also be necessary to optimize the circuit depth.

It is worth discussing the fault-tolerant properties of transversal and fold-transversal gates. Transversal gates do not spread errors inside each code block, and each faulty gate can lead to no more than one error per block, so the ability to correct errors after propagation is guaranteed if the total number of faults and input errors is no more than $(d-1)/2$ (where $d$ is the code distance). Fold-transversal gates, in contrast, are not fault tolerant in a strict sense \cite{AGP06} since two-qubit gates applied between physical qubits inside a code block can spread errors. In particular, a physical controlled-$\g{Z}$ gate can propagate $\g{YI}$ to $\g{YZ}$ (which is treated as a product of $\g{XI}$ and $\g{ZZ}$ if the code can correct $\g{X}$-type and $\g{Z}$-type errors separately), and also a two-qubit Pauli error may arise from a single gate fault. Therefore, in the standard circuit-level depolarizing noise model where each single-qubit gate is followed by a single-qubit Pauli error $P \in \{\g{X},\g{Y},\g{Z}\}$ with error probability $p/3$ each, and each two-qubit gate is followed by a two-qubit Pauli error $P_1\otimes P_2 \in \{\g{I},\g{X},\g{Y},\g{Z}\}^{\otimes 2}\setminus\{\g{I}\otimes \g{I}\}$ with error probability $p/15$ each, fold-transversal gates do not satisfy the conditions for fault-tolerant gates proposed in \cite{AGP06}. Nevertheless, numerical simulations with specific code and gate constructions in several works \cite{GV24,CCLP24,CBZGM+25,SMB25} have shown that the performance of fold-transversal gates could approach ideal fault-tolerant operations in various noise models, so they can still be practically useful.

Suppose that a general Clifford circuit can be compiled as a sequence of transversal and fold-transversal gates (as in the case of quantum Reed--Muller codes) and assume that fold-transversal gates are fault tolerant under a particular noise model. To make the whole implementation fault tolerant, one needs to interleave FTEC gadgets into the sequence to constantly remove noise from the system. Note that such FTEC gadgets also require some space and time overhead. It is interesting to ask how much overhead is required for this procedure, and how it compares to other procedures to implement logical gates on a high-rate code such as the techniques in \cite{GR24,MGFCS+25} that require additional ancilla blocks, or teleportation-based FTQC schemes \cite{ZLC00,BZHJL15}. Careful simulations are required for this comparison, so we leave it for future work.  

Last, we note that the fault-tolerant simulation of Clifford gates alone are not sufficient to demonstrate the full performance of quantum computers, as any Clifford circuit can be efficiently simulated by a classical computer \cite{Gottesman98}. To achieve universal quantum computation, we still need to implement logical non-Clifford gates fault-tolerantly. One possible future direction could be developing magic state preparation protocols (such as the ones in \cite{KLZ97,BK05,ITHF25,GSJ24}) for the quantum Reed--Muller codes of our interest. It is also known that some quantum Reed--Muller codes which are not self-dual allow the application of logical non-Clifford gates through transversal gates \cite{BCHK26}. Another direction could be developing code-switching techniques (such as the ones in \cite{PR13,ADP14,GCZ25,THLGH25}) to fault-tolerantly switch between a high-rate quantum Reed--Muller code that admits all logical Clifford gates and another high-rate code that admits logical non-Clifford gates.

\section{Author contributions}

Theerapat Tansuwannont led the project, developed the construction of addressable gates for the quantum Reed--Muller codes, and contributed to the mathematical statements in \cref{sec:CRM_QRM,sec:logical_for_QRM,sec:proof_of_thms}. Tim Chan developed the Python package \cite{Chan2026} for explicit construction of addressable gates for the quantum Reed--Muller codes and made all figures. Ryuji Takagi developed the fundamental limitations of the circuit depth in realizing an arbitrary logical Clifford gate and contributed to the mathematical statements in \cref{sec:limitations}. All authors took part in writing \cref{sec:intro,sec:discussion} and revising the manuscript.

\section{Acknowledgements}

We thank Victor Albert, Keisuke Fujii, Anqi Gong, Shubham Jain, and Yugo Takada for helpful discussions on our preliminary results. We also thank Xiaozhen Fu, Jin Ming Koh, Shayan Majidy, Tom Scruby, Shiro Tamiya, and Satoshi Yoshida for interesting discussions.
Theerapat Tansuwannont was supported by JST Moonshot R\&D Grant Number JPMJMS2061 and JPMJMS256E. Tim Chan acknowledges support from
an EPSRC DTP studentship,
JST ASPIRE Japan Grant Number JPMJAP2319,
and three EPSRC projects:
QCS Hub (EP/T001062/1),
RoaRQ (EP/W032635/1),
and SEEQA (EP/Y004655/1).
Ryuji Takagi was supported by JST CREST Grant Number JPMJCR23I3, JSPS KAKENHI Grant Number JP24K16975, JP25K00924, and MEXT KAKENHI Grant-in-Aid for Transformative
Research Areas A ``Extreme Universe” Grant Number JP24H00943.

\appendix

\section{Proofs of theorems} \label{sec:proof_of_thms}

In this section, we prove \cref{thm:QK_preserve_stb,thm:logical_from_fold,thm:fold_product} which verify that the phase-type fold-transversal gate $\g{U_{PH}}(Q(K))$ from any permutation $Q(K)$ is a valid logical gate, provide the logical operation of each $\g{U_{PH}}(Q(K))$, and provide constructions of addressable $\logg{S}$ and addressable $\logg{C_{00}Z}$ gates from a product of $\g{U_{PH}}(Q(K))$.

To ease our explanation, we introduce some notations and lemmas that will be useful in the main proofs. We first define a replacement operator from an automorphism $Q(i,j)$ as follows.
\begin{definition} \label{def:RK}
    Consider a classical Reed--Muller code $\textrm{RM}(r,m)$ where $m$ is a positive even number and $r$ is an integer satisfying $0 \leq r \leq m$. Let $Q(i,j)$ be an automorphism of $\textrm{RM}(r,m)$ as defined in \cref{def:PQ_def} with the corresponding permutation matrix $M_{Q(i,j)}$, where $i,j \in [m]$, $i \neq j$. Also, let $K=\{(i_1,i_2),\dots,(i_{2c-1},i_{2c})\}$ be a set of ordered pairs of size $0 \leq |K| \leq m/2$ as defined in \cref{def:K}. The \emph{replacement operator} $R(i,j)$ for a pair of vector indices $(i,j)$ is defined as,
    \begin{equation}
        R(i,j) = M_{Q(i,j)}+I_n,
    \end{equation}
    where $I_n$ is the $n \times n$ identity matrix and the addition is modulo 2.
    For any set $K\neq \emptyset$, the \emph{replacement operator} $R(K)$ is defined as,
    \begin{equation}
        R(K) = R(i_1,i_2)\cdots R(i_{2c-1},i_{2c}). \label{eq:RK}
    \end{equation}
    The replacement operator for the empty set is defined as $R(\emptyset) = I_n$.
\end{definition}

Note that because we require the vector indices $i_1,\dots,i_{2c}$ that appear in $K$ to all be distinct, the order of operations in \cref{eq:RK} does not matter; i.e., $R(i_1,i_2)\cdots R(i_{2c-1},i_{2c})=R(i_{2c-1},i_{2c})\cdots R(i_{1},i_{2})$.

$R(i,j)$ and $R(K)$ have the following properties.
\begin{lemma} \label{lem:RK_on_v}
Let $A$ be any set of vector indices of size $0 \leq |A| \leq m/2$. For any vector indices $i,j\in[m]$,
    \begin{equation}
        R(i,j)\mathbf{v}_{A} = \begin{cases}
            \mathbf{v}_{\left(A\cup\{j\}\right)\setminus \{i\}} & \text{if } i \in A, \\
            \mathbf{0} & \text{if } i \notin A.
        \end{cases}
    \end{equation}
    For any set of ordered pairs $K$ (as defined in \cref{def:K}) such that $K \neq \emptyset$,
    \begin{equation}
        R(K)\mathbf{v}_{A} = \begin{cases}
            \mathbf{v}_{\left(A\cup F_2(K)\right)\setminus F_1(K)} & \text{if } F_1(K) \subseteq A, \\
            \mathbf{0} & \text{if } F_1(K) \nsubseteq A.
        \end{cases}
    \end{equation}
\end{lemma}
\begin{proof}
    First, observe that $R(i,j)\mathbf{v}_{A} = M_{Q(i,j)}\mathbf{v}_{A}+\mathbf{v}_{A}$. If $i \in A$, we have that $M_{Q(i,j)}\mathbf{v}_{A}=\mathbf{v}_{A}+\mathbf{v}_{\left(A\cup\{j\}\right)\setminus \{i\}}$, thus $R(i,j)\mathbf{v}_{A}=\mathbf{v}_{\left(A\cup\{j\}\right)\setminus \{i\}}$. If $i \notin A$, we have that $M_{Q(i,j)}\mathbf{v}_{A}=\mathbf{v}_{A}+\mathbf{v}_{A}=\mathbf{0}$.

    Next, let $K=\{(i_1,i_2),\dots,(i_{2c-1},i_{2c})\}$. We have that $F_1(K)=\{i_1,\dots,i_{2c-1}\}$. If $F_1(K) \subseteq A$, then $R(K)=R(i_1,i_2)\cdots R(i_{2c-1},i_{2c})$ transforms $\mathbf{v}_{A}$ as follows:
    \begin{equation}
        \mathbf{v}_{A} \xmapsto{R(i_1,i_2)} \mathbf{v}_{\left(A\cup\{i_2\}\right)\setminus \{i_1\}} \xmapsto{R(i_3,i_4)} \mathbf{v}_{\left(A\cup\{i_2,i_4\}\right)\setminus \{i_1,i_3\}} \xmapsto{R(i_5,i_6)} \cdots \xmapsto{R(i_{2c-1},i_{2c})} \mathbf{v}_{\left(A\cup F_2(K)\right)\setminus F_1(K)}. \nonumber
    \end{equation}
    If $F_1(K) \nsubseteq A$, there exists $i_\alpha \in F_1(K)$ which is not an element of $A$. Thus, there exists a term $R(i_\alpha,i_\beta)$ in the product $R(i_1,i_2)\cdots R(i_{2c-1},i_{2c})$ such that $R(i_\alpha,i_\beta)\mathbf{v}_{A}=\mathbf{v}_{A}+\mathbf{v}_{A}=\mathbf{0}$. This implies that $R(K)\mathbf{v}_{A}=\mathbf{0}$.
\end{proof}

Additionally, for each $K$, we define $\mathfrak{Q}(K)$ to be the sum of permutation matrices of all subsets of $K$ and define $\mathfrak{R}(K)$ to be the sum of replacement operators of all subsets of $K$. The formal definitions are as follows.

\begin{definition} \label{def:nestedQK}
For any set of ordered pair of vector indices $K$ as defined in \cref{def:K}, we define $\mathfrak{Q}(K)$ and $\mathfrak{R}(K)$ as
    \begin{equation}
        \mathfrak{Q}(K) = \sum_{L \subseteq K} M_{Q(L)},
    \end{equation}
    and,
    \begin{equation}
        \mathfrak{R}(K) = \sum_{L \subseteq K} R(L),
    \end{equation}
    where the sum is modulo 2. 
\end{definition}
By \cref{def:nestedQK}, we have $\mathfrak{Q}(\emptyset)=I_n$ and $\mathfrak{R}(\emptyset)=I_n$. Using \cref{def:RK,def:nestedQK}, we can prove the following lemma.
\begin{lemma} \label{lem:RK_QK_equiv}
    Let $m$ be a positive even number, $K=\{(i_1,i_2),\dots,(i_{2c-1},i_{2c})\}$ be a set of ordered pair of vector indices as defined in \cref{def:K}, and let $i_\alpha,i_\beta \in [m]$. 
    \begin{enumerate}
        \item For any $i_\alpha,i_\beta$,
        \begin{enumerate}
            \item $R(i_\alpha,i_\beta) = \mathfrak{Q}(i_\alpha,i_\beta)$;
        \item $M_{Q(i_\alpha,i_\beta)} = \mathfrak{R}(i_\alpha,i_\beta)$.
        \end{enumerate}
        \item For any $K$,
        \begin{enumerate}
            \item $R(K) = R(i_1,i_2)\cdots R(i_{2c-1},i_{2c}) = \mathfrak{Q}(i_1,i_2)\cdots\mathfrak{Q}(i_{2c-1},i_{2c}) = \mathfrak{Q}(K)$;
            \item $M_{Q(K)} = M_{Q(i_1,i_2)}\cdots M_{Q(i_{2c-1},i_{2c})} = \mathfrak{R}(i_1,i_2)\cdots\mathfrak{R}(i_{2c-1},i_{2c}) = \mathfrak{R}(K)$.
        \end{enumerate}
    \end{enumerate}
\end{lemma}
\begin{proof}
    (1.a) and (1.b) can be easily shown using the facts that $R(\emptyset)=I_n$, $M_{Q(\emptyset)}=I_n$, and $R(i,j)=M_{Q(i,j)}+I_n$.

    (2.a) Using the result from (1.a), we can write,
    \begin{equation}
        R(K)=R(i_1,i_2)\cdots R(i_{2c-1},i_{2c}) = \mathfrak{Q}(i_1,i_2)\cdots\mathfrak{Q}(i_{2c-1},i_{2c}). \nonumber
    \end{equation}
    Additionally, using the fact that $R(i,j)=M_{Q(i,j)}+I_n$, we can write, 
    \begin{equation}
        R(K)=R(i_1,i_2)\cdots R(i_{2c-1},i_{2c}) = \left(I_n+M_{Q(i_1,i_2)}\right)\cdots\left(I_n+M_{Q(i_{2c-1},i_{2c})}\right) = \sum_{L \subseteq K} M_{Q(L)} = \mathfrak{Q}(K). \nonumber
    \end{equation}

    (2.b) Using the result from (1.b), we can write,
    \begin{equation}
        M_{Q(K)}=M_{Q(i_1,i_2)}\cdots M_{Q(i_{2c-1},i_{2c})} = \mathfrak{R}(i_1,i_2)\cdots\mathfrak{R}(i_{2c-1},i_{2c}).\nonumber
    \end{equation}
    Additionally, using the fact that $M_{Q(i,j)}=R(i,j)+I_n$, we can write, 
    \begin{equation}
        M_{Q(K)}=M_{Q(i_1,i_2)}\cdots M_{Q(i_{2c-1},i_{2c})} = \left(I_n+R(i_1,i_2)\right)\cdots\left(I_n+R(i_{2c-1},i_{2c})\right) = \sum_{L \subseteq K} R(L) = \mathfrak{R}(K). \nonumber
    \end{equation}
\end{proof}

In addition, we prove the Hamming weight of a vector of the form $\mathbf{v}_A \wedge Q(K)\mathbf{v}_A$ where $0 \leq |A| \leq \frac{m}{2}-1$ or $\mathbf{v}_B \wedge Q(K)\mathbf{v}_B$ where $|B|=\frac{m}{2}$ in the following lemma.
\begin{lemma} \label{lem:phase_from_QK}
    Let $m$ be a positive even number, $\mathrm{QRM}(m)$ be the quantum Reed--Muller code defined in \cref{def:QRM}, and $K$ be a set of ordered pairs defined in \cref{def:K} of size $0 \leq |K| \leq m/2$ with the corresponding automorphism $Q(K)$. Also, let $A \subsetneq [m]$ be a set of basis vector indices of size $0 \leq |A| \leq \frac{m}{2}-1$, and $B \subsetneq [m]$ be a set of basis vector indices of size $|B|=\frac{m}{2}$.
    \begin{enumerate}
        \item For any $K$ and any $A$, $\mathrm{wt}(\mathbf{v}_A \wedge Q(K)\mathbf{v}_A)$ is divisible by 4.
        \item For any $K$ and any $B$,
        \begin{enumerate}
            \item if there exists $(i,j) \in K$ such that $i,j \in B$, then $\mathrm{wt}(\mathbf{v}_B \wedge Q(K)\mathbf{v}_B)=0$.
            \item if there is no $(i,j) \in K$ such that $i,j \in B$,
            \begin{enumerate}
                \item if $|B \cap F_1(K)| = m/2$ (possible when $|K|=m/2$), then $\mathrm{wt}(\mathbf{v}_B \wedge Q(K)\mathbf{v}_B)=1$;
                \item if $|B \cap F_1(K)| = m/2-1$ (possible when $|K|=m/2$ or $m/2-1$), then $\mathrm{wt}(\mathbf{v}_B \wedge Q(K)\mathbf{v}_B)=2$;
                \item if $|B \cap F_1(K)| \leq m/2-2$, then $\mathrm{wt}(\mathbf{v}_B \wedge Q(K)\mathbf{v}_B)$ is divisible by 4.
            \end{enumerate}
        \end{enumerate}
    \end{enumerate}
\end{lemma}
\begin{proof}
    (1) First, observe that for any vector indices $i,j$, 
    \begin{equation}
        \mathbf{v}_i\wedge Q(i,j)\mathbf{v}_i = \mathbf{v}_i\wedge (\mathbf{v}_i+\mathbf{v}_j) = \mathbf{v}_i\wedge(\mathbf{1}+\mathbf{v}_j).
    \end{equation}
    Next, observe that $\mathbf{v}_A=\mathbf{v}_{A\setminus F_1(K)}\wedge \mathbf{v}_{A\cap F_1(K)}$, and $Q(K)\mathbf{v}_A = \mathbf{v}_{A\setminus F_1(K)}\wedge \left(Q(K)\mathbf{v}_{A\cap F_1(K)}\right)$. Thus,
    \begin{align}
        \mathbf{v}_A \wedge Q(K)\mathbf{v}_A &= \mathbf{v}_{A\setminus F_1(K)} \wedge \mathbf{v}_{A\cap F_1(K)} \wedge \left(Q(K)\mathbf{v}_{A\cap F_1(K)}\right), \nonumber \\
        &= \mathbf{v}_{A\setminus F_1(K)} \wedge \bigwedge_{\substack{i \in A \cap F_1(K)\\ (i,j) \in K}} \left(\mathbf{v}_i\wedge(\mathbf{1}+\mathbf{v}_j)\right). \label{eq:lem_phase1}
    \end{align}
    If there exists $(i,j) \in K$ such that $i,j \in A$, then \cref{eq:lem_phase1} is zero since $\mathbf{v}_j\wedge(\mathbf{1}+\mathbf{v}_j)=\mathbf{0}$. In this case, $\mathrm{wt}(\mathbf{v}_A \wedge Q(K)\mathbf{v}_A) = 0$. If there is no such $(i,j)$, then \cref{eq:lem_phase1} is a wedge product consisting of $|A\setminus F_1(K)|+2|A \cap F_1(K)|$ terms in total. Let $l = |A \cap F_1(K)|$, so $|A\setminus F_1(K)|+2|A \cap F_1(K)| = |A|-l+2l = |A|+l$, and thus $\mathrm{wt}(\mathbf{v}_A \wedge Q(K)\mathbf{v}_A) = 2^{m-|A|-l}$ by \cref{prop:Hamming_weight}. Since $0 \leq |A| \leq \frac{m}{2}-1$ and $l$ is an integer such that $0 \leq l \leq |A|$, we have that $2 \leq m-|A|-l \leq m$. Therefore, $\mathrm{wt}(\mathbf{v}_A \wedge Q(K)\mathbf{v}_A)$ is divisible by 4.

    (2) Similar to the previous case, we have that 
    \begin{equation}
        \mathbf{v}_B \wedge Q(K)\mathbf{v}_B = \mathbf{v}_{B\setminus F_1(K)} \wedge \bigwedge_{\substack{i \in B \cap F_1(K)\\ (i,j) \in K}} \left(\mathbf{v}_i\wedge(\mathbf{1}+\mathbf{v}_j)\right). \label{eq:lem_phase2}
    \end{equation}
    If there exists $(i,j) \in K$ such that $i,j \in B$, then \cref{eq:lem_phase2} is zero, and $\mathrm{wt}(\mathbf{v}_B \wedge Q(K)\mathbf{v}_B) = 0$. If there is no such $(i,j)$, then $\mathrm{wt}(\mathbf{v}_B \wedge Q(K)\mathbf{v}_B) = 2^{m-|B|-l}$, where $l = |B \cap F_1(K)|$. Note that $|B|=m/2$, and $0 \leq l \leq |B|$. (i) If $l = m/2$ (which can occur only when $|K|=m/2$), then $\mathrm{wt}(\mathbf{v}_B \wedge Q(K)\mathbf{v}_B)=1$. (ii) If $l = m/2-1$ (which can occur only when $|K|=m/2$ or $m/2-1$), then $\mathrm{wt}(\mathbf{v}_B \wedge Q(K)\mathbf{v}_B)=2$. (iii) If $l \leq m/2-2$, then $\mathrm{wt}(\mathbf{v}_B \wedge Q(K)\mathbf{v}_B)$ is divisible by 4.
\end{proof}

Now we are ready to prove \cref{thm:QK_preserve_stb,thm:logical_from_fold,thm:fold_product}.

\subsection{Proof of \cref{thm:QK_preserve_stb}}
By \cref{prop:SW_PH_transform} and \cref{lem:RK_QK_equiv}, $\g{U_{PH}}\left(Q(K)\right)$ maps $g_x(A)$ as follows.
\begin{align}
    g_x(A)=\g{X}(\mathbf{v}_A) &\mapsto (\up{i})^c \g{X}(\mathbf{v}_A)\g{Z}(Q(K)\mathbf{v}_A) \nonumber \\
    &= (\up{i})^c \g{X}(\mathbf{v}_A)\g{Z}(\mathfrak{R}(K)\mathbf{v}_A) \nonumber \\
    &= (\up{i})^c \g{X}(\mathbf{v}_A)\g{Z}\left(\sum_{L\subseteq K}R(L)\mathbf{v}_A\right) \nonumber \\
    &= (\up{i})^c \g{X}(\mathbf{v}_A)\prod_{L\subseteq K}\g{Z}(R(L)\mathbf{v}_A),
\end{align} 
where $c = \mathrm{wt}(\mathbf{v}_A \wedge Q(K)\mathbf{v}_A)$. By \cref{lem:phase_from_QK}, $\mathrm{wt}(\mathbf{v}_A \wedge Q(K)\mathbf{v}_A)$ is divisible by 4 for any $K$ and any $A$, so the phase factor $(\up{i})^c$ is always 1. We also find that for any $L \subseteq K$, 
    \begin{enumerate}
        \item if $F_1(L) \subseteq A$, then $R(L)\mathbf{v}_A = \mathbf{v}_{A'}$ where $A'=\left(A\cup F_2(L)\right) \setminus F_1(L)$ (by \cref{lem:RK_on_v}). Thus, $\g{Z}(R(L)\mathbf{v}_A)= \g{Z}(\mathbf{v}_{A'})$, which is another stabilizer generator since $|A'| \leq |A|$;
        \item if $F_1(L) \nsubseteq A$, then $R(L)\mathbf{v}_A = \mathbf{0}$ (by \cref{lem:RK_on_v}). Thus, $\g{Z}(R(L)\mathbf{v}_A)=\g{I}^{\otimes n}$.
    \end{enumerate}
    As a result, $\g{U_{PH}}\left(Q(K)\right)$ maps any $\g{X}$-type stabilizer generator to a product of the same $\g{X}$-type stabilizer generator and some $\g{Z}$-type stabilizer, and maps any $\g{Z}$-type stabilizer generator to itself. Thus, the stabilizer group of is preserved under the operation of $\g{U_{PH}}\left(Q(K)\right)$. \qed

\subsection{Proof of \cref{thm:logical_from_fold}}
Consider each logical qubit $B$. $\g{U_{PH}}\left(Q(K)\right)$ maps $\logg{Z}(B)$ to itself and maps $\logg{X}(B)$ as follows (by \cref{prop:SW_PH_transform} and \cref{lem:RK_QK_equiv}).
\begin{align}
    \logg{X}(B)=\g{X}(\mathbf{v}_B) &\mapsto (\up{i})^c \g{X}(\mathbf{v}_B)\g{Z}(Q(K)\mathbf{v}_B) \nonumber \\
    &=(\up{i})^c \g{X}(\mathbf{v}_B)\g{Z}(\mathfrak{R}(K)\mathbf{v}_B) \nonumber \\
    &=(\up{i})^c \g{X}(\mathbf{v}_B)\g{Z}\left(\sum_{L \subseteq K} R(L)\mathbf{v}_B\right) \nonumber \\
    &=(\up{i})^c \g{X}(\mathbf{v}_B)\prod_{L \subseteq K}\g{Z}(R(L)\mathbf{v}_B),
\end{align}
where $c = \mathrm{wt}(\mathbf{v}_B \wedge Q(K)\mathbf{v}_B)$.
We find that for any $L \subseteq K$,
    \begin{enumerate}
        \item If $F_1(L) \subseteq B$ and $F_2(L)\cap B = \emptyset$, then $R(L)\mathbf{v}_B = \mathbf{v}_{B''}$, where $B'' = \left(B\cup F_2(L)\right)\setminus F_1(L)$ is a set of size $m/2$. Note that $\g{Z}(\mathbf{v}_{B''}) = \logg{Z}((B'')^\up{c}) = \logg{Z}({B'})$ where $B' = \left(B^\up{c}\cup F_1(L)\right)\setminus F_2(L)$; 
        \item If $F_1(L) \nsubseteq B$, then $R(L)\mathbf{v}_B = \mathbf{0}$. Note that $\g{Z}(\mathbf{0})=\g{I}^{\otimes n}$; 
        \item If $F_1(L) \subseteq B$ but $F_2(L) \cap B \neq \emptyset$ (which is possible only when $|L| \leq m/2-1$), then $R(L)\mathbf{v}_B= \mathbf{v}_A$ where $A = \left(B\cup F_2(L)\right)\setminus F_1(L)$ is a set of size $m/2-|F_2(L) \cap B| \leq m/2-1$. Note that $\g{Z}(\mathbf{v}_{A})$ is a $\g{Z}$-type stabilizer. 
    \end{enumerate}
From this observation, the terms in $\prod_{L \subseteq K}\g{Z}(R(L)\mathbf{v}_B)$ that contribute to logical $\g{Z}$ operators are the terms in which $L$ satisfies $F_1(L) \subseteq B$ and $F_2(L)\cap B = \emptyset$. 

First, let us consider the case that $|K| \leq m/2-2$. In this case, $c$ is always divisible by 4 (by \cref{lem:phase_from_QK}), so for a logical qubit $B$, 
\begin{equation}
    \logg{X}(B)=\g{X}(\mathbf{v}_B) \mapsto \g{X}(\mathbf{v}_B)\prod_{L \subseteq K}\g{Z}(R(L)\mathbf{v}_B).
\end{equation}
The transformation corresponding to a specific $L$ is $\logg{X}(B) \mapsto \logg{X}(B)\logg{Z}({B'})$ where $B' = \left(B^\up{c}\cup F_1(L)\right)\setminus F_2(L)$. Meanwhile, a similar term can also be found in the transformation for the logical qubit $B'$; $\logg{X}(B') \mapsto \logg{X}(B')\logg{Z}({B})$ where $B = \left((B')^\up{c}\cup F_1(L)\right)\setminus F_2(L)$. Such terms correspond to a logical $\g{C_{11}Z}$ gate $\logg{C_{11}Z}(B,B')$. In other words, $\g{U_{PH}}\left(Q(K)\right)$ induces $\logg{C_{11}Z}(B,B')$ between any unordered pair of logical qubits $\{B,B'\}$ satisfying $F_1(L) \subseteq B$, $F_2(L)\cap B = \emptyset$, and $B' = \left(B^\up{c}\cup F_1(L)\right)\setminus F_2(L)$ for all $L \subseteq K$. 

Next, let us consider the case that $|K| = m/2-1$. In this case, $c=2$ only if $|B\cap F_1(K)|=m/2-1$ and there is no $(i,j) \in K$ such that $i,j \in B$; otherwise, $c$ is divisible by 4 (by \cref{lem:phase_from_QK}). We find the following.
\begin{enumerate}
    \item A logical qubit $B$ that satisfies the conditions that give $c=2$ is of the form $B = F_1(K)\cup\{j'\}$ where $j' \in [m]\setminus\left(F_1(K)\cup F_2(K)\right)$. For such $B$, 
    \begin{equation}
        \logg{X}(B)=\g{X}(\mathbf{v}_B) \mapsto (-1)\g{X}(\mathbf{v}_B)\prod_{L \subseteq K}\g{Z}(R(L)\mathbf{v}_B).
    \end{equation}
    We can write the transformation corresponding to a specific $L \subseteq K$ satisfying $F_1(L) \subseteq B$ and $F_2(L)\cap B = \emptyset$ as follows.
    \begin{enumerate}
        \item For $L$ such that $|L|= m/2-1$, $\logg{X}(B) \mapsto (-1)\logg{X}(B)\logg{Z}({B'})$ and $\logg{X}(B') \mapsto (-1)\logg{X}(B')\logg{Z}({B})$, where $B' = \left(B^\up{c}\cup F_1(L)\right)\setminus F_2(L)$; and
        \item For $L$ such that $|L|\leq m/2-2$, $\logg{X}(B) \mapsto \logg{X}(B)\logg{Z}({B'})$ and $\logg{X}(B') \mapsto \logg{X}(B')\logg{Z}({B})$, where $B' = \left(B^\up{c}\cup F_1(L)\right)\setminus F_2(L)$.
    \end{enumerate}
    Note that there is only one $L \subseteq K$ of size $|L|=m/2-1$, which is $L=K$. Thus, the overall phase from the combined transformation of $\logg{X}(B)$ is $-1$ as expected.
    \item For any logical qubit $B$ that does not satisfy the conditions above (i.e., $|B\cap F_1(K)|\leq m/2-2$ or there is $(i,j) \in K$ such that $i,j \in B$), we find that
    \begin{equation}
        \logg{X}(B)=\g{X}(\mathbf{v}_B) \mapsto \g{X}(\mathbf{v}_B)\prod_{L \subseteq K}\g{Z}(R(L)\mathbf{v}_B).
    \end{equation}
    We can write the transformation corresponding to a specific $L \subseteq K$ satisfying $F_1(L) \subseteq B$ and $F_2(L)\cap B = \emptyset$ as $\logg{X}(B) \mapsto \logg{X}(B)\logg{Z}({B'})$ and $\logg{X}(B') \mapsto \logg{X}(B')\logg{Z}({B})$, where $B' = \left(B^\up{c}\cup F_1(L)\right)\setminus F_2(L)$ for any $L$. 
\end{enumerate}
Combining two subcases, we find that if $|K|=m/2-1$, $\g{U_{PH}}\left(Q(K)\right)$ induces $\logg{C_{11}Z}(B,B')$ between any unordered pair of logical qubits $\{B,B'\}$ satisfying $F_1(L) \subseteq B$, $F_2(L)\cap B = \emptyset$, and $B' = \left(B^\up{c}\cup F_1(L)\right)\setminus F_2(L)$ for all $L \subseteq K$ such that $|L|\leq m/2-2$, and induces $\logg{C_{00}Z}(D,D')$ between the unordered pair of logical qubits $\{D,D'\}$ satisfying $F_1(K) \subseteq D$, $F_2(K)\cap D = \emptyset$, and $D' = \left(D^\up{c}\cup F_1(K)\right)\setminus F_2(K)$.

Last, let us consider the case that $|K| = m/2$. In this case, $c=1$ if $|B\cap F_1(K)|=m/2$ and there is no $(i,j) \in K$ such that $i,j \in B$, or $c=2$ if $|B\cap F_1(K)|=m/2-1$ and there is no $(i,j) \in K$ such that $i,j \in B$, or $c$ is divisible by 4 otherwise (by \cref{lem:phase_from_QK}). We find the following.
\begin{enumerate}
    \item The only logical qubit $B$ that satisfies the conditions for $c=1$ is $B=F_1(K)$. For this logical qubit, 
    \begin{equation}
        \logg{X}(B)=\g{X}(\mathbf{v}_B) \mapsto (\up{i})\g{X}(\mathbf{v}_B)\prod_{L \subseteq K}\g{Z}(R(L)\mathbf{v}_B).
    \end{equation}
    We can write the transformation corresponding to a specific $L \subseteq K$ satisfying $F_1(L) \subseteq B$ and $F_2(L)\cap B = \emptyset$ as follows: 
    \begin{enumerate}
        \item For $L = K$ ($|L|=m/2$), $\logg{X}(B) \mapsto (\up{i})^{m+1}\logg{X}(B)\logg{Z}({B})$;
        \item For $L$ such that $|L|= m/2-1$, $\logg{X}(B) \mapsto (-1)\logg{X}(B)\logg{Z}({B'})$ and $\logg{X}(B') \mapsto (-1)\logg{X}(B')\logg{Z}({B})$, where $B' = \left(B^\up{c}\cup F_1(L)\right)\setminus F_2(L)$;
        \item For $L$ such that $|L|\leq m/2-2$, $\logg{X}(B) \mapsto \logg{X}(B)\logg{Z}({B'})$ and $\logg{X}(B') \mapsto \logg{X}(B')\logg{Z}({B})$ where $B' = \left(B^\up{c}\cup F_1(L)\right)\setminus F_2(L)$.
    \end{enumerate}
    Consider $L \subseteq K$ that satisfies $F_1(L) \subseteq B$ and $F_2(L)\cap B = \emptyset$. There is only one $L$ of size $|L|=m/2$, and there are $m/2$ $L$'s of size $|L| = m/2-1$. So the overall phase from the combined transformation of $\logg{X}(B)$ is $(-1)^{m/2}(\up{i})^{m+1} = (\up{i})^{2m+1} = \up{i}$ as expected.
    \item A logical qubit $B$ that satisfies the conditions for $c=2$ is of the form $B = (F_1(K)\cup \{j\})\setminus\{i\}$ for some $(i,j) \in K$. For such a logical qubit, 
    \begin{equation}
        \logg{X}(B)=\g{X}(\mathbf{v}_B) \mapsto (-1)\g{X}(\mathbf{v}_B)\prod_{L \subseteq K}\g{Z}(R(L)\mathbf{v}_B).
    \end{equation}
    We can write the transformation corresponding to a specific $L \subseteq K$ satisfying $F_1(L) \subseteq B$ and $F_2(L)\cap B = \emptyset$ as follows: 
    \begin{enumerate}
        \item For $L$ such that $|L|= m/2-1$, $\logg{X}(B) \mapsto (-1)\logg{X}(B)\logg{Z}({B'})$ and $\logg{X}(B') \mapsto (-1)\logg{X}(B')\logg{Z}({B})$, where $B' = \left(B^\up{c}\cup F_1(L)\right)\setminus F_2(L)$;
        \item For $L$ such that $|L|\leq m/2-2$, $\logg{X}(B) \mapsto \logg{X}(B)\logg{Z}({B'})$ and $\logg{X}(B') \mapsto \logg{X}(B')\logg{Z}({B})$, where $B' = \left(B^\up{c}\cup F_1(L)\right)\setminus F_2(L)$.
    \end{enumerate}
    Consider $L \subseteq K$ that satisfies $F_1(L) \subseteq B$ and $F_2(L)\cap B = \emptyset$. There is only one $L$ of size $|L| = m/2-1$, which is $L$ such that $F_1(L)=B \cap F_1(K)$, so the overall phase from the combined transformation of $\logg{X}(B)$ is $-1$ as expected. 
    \item For any logical qubit $B$ that does not satisfy the conditions for $c=1$ or $c=2$ (i.e., $|B\cap F_1(K)|\leq m/2-2$ or there is $(i,j) \in K$ such that $i,j \in B$), we find that
    \begin{equation}
        \logg{X}(B)=\g{X}(\mathbf{v}_B) \mapsto \g{X}(\mathbf{v}_B)\prod_{L \subseteq K}\g{Z}(R(L)\mathbf{v}_B).
    \end{equation}
    We can write the transformation corresponding to a specific $L \subseteq K$ satisfying $F_1(L) \subseteq B$ and $F_2(L)\cap B = \emptyset$ as $\logg{X}(B) \mapsto \logg{X}(B)\logg{Z}({B'})$ and $\logg{X}(B') \mapsto \logg{X}(B')\logg{Z}({B})$, where $B' = \left(B^\up{c}\cup F_1(L)\right)\setminus F_2(L)$ for any $L$.
\end{enumerate}
Combining all three subcases, we find that if $|K|=m/2$, $\g{U_{PH}}\left(Q(K)\right)$ induces $\logg{C_{11}Z}(B,B')$ between any unordered pair of logical qubits $\{B,B'\}$ satisfying $F_1(L) \subseteq B$, $F_2(L)\cap B = \emptyset$, and $B' = \left(B^\up{c}\cup F_1(L)\right)\setminus F_2(L)$ for all $L \subseteq K$ such that $|L|\leq m/2-2$, induces $\logg{C_{00}Z}(D,D')$ between any unordered pair of logical qubits $\{D,D'\}$ satisfying $F_1(L) \subseteq D$, $F_2(L)\cap D = \emptyset$, and $D' = \left(D^\up{c}\cup F_1(L)\right)\setminus F_2(L)$ for all $L \subseteq K$ such that $|L|= m/2-1$, and induces $\logg{S}(F_1(K))$ if $m/2$ is even (or induces $\logg{S}^{\dagger}(F_1(K))$ if $m/2$ is odd). This completes the proof. \qed

\subsection{Proof of \cref{thm:fold_product}.} 

We aim to find the overall logical operation of the product $\prod_{L \subseteq K}\g{U_{PH}}\left(Q(L)\right)$. First, consider each $L \subseteq K$. By \cref{thm:logical_from_fold}, $\g{U_{PH}}\left(Q(L)\right)$ implements a product of logical gates of the form $\logg{S}(B)$, $\logg{C_{00}Z}(B,B')$, or $\logg{C_{11}Z}(B,B')$; each gate is applied to a logical qubit $B$ or between logical qubits $B,B'$ that satisfy certain conditions. We can associate each of $\logg{S}(B)$, $\logg{C_{00}Z}(B,B')$, or $\logg{C_{11}Z}(B,B')$ that arises from $\g{U_{PH}}\left(Q(L)\right)$ with its corresponding subset $M\subseteq L$ (multiple logical gates can be associated with the same subset $M$).

Next, suppose that there exist subsets $L_1,L_2,M \subseteq K$ such that $L_1 \neq L_2$, $M \subseteq L_1$, and $M \subseteq L_2$. If $M \subseteq L_1$ leads to a logical gate $\logg{G}$ that arises from $\g{U_{PH}}\left(Q(L_1)\right)$, the same subset $M \subseteq L_2$ also leads to the same logical gate $\logg{G}$ that arises from $\g{U_{PH}}\left(Q(L_2)\right)$; i.e., the logical gate $\logg{G}$ that corresponds to a certain $M \subseteq K$ can appear from multiple terms in the product $\prod_{L \subseteq K}\g{U_{PH}}\left(Q(L)\right)$. Note that $\logg{G}$ can be $\logg{C_{11}Z}(B,B')$, $\logg{C_{00}Z}(B,B')$, or $\logg{S}(B)$, and for any $B,B'$, $\logg{C_{11}Z}(B,B')^2 = \logg{C_{00}Z}(B,B')^2 = \logg{I}$. Here we want to show that any logical gate $\logg{G}$ associated with $M \subsetneq K$ always arises from an even number of terms in $\prod_{L \subseteq K}\g{U_{PH}}\left(Q(L)\right)$, so such $\logg{G}$ disappears from the overall logical operation of $\prod_{L \subseteq K}\g{U_{PH}}\left(Q(L)\right)$, and the remaining logical gate is the gate associated with $M=K$.

We will show that for any $M \subseteq K$, the number of subsets $L\subseteq K$ such that $M \subseteq L$ is $2^{|K\setminus M|}$; Suppose that $K = \{a_1,\dots,a_{|K|}\}$ and $M = \{a_1,\dots,a_{|M|}\}$. We have that $K\setminus M = \{a_{|M|+1},\dots,a_{|K|}\}$. Any $L\subseteq K$ such that $M \subseteq L$ is of the form $L = M\cup A$ where $A$ is a subset of $K\setminus M$. The number of possible subsets of $K\setminus M$ is $2^{|K\setminus M|}$, so the number of $L\subseteq K$ such that $M \subseteq L$ is also $2^{|K\setminus M|}$. Note that $2^{|K\setminus M|}$ is even if $M\subsetneq K$, and it is odd if $M=K$.

Consider the overall logical operation of $\prod_{L \subseteq K}\g{U_{PH}}\left(Q(L)\right)$. For any size of $K$, any logical gate associated with $M \subsetneq K$ (which is either $\logg{C_{11}Z}$ or $\logg{C_{00}Z}$) arises from an even number of terms in $\prod_{L \subseteq K}\g{U_{PH}}\left(Q(L)\right)$, so it disappears. The remaining logical gates are those associated with $M=K$, which are as follows:
\begin{enumerate}
    \item If $|K| \leq m/2-2$, the remaining logical gates consist of $\logg{C_{11}Z}(B,B')$ applied to any unordered pair of logical qubits $\{B,B'\}$ satisfying $F_1(K) \subseteq B$, $F_2(K)\cap B = \emptyset$, and $B' =  \left(B^\up{c}\cup F_1(K)\right)\setminus F_2(K)$.
    \item If $|K| = m/2-1$, the remaining logical gate is $\logg{C_{00}Z}(B,B')$ applied to the logical qubits $B,B'$ satisfying $F_1(K) \subseteq B$, $F_2(K)\cap B = \emptyset$, and $B' =  \left(B^\up{c}\cup F_1(K)\right)\setminus F_2(K)$.
    \item If $|K| = m/2$, the remaining logical gate is $\logg{S}(F_1(K))$ if $m/2$ is even (or $\logg{S}^{\dagger}(F_1(K))$ if $m/2$ is odd). \qed
\end{enumerate}

\newpage 

\section{Stabilizer generators and logical Pauli operators of $\mathrm{QRM}(6)$} \label{sec:ex_m6}

In this section, stabilizer generators and logical operators of the quantum Reed--Muller code $\mathrm{QRM}(6)$ are provided. In particular,
\begin{enumerate}
    \item A choice of stabilizer generators of $\mathrm{QRM}(6)$ according to \cref{prop:QRM_stb} is provided in \cref{tab:all_g_from_vA_m6}. This choice corresponds to the generators $\{\mathbf{v}_A: 0\leq|A|\leq 2\}$ of the classical Reed--Muller code $\mathrm{RM}(2,6)$.
    \item A choice of minimum-weight stabilizer generators of $\mathrm{QRM}(6)$ constructed from \cref{thm:reduced_weight} is provided in \cref{tab:all_g_from_hA_m6}. For this choice, we let $\mathbf{h}_A = \mathbf{v}_A$ when $|A|=2$, $\mathbf{h}_{i}=\mathbf{v}_{i}\left(\mathbf{1}+\mathbf{v}_{6}\right)$ for $i\in\{1,\dots,5\}$ and $\mathbf{h}_{6}=\mathbf{v}_{6}\left(\mathbf{1}+\mathbf{v}_{5}\right)$, and $\mathbf{h}_\emptyset=\left(\mathbf{1}+\mathbf{v}_{5}\right)\left(\mathbf{1}+\mathbf{v}_{6}\right)$. Note that $\{\mathbf{h}_A: 0\leq|A|\leq 2\}$ also generates $\mathrm{RM}(2,6)$, and all stabilizer generators in this choice have weight $2\sqrt{n}=16$.
    \item Logical Pauli operators of $\mathrm{QRM}(6)$ according to \cref{def:log_def,def:canonical_indices} are provided in \cref{tab:logical_m=6}.
\end{enumerate}

\newpage

\begin{landscape}
\vspace*{\fill}
    
\begin{table}[!htbp]
    \centering       
    \resizebox{\linewidth}{!}{
    \begin{tabular}{| c | c | K{2.2mm} | K{2.2mm} | K{2.2mm} | K{2.2mm} | K{2.2mm} | K{2.2mm} | K{2.2mm} | K{2.2mm} | K{2.2mm} | K{2.2mm} | K{2.2mm} | K{2.2mm} | K{2.2mm} | K{2.2mm} | K{2.2mm} | K{2.2mm} | K{2.2mm} | K{2.2mm} | K{2.2mm} | K{2.2mm} | K{2.2mm} | K{2.2mm} | K{2.2mm} | K{2.2mm} | K{2.2mm} | K{2.2mm} | K{2.2mm} | K{2.2mm} | K{2.2mm} | K{2.2mm} | K{2.2mm} | K{2.2mm} | K{2.2mm} | K{2.2mm} | K{2.2mm} | K{2.2mm} | K{2.2mm} | K{2.2mm} | K{2.2mm} | K{2.2mm} | K{2.2mm} | K{2.2mm} | K{2.2mm} | K{2.2mm} | K{2.2mm} | K{2.2mm} | K{2.2mm} | K{2.2mm} | K{2.2mm} | K{2.2mm} | K{2.2mm} | K{2.2mm} | K{2.2mm} | K{2.2mm} | K{2.2mm} | K{2.2mm} | K{2.2mm} | K{2.2mm} | K{2.2mm} | K{2.2mm} | K{2.2mm} | K{2.2mm} | K{2.2mm} | K{2.2mm} |}
         \hline
         \parbox[t]{3.5mm}{\multirow{24}{*}{\rotatebox[origin=c]{90}{\textrm{$\g{X}$-type}}}} & $A$ & \multicolumn{64}{|c|}{$g_x(A)$} \\
         \cline{2-66}
         & $\emptyset$ & $\g{X}$ & $\g{X}$ & $\g{X}$ & $\g{X}$ & $\g{X}$ & $\g{X}$ & $\g{X}$ & $\g{X}$ & $\g{X}$ & $\g{X}$ & $\g{X}$ & $\g{X}$ & $\g{X}$ & $\g{X}$ & $\g{X}$ & $\g{X}$ &
         $\g{X}$ & $\g{X}$ & $\g{X}$ & $\g{X}$ & $\g{X}$ & $\g{X}$ & $\g{X}$ & $\g{X}$ & $\g{X}$ & $\g{X}$ & $\g{X}$ & $\g{X}$ & $\g{X}$ & $\g{X}$ & $\g{X}$ & $\g{X}$ &
         $\g{X}$ & $\g{X}$ & $\g{X}$ & $\g{X}$ & $\g{X}$ & $\g{X}$ & $\g{X}$ & $\g{X}$ & $\g{X}$ & $\g{X}$ & $\g{X}$ & $\g{X}$ & $\g{X}$ & $\g{X}$ & $\g{X}$ & $\g{X}$ &
         $\g{X}$ & $\g{X}$ & $\g{X}$ & $\g{X}$ & $\g{X}$ & $\g{X}$ & $\g{X}$ & $\g{X}$ & $\g{X}$ & $\g{X}$ & $\g{X}$ & $\g{X}$ & $\g{X}$ & $\g{X}$ & $\g{X}$ & $\g{X}$ \\
         \cline{2-66}
         & $\{1\}$ & $\g{I}$ & $\g{X}$ & $\g{I}$ & $\g{X}$ & $\g{I}$ & $\g{X}$ & $\g{I}$ & $\g{X}$ & $\g{I}$ & $\g{X}$ & $\g{I}$ & $\g{X}$ & $\g{I}$ & $\g{X}$ & $\g{I}$ & $\g{X}$ &
         $\g{I}$ & $\g{X}$ & $\g{I}$ & $\g{X}$ & $\g{I}$ & $\g{X}$ & $\g{I}$ & $\g{X}$ & $\g{I}$ & $\g{X}$ & $\g{I}$ & $\g{X}$ & $\g{I}$ & $\g{X}$ & $\g{I}$ & $\g{X}$ &
         $\g{I}$ & $\g{X}$ & $\g{I}$ & $\g{X}$ & $\g{I}$ & $\g{X}$ & $\g{I}$ & $\g{X}$ & $\g{I}$ & $\g{X}$ & $\g{I}$ & $\g{X}$ & $\g{I}$ & $\g{X}$ & $\g{I}$ & $\g{X}$ &
         $\g{I}$ & $\g{X}$ & $\g{I}$ & $\g{X}$ & $\g{I}$ & $\g{X}$ & $\g{I}$ & $\g{X}$ & $\g{I}$ & $\g{X}$ & $\g{I}$ & $\g{X}$ & $\g{I}$ & $\g{X}$ & $\g{I}$ & $\g{X}$ \\
         \cline{2-66}
         & $\{2\}$ & $\g{I}$ & $\g{I}$ & $\g{X}$ & $\g{X}$ & $\g{I}$ & $\g{I}$ & $\g{X}$ & $\g{X}$ & $\g{I}$ & $\g{I}$ & $\g{X}$ & $\g{X}$ & $\g{I}$ & $\g{I}$ & $\g{X}$ & $\g{X}$ &
         $\g{I}$ & $\g{I}$ & $\g{X}$ & $\g{X}$ & $\g{I}$ & $\g{I}$ & $\g{X}$ & $\g{X}$ & $\g{I}$ & $\g{I}$ & $\g{X}$ & $\g{X}$ & $\g{I}$ & $\g{I}$ & $\g{X}$ & $\g{X}$ &
         $\g{I}$ & $\g{I}$ & $\g{X}$ & $\g{X}$ & $\g{I}$ & $\g{I}$ & $\g{X}$ & $\g{X}$ & $\g{I}$ & $\g{I}$ & $\g{X}$ & $\g{X}$ & $\g{I}$ & $\g{I}$ & $\g{X}$ & $\g{X}$ &
         $\g{I}$ & $\g{I}$ & $\g{X}$ & $\g{X}$ & $\g{I}$ & $\g{I}$ & $\g{X}$ & $\g{X}$ & $\g{I}$ & $\g{I}$ & $\g{X}$ & $\g{X}$ & $\g{I}$ & $\g{I}$ & $\g{X}$ & $\g{X}$ \\
         \cline{2-66}
         & $\{3\}$ & $\g{I}$ & $\g{I}$ & $\g{I}$ & $\g{I}$ & $\g{X}$ & $\g{X}$ & $\g{X}$ & $\g{X}$ & $\g{I}$ & $\g{I}$ & $\g{I}$ & $\g{I}$ & $\g{X}$ & $\g{X}$ & $\g{X}$ & $\g{X}$ &
         $\g{I}$ & $\g{I}$ & $\g{I}$ & $\g{I}$ & $\g{X}$ & $\g{X}$ & $\g{X}$ & $\g{X}$ & $\g{I}$ & $\g{I}$ & $\g{I}$ & $\g{I}$ & $\g{X}$ & $\g{X}$ & $\g{X}$ & $\g{X}$ &
         $\g{I}$ & $\g{I}$ & $\g{I}$ & $\g{I}$ & $\g{X}$ & $\g{X}$ & $\g{X}$ & $\g{X}$ & $\g{I}$ & $\g{I}$ & $\g{I}$ & $\g{I}$ & $\g{X}$ & $\g{X}$ & $\g{X}$ & $\g{X}$ &
         $\g{I}$ & $\g{I}$ & $\g{I}$ & $\g{I}$ & $\g{X}$ & $\g{X}$ & $\g{X}$ & $\g{X}$ & $\g{I}$ & $\g{I}$ & $\g{I}$ & $\g{I}$ & $\g{X}$ & $\g{X}$ & $\g{X}$ & $\g{X}$ \\
         \cline{2-66}
         & $\{4\}$ & $\g{I}$ & $\g{I}$ & $\g{I}$ & $\g{I}$ & $\g{I}$ & $\g{I}$ & $\g{I}$ & $\g{I}$ & $\g{X}$ & $\g{X}$ & $\g{X}$ & $\g{X}$ & $\g{X}$ & $\g{X}$ & $\g{X}$ & $\g{X}$ &
         $\g{I}$ & $\g{I}$ & $\g{I}$ & $\g{I}$ & $\g{I}$ & $\g{I}$ & $\g{I}$ & $\g{I}$ & $\g{X}$ & $\g{X}$ & $\g{X}$ & $\g{X}$ & $\g{X}$ & $\g{X}$ & $\g{X}$ & $\g{X}$ &
         $\g{I}$ & $\g{I}$ & $\g{I}$ & $\g{I}$ & $\g{I}$ & $\g{I}$ & $\g{I}$ & $\g{I}$ & $\g{X}$ & $\g{X}$ & $\g{X}$ & $\g{X}$ & $\g{X}$ & $\g{X}$ & $\g{X}$ & $\g{X}$ &
         $\g{I}$ & $\g{I}$ & $\g{I}$ & $\g{I}$ & $\g{I}$ & $\g{I}$ & $\g{I}$ & $\g{I}$ & $\g{X}$ & $\g{X}$ & $\g{X}$ & $\g{X}$ & $\g{X}$ & $\g{X}$ & $\g{X}$ & $\g{X}$ \\
         \cline{2-66}
         & $\{5\}$ & $\g{I}$ & $\g{I}$ & $\g{I}$ & $\g{I}$ & $\g{I}$ & $\g{I}$ & $\g{I}$ & $\g{I}$ & $\g{I}$ & $\g{I}$ & $\g{I}$ & $\g{I}$ & $\g{I}$ & $\g{I}$ & $\g{I}$ & $\g{I}$ &
         $\g{X}$ & $\g{X}$ & $\g{X}$ & $\g{X}$ & $\g{X}$ & $\g{X}$ & $\g{X}$ & $\g{X}$ & $\g{X}$ & $\g{X}$ & $\g{X}$ & $\g{X}$ & $\g{X}$ & $\g{X}$ & $\g{X}$ & $\g{X}$ &
         $\g{I}$ & $\g{I}$ & $\g{I}$ & $\g{I}$ & $\g{I}$ & $\g{I}$ & $\g{I}$ & $\g{I}$ & $\g{I}$ & $\g{I}$ & $\g{I}$ & $\g{I}$ & $\g{I}$ & $\g{I}$ & $\g{I}$ & $\g{I}$ &
         $\g{X}$ & $\g{X}$ & $\g{X}$ & $\g{X}$ & $\g{X}$ & $\g{X}$ & $\g{X}$ & $\g{X}$ & $\g{X}$ & $\g{X}$ & $\g{X}$ & $\g{X}$ & $\g{X}$ & $\g{X}$ & $\g{X}$ & $\g{X}$ \\
         \cline{2-66}
         & $\{6\}$ & $\g{I}$ & $\g{I}$ & $\g{I}$ & $\g{I}$ & $\g{I}$ & $\g{I}$ & $\g{I}$ & $\g{I}$ & $\g{I}$ & $\g{I}$ & $\g{I}$ & $\g{I}$ & $\g{I}$ & $\g{I}$ & $\g{I}$ & $\g{I}$ & $\g{I}$ & $\g{I}$ & $\g{I}$ & $\g{I}$ & $\g{I}$ & $\g{I}$ & $\g{I}$ & $\g{I}$ & $\g{I}$ & $\g{I}$ & $\g{I}$ & $\g{I}$ & $\g{I}$ & $\g{I}$ & $\g{I}$ & $\g{I}$ &
         $\g{X}$ & $\g{X}$ & $\g{X}$ & $\g{X}$ & $\g{X}$ & $\g{X}$ & $\g{X}$ & $\g{X}$ & $\g{X}$ & $\g{X}$ & $\g{X}$ & $\g{X}$ & $\g{X}$ & $\g{X}$ & $\g{X}$ & $\g{X}$ & $\g{X}$ & $\g{X}$ & $\g{X}$ & $\g{X}$ & $\g{X}$ & $\g{X}$ & $\g{X}$ & $\g{X}$ & $\g{X}$ & $\g{X}$ & $\g{X}$ & $\g{X}$ & $\g{X}$ & $\g{X}$ & $\g{X}$ & $\g{X}$ \\
         \cline{2-66}
         & $\{1,2\}$ & $\g{I}$ & $\g{I}$ & $\g{I}$ & $\g{X}$ & $\g{I}$ & $\g{I}$ & $\g{I}$ & $\g{X}$ & $\g{I}$ & $\g{I}$ & $\g{I}$ & $\g{X}$ & $\g{I}$ & $\g{I}$ & $\g{I}$ & $\g{X}$ & $\g{I}$ & $\g{I}$ & $\g{I}$ & $\g{X}$ & $\g{I}$ & $\g{I}$ & $\g{I}$ & $\g{X}$ & $\g{I}$ & $\g{I}$ & $\g{I}$ & $\g{X}$ & $\g{I}$ & $\g{I}$ & $\g{I}$ & $\g{X}$ & $\g{I}$ & $\g{I}$ & $\g{I}$ & $\g{X}$ & $\g{I}$ & $\g{I}$ & $\g{I}$ & $\g{X}$ & $\g{I}$ & $\g{I}$ & $\g{I}$ & $\g{X}$ & $\g{I}$ & $\g{I}$ & $\g{I}$ & $\g{X}$ & $\g{I}$ & $\g{I}$ & $\g{I}$ & $\g{X}$ & $\g{I}$ & $\g{I}$ & $\g{I}$ & $\g{X}$ & $\g{I}$ & $\g{I}$ & $\g{I}$ & $\g{X}$ & $\g{I}$ & $\g{I}$ & $\g{I}$ & $\g{X}$ \\
         \cline{2-66}
         & $\{1,3\}$ & $\g{I}$ & $\g{I}$ & $\g{I}$ & $\g{I}$ & $\g{I}$ & $\g{X}$ & $\g{I}$ & $\g{X}$ & $\g{I}$ & $\g{I}$ & $\g{I}$ & $\g{I}$ & $\g{I}$ & $\g{X}$ & $\g{I}$ & $\g{X}$ & $\g{I}$ & $\g{I}$ & $\g{I}$ & $\g{I}$ & $\g{I}$ & $\g{X}$ & $\g{I}$ & $\g{X}$ & $\g{I}$ & $\g{I}$ & $\g{I}$ & $\g{I}$ & $\g{I}$ & $\g{X}$ & $\g{I}$ & $\g{X}$ & $\g{I}$ & $\g{I}$ & $\g{I}$ & $\g{I}$ & $\g{I}$ & $\g{X}$ & $\g{I}$ & $\g{X}$ & $\g{I}$ & $\g{I}$ & $\g{I}$ & $\g{I}$ & $\g{I}$ & $\g{X}$ & $\g{I}$ & $\g{X}$ & $\g{I}$ & $\g{I}$ & $\g{I}$ & $\g{I}$ & $\g{I}$ & $\g{X}$ & $\g{I}$ & $\g{X}$ & $\g{I}$ & $\g{I}$ & $\g{I}$ & $\g{I}$ & $\g{I}$ & $\g{X}$ & $\g{I}$ & $\g{X}$ \\
         \cline{2-66}
         & $\{1,4\}$ & $\g{I}$ & $\g{I}$ & $\g{I}$ & $\g{I}$ & $\g{I}$ & $\g{I}$ & $\g{I}$ & $\g{I}$ & $\g{I}$ & $\g{X}$ & $\g{I}$ & $\g{X}$ & $\g{I}$ & $\g{X}$ & $\g{I}$ & $\g{X}$ & $\g{I}$ & $\g{I}$ & $\g{I}$ & $\g{I}$ & $\g{I}$ & $\g{I}$ & $\g{I}$ & $\g{I}$ & $\g{I}$ & $\g{X}$ & $\g{I}$ & $\g{X}$ & $\g{I}$ & $\g{X}$ & $\g{I}$ & $\g{X}$ & $\g{I}$ & $\g{I}$ & $\g{I}$ & $\g{I}$ & $\g{I}$ & $\g{I}$ & $\g{I}$ & $\g{I}$ & $\g{I}$ & $\g{X}$ & $\g{I}$ & $\g{X}$ & $\g{I}$ & $\g{X}$ & $\g{I}$ & $\g{X}$ & $\g{I}$ & $\g{I}$ & $\g{I}$ & $\g{I}$ & $\g{I}$ & $\g{I}$ & $\g{I}$ & $\g{I}$ & $\g{I}$ & $\g{X}$ & $\g{I}$ & $\g{X}$ & $\g{I}$ & $\g{X}$ & $\g{I}$ & $\g{X}$ \\
         \cline{2-66}
         & $\{1,5\}$ & $\g{I}$ & $\g{I}$ & $\g{I}$ & $\g{I}$ & $\g{I}$ & $\g{I}$ & $\g{I}$ & $\g{I}$ & $\g{I}$ & $\g{I}$ & $\g{I}$ & $\g{I}$ & $\g{I}$ & $\g{I}$ & $\g{I}$ & $\g{I}$ & $\g{I}$ & $\g{X}$ & $\g{I}$ & $\g{X}$ & $\g{I}$ & $\g{X}$ & $\g{I}$ & $\g{X}$ & $\g{I}$ & $\g{X}$ & $\g{I}$ & $\g{X}$ & $\g{I}$ & $\g{X}$ & $\g{I}$ & $\g{X}$ & $\g{I}$ & $\g{I}$ & $\g{I}$ & $\g{I}$ & $\g{I}$ & $\g{I}$ & $\g{I}$ & $\g{I}$ & $\g{I}$ & $\g{I}$ & $\g{I}$ & $\g{I}$ & $\g{I}$ & $\g{I}$ & $\g{I}$ & $\g{I}$ & $\g{I}$ & $\g{X}$ & $\g{I}$ & $\g{X}$ & $\g{I}$ & $\g{X}$ & $\g{I}$ & $\g{X}$ & $\g{I}$ & $\g{X}$ & $\g{I}$ & $\g{X}$ & $\g{I}$ & $\g{X}$ & $\g{I}$ & $\g{X}$ \\
         \cline{2-66}
         & $\{1,6\}$ & $\g{I}$ & $\g{I}$ & $\g{I}$ & $\g{I}$ & $\g{I}$ & $\g{I}$ & $\g{I}$ & $\g{I}$ & $\g{I}$ & $\g{I}$ & $\g{I}$ & $\g{I}$ & $\g{I}$ & $\g{I}$ & $\g{I}$ & $\g{I}$ & $\g{I}$ & $\g{I}$ & $\g{I}$ & $\g{I}$ & $\g{I}$ & $\g{I}$ & $\g{I}$ & $\g{I}$ & $\g{I}$ & $\g{I}$ & $\g{I}$ & $\g{I}$ & $\g{I}$ & $\g{I}$ & $\g{I}$ & $\g{I}$ & $\g{I}$ & $\g{X}$ & $\g{I}$ & $\g{X}$ & $\g{I}$ & $\g{X}$ & $\g{I}$ & $\g{X}$ & $\g{I}$ & $\g{X}$ & $\g{I}$ & $\g{X}$ & $\g{I}$ & $\g{X}$ & $\g{I}$ & $\g{X}$ & $\g{I}$ & $\g{X}$ & $\g{I}$ & $\g{X}$ & $\g{I}$ & $\g{X}$ & $\g{I}$ & $\g{X}$ & $\g{I}$ & $\g{X}$ & $\g{I}$ & $\g{X}$ & $\g{I}$ & $\g{X}$ & $\g{I}$ & $\g{X}$ \\
         \cline{2-66}
         & $\{2,3\}$ & $\g{I}$ & $\g{I}$ & $\g{I}$ & $\g{I}$ & $\g{I}$ & $\g{I}$ & $\g{X}$ & $\g{X}$ & $\g{I}$ & $\g{I}$ & $\g{I}$ & $\g{I}$ & $\g{I}$ & $\g{I}$ & $\g{X}$ & $\g{X}$ & $\g{I}$ & $\g{I}$ & $\g{I}$ & $\g{I}$ & $\g{I}$ & $\g{I}$ & $\g{X}$ & $\g{X}$ & $\g{I}$ & $\g{I}$ & $\g{I}$ & $\g{I}$ & $\g{I}$ & $\g{I}$ & $\g{X}$ & $\g{X}$ & $\g{I}$ & $\g{I}$ & $\g{I}$ & $\g{I}$ & $\g{I}$ & $\g{I}$ & $\g{X}$ & $\g{X}$ & $\g{I}$ & $\g{I}$ & $\g{I}$ & $\g{I}$ & $\g{I}$ & $\g{I}$ & $\g{X}$ & $\g{X}$ & $\g{I}$ & $\g{I}$ & $\g{I}$ & $\g{I}$ & $\g{I}$ & $\g{I}$ & $\g{X}$ & $\g{X}$ & $\g{I}$ & $\g{I}$ & $\g{I}$ & $\g{I}$ & $\g{I}$ & $\g{I}$ & $\g{X}$ & $\g{X}$ \\
         \cline{2-66}
         & $\{2,4\}$ & $\g{I}$ & $\g{I}$ & $\g{I}$ & $\g{I}$ & $\g{I}$ & $\g{I}$ & $\g{I}$ & $\g{I}$ & $\g{I}$ & $\g{I}$ & $\g{X}$ & $\g{X}$ & $\g{I}$ & $\g{I}$ & $\g{X}$ & $\g{X}$ & $\g{I}$ & $\g{I}$ & $\g{I}$ & $\g{I}$ & $\g{I}$ & $\g{I}$ & $\g{I}$ & $\g{I}$ & $\g{I}$ & $\g{I}$ & $\g{X}$ & $\g{X}$ & $\g{I}$ & $\g{I}$ & $\g{X}$ & $\g{X}$ & $\g{I}$ & $\g{I}$ & $\g{I}$ & $\g{I}$ & $\g{I}$ & $\g{I}$ & $\g{I}$ & $\g{I}$ & $\g{I}$ & $\g{I}$ & $\g{X}$ & $\g{X}$ & $\g{I}$ & $\g{I}$ & $\g{X}$ & $\g{X}$ & $\g{I}$ & $\g{I}$ & $\g{I}$ & $\g{I}$ & $\g{I}$ & $\g{I}$ & $\g{I}$ & $\g{I}$ & $\g{I}$ & $\g{I}$ & $\g{X}$ & $\g{X}$ & $\g{I}$ & $\g{I}$ & $\g{X}$ & $\g{X}$ \\
         \cline{2-66}
         & $\{2,5\}$ & $\g{I}$ & $\g{I}$ & $\g{I}$ & $\g{I}$ & $\g{I}$ & $\g{I}$ & $\g{I}$ & $\g{I}$ & $\g{I}$ & $\g{I}$ & $\g{I}$ & $\g{I}$ & $\g{I}$ & $\g{I}$ & $\g{I}$ & $\g{I}$ & $\g{I}$ & $\g{I}$ & $\g{X}$ & $\g{X}$ & $\g{I}$ & $\g{I}$ & $\g{X}$ & $\g{X}$ & $\g{I}$ & $\g{I}$ & $\g{X}$ & $\g{X}$ & $\g{I}$ & $\g{I}$ & $\g{X}$ & $\g{X}$ & $\g{I}$ & $\g{I}$ & $\g{I}$ & $\g{I}$ & $\g{I}$ & $\g{I}$ & $\g{I}$ & $\g{I}$ & $\g{I}$ & $\g{I}$ & $\g{I}$ & $\g{I}$ & $\g{I}$ & $\g{I}$ & $\g{I}$ & $\g{I}$ & $\g{I}$ & $\g{I}$ & $\g{X}$ & $\g{X}$ & $\g{I}$ & $\g{I}$ & $\g{X}$ & $\g{X}$ & $\g{I}$ & $\g{I}$ & $\g{X}$ & $\g{X}$ & $\g{I}$ & $\g{I}$ & $\g{X}$ & $\g{X}$ \\
         \cline{2-66}
         & $\{2,6\}$ & $\g{I}$ & $\g{I}$ & $\g{I}$ & $\g{I}$ & $\g{I}$ & $\g{I}$ & $\g{I}$ & $\g{I}$ & $\g{I}$ & $\g{I}$ & $\g{I}$ & $\g{I}$ & $\g{I}$ & $\g{I}$ & $\g{I}$ & $\g{I}$ & $\g{I}$ & $\g{I}$ & $\g{I}$ & $\g{I}$ & $\g{I}$ & $\g{I}$ & $\g{I}$ & $\g{I}$ & $\g{I}$ & $\g{I}$ & $\g{I}$ & $\g{I}$ & $\g{I}$ & $\g{I}$ & $\g{I}$ & $\g{I}$ & $\g{I}$ & $\g{I}$ & $\g{X}$ & $\g{X}$ & $\g{I}$ & $\g{I}$ & $\g{X}$ & $\g{X}$ & $\g{I}$ & $\g{I}$ & $\g{X}$ & $\g{X}$ & $\g{I}$ & $\g{I}$ & $\g{X}$ & $\g{X}$ & $\g{I}$ & $\g{I}$ & $\g{X}$ & $\g{X}$ & $\g{I}$ & $\g{I}$ & $\g{X}$ & $\g{X}$ & $\g{I}$ & $\g{I}$ & $\g{X}$ & $\g{X}$ & $\g{I}$ & $\g{I}$ & $\g{X}$ & $\g{X}$ \\
         \cline{2-66}
         & $\{3,4\}$ & $\g{I}$ & $\g{I}$ & $\g{I}$ & $\g{I}$ & $\g{I}$ & $\g{I}$ & $\g{I}$ & $\g{I}$ & $\g{I}$ & $\g{I}$ & $\g{I}$ & $\g{I}$ & $\g{X}$ & $\g{X}$ & $\g{X}$ & $\g{X}$ & $\g{I}$ & $\g{I}$ & $\g{I}$ & $\g{I}$ & $\g{I}$ & $\g{I}$ & $\g{I}$ & $\g{I}$ & $\g{I}$ & $\g{I}$ & $\g{I}$ & $\g{I}$ & $\g{X}$ & $\g{X}$ & $\g{X}$ & $\g{X}$ & $\g{I}$ & $\g{I}$ & $\g{I}$ & $\g{I}$ & $\g{I}$ & $\g{I}$ & $\g{I}$ & $\g{I}$ & $\g{I}$ & $\g{I}$ & $\g{I}$ & $\g{I}$ & $\g{X}$ & $\g{X}$ & $\g{X}$ & $\g{X}$ & $\g{I}$ & $\g{I}$ & $\g{I}$ & $\g{I}$ & $\g{I}$ & $\g{I}$ & $\g{I}$ & $\g{I}$ & $\g{I}$ & $\g{I}$ & $\g{I}$ & $\g{I}$ & $\g{X}$ & $\g{X}$ & $\g{X}$ & $\g{X}$ \\
         \cline{2-66}
         & $\{3,5\}$ & $\g{I}$ & $\g{I}$ & $\g{I}$ & $\g{I}$ & $\g{I}$ & $\g{I}$ & $\g{I}$ & $\g{I}$ & $\g{I}$ & $\g{I}$ & $\g{I}$ & $\g{I}$ & $\g{I}$ & $\g{I}$ & $\g{I}$ & $\g{I}$ & $\g{I}$ & $\g{I}$ & $\g{I}$ & $\g{I}$ & $\g{X}$ & $\g{X}$ & $\g{X}$ & $\g{X}$ & $\g{I}$ & $\g{I}$ & $\g{I}$ & $\g{I}$ & $\g{X}$ & $\g{X}$ & $\g{X}$ & $\g{X}$ & $\g{I}$ & $\g{I}$ & $\g{I}$ & $\g{I}$ & $\g{I}$ & $\g{I}$ & $\g{I}$ & $\g{I}$ & $\g{I}$ & $\g{I}$ & $\g{I}$ & $\g{I}$ & $\g{I}$ & $\g{I}$ & $\g{I}$ & $\g{I}$ & $\g{I}$ & $\g{I}$ & $\g{I}$ & $\g{I}$ & $\g{X}$ & $\g{X}$ & $\g{X}$ & $\g{X}$ & $\g{I}$ & $\g{I}$ & $\g{I}$ & $\g{I}$ & $\g{X}$ & $\g{X}$ & $\g{X}$ & $\g{X}$ \\
         \cline{2-66}
         & $\{3,6\}$ & $\g{I}$ & $\g{I}$ & $\g{I}$ & $\g{I}$ & $\g{I}$ & $\g{I}$ & $\g{I}$ & $\g{I}$ & $\g{I}$ & $\g{I}$ & $\g{I}$ & $\g{I}$ & $\g{I}$ & $\g{I}$ & $\g{I}$ & $\g{I}$ & $\g{I}$ & $\g{I}$ & $\g{I}$ & $\g{I}$ & $\g{I}$ & $\g{I}$ & $\g{I}$ & $\g{I}$ & $\g{I}$ & $\g{I}$ & $\g{I}$ & $\g{I}$ & $\g{I}$ & $\g{I}$ & $\g{I}$ & $\g{I}$ & 
         $\g{I}$ & $\g{I}$ & $\g{I}$ & $\g{I}$ & $\g{X}$ & $\g{X}$ & $\g{X}$ & $\g{X}$ & $\g{I}$ & $\g{I}$ & $\g{I}$ & $\g{I}$ & $\g{X}$ & $\g{X}$ & $\g{X}$ & $\g{X}$ &
         $\g{I}$ & $\g{I}$ & $\g{I}$ & $\g{I}$ & $\g{X}$ & $\g{X}$ & $\g{X}$ & $\g{X}$ & $\g{I}$ & $\g{I}$ & $\g{I}$ & $\g{I}$ & $\g{X}$ & $\g{X}$ & $\g{X}$ & $\g{X}$ \\
         \cline{2-66}
         & $\{4,5\}$ & $\g{I}$ & $\g{I}$ & $\g{I}$ & $\g{I}$ & $\g{I}$ & $\g{I}$ & $\g{I}$ & $\g{I}$ & $\g{I}$ & $\g{I}$ & $\g{I}$ & $\g{I}$ & $\g{I}$ & $\g{I}$ & $\g{I}$ & $\g{I}$ &
         $\g{I}$ & $\g{I}$ & $\g{I}$ & $\g{I}$ & $\g{I}$ & $\g{I}$ & $\g{I}$ & $\g{I}$ & $\g{X}$ & $\g{X}$ & $\g{X}$ & $\g{X}$ & $\g{X}$ & $\g{X}$ & $\g{X}$ & $\g{X}$ &
         $\g{I}$ & $\g{I}$ & $\g{I}$ & $\g{I}$ & $\g{I}$ & $\g{I}$ & $\g{I}$ & $\g{I}$ & $\g{I}$ & $\g{I}$ & $\g{I}$ & $\g{I}$ & $\g{I}$ & $\g{I}$ & $\g{I}$ & $\g{I}$ &
         $\g{I}$ & $\g{I}$ & $\g{I}$ & $\g{I}$ & $\g{I}$ & $\g{I}$ & $\g{I}$ & $\g{I}$ & $\g{X}$ & $\g{X}$ & $\g{X}$ & $\g{X}$ & $\g{X}$ & $\g{X}$ & $\g{X}$ & $\g{X}$ \\
         \cline{2-66}
         & $\{4,6\}$ & $\g{I}$ & $\g{I}$ & $\g{I}$ & $\g{I}$ & $\g{I}$ & $\g{I}$ & $\g{I}$ & $\g{I}$ & $\g{I}$ & $\g{I}$ & $\g{I}$ & $\g{I}$ & $\g{I}$ & $\g{I}$ & $\g{I}$ & $\g{I}$ & $\g{I}$ & $\g{I}$ & $\g{I}$ & $\g{I}$ & $\g{I}$ & $\g{I}$ & $\g{I}$ & $\g{I}$ & $\g{I}$ & $\g{I}$ & $\g{I}$ & $\g{I}$ & $\g{I}$ & $\g{I}$ & $\g{I}$ & $\g{I}$ &
         $\g{I}$ & $\g{I}$ & $\g{I}$ & $\g{I}$ & $\g{I}$ & $\g{I}$ & $\g{I}$ & $\g{I}$ & $\g{X}$ & $\g{X}$ & $\g{X}$ & $\g{X}$ & $\g{X}$ & $\g{X}$ & $\g{X}$ & $\g{X}$ &
         $\g{I}$ & $\g{I}$ & $\g{I}$ & $\g{I}$ & $\g{I}$ & $\g{I}$ & $\g{I}$ & $\g{I}$ & $\g{X}$ & $\g{X}$ & $\g{X}$ & $\g{X}$ & $\g{X}$ & $\g{X}$ & $\g{X}$ & $\g{X}$ \\
         \cline{2-66}
         & $\{5,6\}$ & $\g{I}$ & $\g{I}$ & $\g{I}$ & $\g{I}$ & $\g{I}$ & $\g{I}$ & $\g{I}$ & $\g{I}$ & $\g{I}$ & $\g{I}$ & $\g{I}$ & $\g{I}$ & $\g{I}$ & $\g{I}$ & $\g{I}$ & $\g{I}$ & $\g{I}$ & $\g{I}$ & $\g{I}$ & $\g{I}$ & $\g{I}$ & $\g{I}$ & $\g{I}$ & $\g{I}$ & $\g{I}$ & $\g{I}$ & $\g{I}$ & $\g{I}$ & $\g{I}$ & $\g{I}$ & $\g{I}$ & $\g{I}$ &
         $\g{I}$ & $\g{I}$ & $\g{I}$ & $\g{I}$ & $\g{I}$ & $\g{I}$ & $\g{I}$ & $\g{I}$ & $\g{I}$ & $\g{I}$ & $\g{I}$ & $\g{I}$ & $\g{I}$ & $\g{I}$ & $\g{I}$ & $\g{I}$ &
         $\g{X}$ & $\g{X}$ & $\g{X}$ & $\g{X}$ & $\g{X}$ & $\g{X}$ & $\g{X}$ & $\g{X}$ & $\g{X}$ & $\g{X}$ & $\g{X}$ & $\g{X}$ & $\g{X}$ & $\g{X}$ & $\g{X}$ & $\g{X}$ \\
         \hline
    \end{tabular}
    }
    \caption{A choice of stabilizer generators of $\mathrm{QRM}(6)$ according to \cref{prop:QRM_stb}, which correspond to the generators $\{\mathbf{v}_A: 0\leq|A|\leq 2\}$ of the classical Reed--Muller code $\mathrm{RM}(2,6)$. Here we only display $\g{X}$-type stabilizer generators $g_x(A)$. The $\g{Z}$-type stabilizer generators $g_z(A)$ have the same form except that they are tensor products of $\g{I}$ and $\g{Z}$ instead of $\g{I}$ and $\g{X}$.}
        \label{tab:all_g_from_vA_m6}
\end{table}

\vspace*{\fill}
\end{landscape}

\newpage

\begin{landscape}
\vspace*{\fill}

\begin{table}[!htbp]
    \centering       
    \resizebox{\linewidth}{!}{
    \begin{tabular}{| c | c | K{2.2mm} | K{2.2mm} | K{2.2mm} | K{2.2mm} | K{2.2mm} | K{2.2mm} | K{2.2mm} | K{2.2mm} | K{2.2mm} | K{2.2mm} | K{2.2mm} | K{2.2mm} | K{2.2mm} | K{2.2mm} | K{2.2mm} | K{2.2mm} | K{2.2mm} | K{2.2mm} | K{2.2mm} | K{2.2mm} | K{2.2mm} | K{2.2mm} | K{2.2mm} | K{2.2mm} | K{2.2mm} | K{2.2mm} | K{2.2mm} | K{2.2mm} | K{2.2mm} | K{2.2mm} | K{2.2mm} | K{2.2mm} | K{2.2mm} | K{2.2mm} | K{2.2mm} | K{2.2mm} | K{2.2mm} | K{2.2mm} | K{2.2mm} | K{2.2mm} | K{2.2mm} | K{2.2mm} | K{2.2mm} | K{2.2mm} | K{2.2mm} | K{2.2mm} | K{2.2mm} | K{2.2mm} | K{2.2mm} | K{2.2mm} | K{2.2mm} | K{2.2mm} | K{2.2mm} | K{2.2mm} | K{2.2mm} | K{2.2mm} | K{2.2mm} | K{2.2mm} | K{2.2mm} | K{2.2mm} | K{2.2mm} | K{2.2mm} | K{2.2mm} | K{2.2mm} |}
         \hline
         \parbox[t]{3.5mm}{\multirow{23}{*}{\rotatebox[origin=c]{90}{\textrm{$\g{X}$-type}}}} & $A$ & \multicolumn{64}{|c|}{$h_x(A)$} \\
         \cline{2-66}
         & $\emptyset$ & $\g{X}$ & $\g{X}$ & $\g{X}$ & $\g{X}$ & $\g{X}$ & $\g{X}$ & $\g{X}$ & $\g{X}$ & $\g{X}$ & $\g{X}$ & $\g{X}$ & $\g{X}$ & $\g{X}$ & $\g{X}$ & $\g{X}$ & $\g{X}$ &
         $\g{I}$ & $\g{I}$ & $\g{I}$ & $\g{I}$ & $\g{I}$ & $\g{I}$ & $\g{I}$ & $\g{I}$ & $\g{I}$ & $\g{I}$ & $\g{I}$ & $\g{I}$ & $\g{I}$ & $\g{I}$ & $\g{I}$ & $\g{I}$ &
         $\g{I}$ & $\g{I}$ & $\g{I}$ & $\g{I}$ & $\g{I}$ & $\g{I}$ & $\g{I}$ & $\g{I}$ & $\g{I}$ & $\g{I}$ & $\g{I}$ & $\g{I}$ & $\g{I}$ & $\g{I}$ & $\g{I}$ & $\g{I}$ &
         $\g{I}$ & $\g{I}$ & $\g{I}$ & $\g{I}$ & $\g{I}$ & $\g{I}$ & $\g{I}$ & $\g{I}$ & $\g{I}$ & $\g{I}$ & $\g{I}$ & $\g{I}$ & $\g{I}$ & $\g{I}$ & $\g{I}$ & $\g{I}$ \\
         \cline{2-66}
         & $\{1\}$ & $\g{I}$ & $\g{X}$ & $\g{I}$ & $\g{X}$ & $\g{I}$ & $\g{X}$ & $\g{I}$ & $\g{X}$ & $\g{I}$ & $\g{X}$ & $\g{I}$ & $\g{X}$ & $\g{I}$ & $\g{X}$ & $\g{I}$ & $\g{X}$ &
         $\g{I}$ & $\g{X}$ & $\g{I}$ & $\g{X}$ & $\g{I}$ & $\g{X}$ & $\g{I}$ & $\g{X}$ & $\g{I}$ & $\g{X}$ & $\g{I}$ & $\g{X}$ & $\g{I}$ & $\g{X}$ & $\g{I}$ & $\g{X}$ &
         $\g{I}$ & $\g{I}$ & $\g{I}$ & $\g{I}$ & $\g{I}$ & $\g{I}$ & $\g{I}$ & $\g{I}$ & $\g{I}$ & $\g{I}$ & $\g{I}$ & $\g{I}$ & $\g{I}$ & $\g{I}$ & $\g{I}$ & $\g{I}$ &
         $\g{I}$ & $\g{I}$ & $\g{I}$ & $\g{I}$ & $\g{I}$ & $\g{I}$ & $\g{I}$ & $\g{I}$ & $\g{I}$ & $\g{I}$ & $\g{I}$ & $\g{I}$ & $\g{I}$ & $\g{I}$ & $\g{I}$ & $\g{I}$ \\
         \cline{2-66}
         & $\{2\}$ & $\g{I}$ & $\g{I}$ & $\g{X}$ & $\g{X}$ & $\g{I}$ & $\g{I}$ & $\g{X}$ & $\g{X}$ & $\g{I}$ & $\g{I}$ & $\g{X}$ & $\g{X}$ & $\g{I}$ & $\g{I}$ & $\g{X}$ & $\g{X}$ &
         $\g{I}$ & $\g{I}$ & $\g{X}$ & $\g{X}$ & $\g{I}$ & $\g{I}$ & $\g{X}$ & $\g{X}$ & $\g{I}$ & $\g{I}$ & $\g{X}$ & $\g{X}$ & $\g{I}$ & $\g{I}$ & $\g{X}$ & $\g{X}$ &
         $\g{I}$ & $\g{I}$ & $\g{I}$ & $\g{I}$ & $\g{I}$ & $\g{I}$ & $\g{I}$ & $\g{I}$ & $\g{I}$ & $\g{I}$ & $\g{I}$ & $\g{I}$ & $\g{I}$ & $\g{I}$ & $\g{I}$ & $\g{I}$ &
         $\g{I}$ & $\g{I}$ & $\g{I}$ & $\g{I}$ & $\g{I}$ & $\g{I}$ & $\g{I}$ & $\g{I}$ & $\g{I}$ & $\g{I}$ & $\g{I}$ & $\g{I}$ & $\g{I}$ & $\g{I}$ & $\g{I}$ & $\g{I}$ \\
         \cline{2-66}
         & $\{3\}$ & $\g{I}$ & $\g{I}$ & $\g{I}$ & $\g{I}$ & $\g{X}$ & $\g{X}$ & $\g{X}$ & $\g{X}$ & $\g{I}$ & $\g{I}$ & $\g{I}$ & $\g{I}$ & $\g{X}$ & $\g{X}$ & $\g{X}$ & $\g{X}$ &
         $\g{I}$ & $\g{I}$ & $\g{I}$ & $\g{I}$ & $\g{X}$ & $\g{X}$ & $\g{X}$ & $\g{X}$ & $\g{I}$ & $\g{I}$ & $\g{I}$ & $\g{I}$ & $\g{X}$ & $\g{X}$ & $\g{X}$ & $\g{X}$ &
         $\g{I}$ & $\g{I}$ & $\g{I}$ & $\g{I}$ & $\g{I}$ & $\g{I}$ & $\g{I}$ & $\g{I}$ & $\g{I}$ & $\g{I}$ & $\g{I}$ & $\g{I}$ & $\g{I}$ & $\g{I}$ & $\g{I}$ & $\g{I}$ &
         $\g{I}$ & $\g{I}$ & $\g{I}$ & $\g{I}$ & $\g{I}$ & $\g{I}$ & $\g{I}$ & $\g{I}$ & $\g{I}$ & $\g{I}$ & $\g{I}$ & $\g{I}$ & $\g{I}$ & $\g{I}$ & $\g{I}$ & $\g{I}$ \\
         \cline{2-66}
         & $\{4\}$ & $\g{I}$ & $\g{I}$ & $\g{I}$ & $\g{I}$ & $\g{I}$ & $\g{I}$ & $\g{I}$ & $\g{I}$ & $\g{X}$ & $\g{X}$ & $\g{X}$ & $\g{X}$ & $\g{X}$ & $\g{X}$ & $\g{X}$ & $\g{X}$ &
         $\g{I}$ & $\g{I}$ & $\g{I}$ & $\g{I}$ & $\g{I}$ & $\g{I}$ & $\g{I}$ & $\g{I}$ & $\g{X}$ & $\g{X}$ & $\g{X}$ & $\g{X}$ & $\g{X}$ & $\g{X}$ & $\g{X}$ & $\g{X}$ &
         $\g{I}$ & $\g{I}$ & $\g{I}$ & $\g{I}$ & $\g{I}$ & $\g{I}$ & $\g{I}$ & $\g{I}$ & $\g{I}$ & $\g{I}$ & $\g{I}$ & $\g{I}$ & $\g{I}$ & $\g{I}$ & $\g{I}$ & $\g{I}$ &
         $\g{I}$ & $\g{I}$ & $\g{I}$ & $\g{I}$ & $\g{I}$ & $\g{I}$ & $\g{I}$ & $\g{I}$ & $\g{I}$ & $\g{I}$ & $\g{I}$ & $\g{I}$ & $\g{I}$ & $\g{I}$ & $\g{I}$ & $\g{I}$ \\
         \cline{2-66}
         & $\{5\}$ & $\g{I}$ & $\g{I}$ & $\g{I}$ & $\g{I}$ & $\g{I}$ & $\g{I}$ & $\g{I}$ & $\g{I}$ & $\g{I}$ & $\g{I}$ & $\g{I}$ & $\g{I}$ & $\g{I}$ & $\g{I}$ & $\g{I}$ & $\g{I}$ &
         $\g{X}$ & $\g{X}$ & $\g{X}$ & $\g{X}$ & $\g{X}$ & $\g{X}$ & $\g{X}$ & $\g{X}$ & $\g{X}$ & $\g{X}$ & $\g{X}$ & $\g{X}$ & $\g{X}$ & $\g{X}$ & $\g{X}$ & $\g{X}$ &
         $\g{I}$ & $\g{I}$ & $\g{I}$ & $\g{I}$ & $\g{I}$ & $\g{I}$ & $\g{I}$ & $\g{I}$ & $\g{I}$ & $\g{I}$ & $\g{I}$ & $\g{I}$ & $\g{I}$ & $\g{I}$ & $\g{I}$ & $\g{I}$ &
         $\g{I}$ & $\g{I}$ & $\g{I}$ & $\g{I}$ & $\g{I}$ & $\g{I}$ & $\g{I}$ & $\g{I}$ & $\g{I}$ & $\g{I}$ & $\g{I}$ & $\g{I}$ & $\g{I}$ & $\g{I}$ & $\g{I}$ & $\g{I}$ \\
         \cline{2-66}
         & $\{6\}$ & $\g{I}$ & $\g{I}$ & $\g{I}$ & $\g{I}$ & $\g{I}$ & $\g{I}$ & $\g{I}$ & $\g{I}$ & $\g{I}$ & $\g{I}$ & $\g{I}$ & $\g{I}$ & $\g{I}$ & $\g{I}$ & $\g{I}$ & $\g{I}$ & 
         $\g{I}$ & $\g{I}$ & $\g{I}$ & $\g{I}$ & $\g{I}$ & $\g{I}$ & $\g{I}$ & $\g{I}$ & $\g{I}$ & $\g{I}$ & $\g{I}$ & $\g{I}$ & $\g{I}$ & $\g{I}$ & $\g{I}$ & $\g{I}$ &
         $\g{X}$ & $\g{X}$ & $\g{X}$ & $\g{X}$ & $\g{X}$ & $\g{X}$ & $\g{X}$ & $\g{X}$ & $\g{X}$ & $\g{X}$ & $\g{X}$ & $\g{X}$ & $\g{X}$ & $\g{X}$ & $\g{X}$ & $\g{X}$ & 
         $\g{I}$ & $\g{I}$ & $\g{I}$ & $\g{I}$ & $\g{I}$ & $\g{I}$ & $\g{I}$ & $\g{I}$ & $\g{I}$ & $\g{I}$ & $\g{I}$ & $\g{I}$ & $\g{I}$ & $\g{I}$ & $\g{I}$ & $\g{I}$ \\
         \cline{2-66}
         & $\{1,2\}$ & $\g{I}$ & $\g{I}$ & $\g{I}$ & $\g{X}$ & $\g{I}$ & $\g{I}$ & $\g{I}$ & $\g{X}$ & $\g{I}$ & $\g{I}$ & $\g{I}$ & $\g{X}$ & $\g{I}$ & $\g{I}$ & $\g{I}$ & $\g{X}$ & $\g{I}$ & $\g{I}$ & $\g{I}$ & $\g{X}$ & $\g{I}$ & $\g{I}$ & $\g{I}$ & $\g{X}$ & $\g{I}$ & $\g{I}$ & $\g{I}$ & $\g{X}$ & $\g{I}$ & $\g{I}$ & $\g{I}$ & $\g{X}$ & $\g{I}$ & $\g{I}$ & $\g{I}$ & $\g{X}$ & $\g{I}$ & $\g{I}$ & $\g{I}$ & $\g{X}$ & $\g{I}$ & $\g{I}$ & $\g{I}$ & $\g{X}$ & $\g{I}$ & $\g{I}$ & $\g{I}$ & $\g{X}$ & $\g{I}$ & $\g{I}$ & $\g{I}$ & $\g{X}$ & $\g{I}$ & $\g{I}$ & $\g{I}$ & $\g{X}$ & $\g{I}$ & $\g{I}$ & $\g{I}$ & $\g{X}$ & $\g{I}$ & $\g{I}$ & $\g{I}$ & $\g{X}$ \\
         \cline{2-66}
         & $\{1,3\}$ & $\g{I}$ & $\g{I}$ & $\g{I}$ & $\g{I}$ & $\g{I}$ & $\g{X}$ & $\g{I}$ & $\g{X}$ & $\g{I}$ & $\g{I}$ & $\g{I}$ & $\g{I}$ & $\g{I}$ & $\g{X}$ & $\g{I}$ & $\g{X}$ & $\g{I}$ & $\g{I}$ & $\g{I}$ & $\g{I}$ & $\g{I}$ & $\g{X}$ & $\g{I}$ & $\g{X}$ & $\g{I}$ & $\g{I}$ & $\g{I}$ & $\g{I}$ & $\g{I}$ & $\g{X}$ & $\g{I}$ & $\g{X}$ & $\g{I}$ & $\g{I}$ & $\g{I}$ & $\g{I}$ & $\g{I}$ & $\g{X}$ & $\g{I}$ & $\g{X}$ & $\g{I}$ & $\g{I}$ & $\g{I}$ & $\g{I}$ & $\g{I}$ & $\g{X}$ & $\g{I}$ & $\g{X}$ & $\g{I}$ & $\g{I}$ & $\g{I}$ & $\g{I}$ & $\g{I}$ & $\g{X}$ & $\g{I}$ & $\g{X}$ & $\g{I}$ & $\g{I}$ & $\g{I}$ & $\g{I}$ & $\g{I}$ & $\g{X}$ & $\g{I}$ & $\g{X}$ \\
         \cline{2-66}
         & $\{1,4\}$ & $\g{I}$ & $\g{I}$ & $\g{I}$ & $\g{I}$ & $\g{I}$ & $\g{I}$ & $\g{I}$ & $\g{I}$ & $\g{I}$ & $\g{X}$ & $\g{I}$ & $\g{X}$ & $\g{I}$ & $\g{X}$ & $\g{I}$ & $\g{X}$ & $\g{I}$ & $\g{I}$ & $\g{I}$ & $\g{I}$ & $\g{I}$ & $\g{I}$ & $\g{I}$ & $\g{I}$ & $\g{I}$ & $\g{X}$ & $\g{I}$ & $\g{X}$ & $\g{I}$ & $\g{X}$ & $\g{I}$ & $\g{X}$ & $\g{I}$ & $\g{I}$ & $\g{I}$ & $\g{I}$ & $\g{I}$ & $\g{I}$ & $\g{I}$ & $\g{I}$ & $\g{I}$ & $\g{X}$ & $\g{I}$ & $\g{X}$ & $\g{I}$ & $\g{X}$ & $\g{I}$ & $\g{X}$ & $\g{I}$ & $\g{I}$ & $\g{I}$ & $\g{I}$ & $\g{I}$ & $\g{I}$ & $\g{I}$ & $\g{I}$ & $\g{I}$ & $\g{X}$ & $\g{I}$ & $\g{X}$ & $\g{I}$ & $\g{X}$ & $\g{I}$ & $\g{X}$ \\
         \cline{2-66}
         & $\{1,5\}$ & $\g{I}$ & $\g{I}$ & $\g{I}$ & $\g{I}$ & $\g{I}$ & $\g{I}$ & $\g{I}$ & $\g{I}$ & $\g{I}$ & $\g{I}$ & $\g{I}$ & $\g{I}$ & $\g{I}$ & $\g{I}$ & $\g{I}$ & $\g{I}$ & $\g{I}$ & $\g{X}$ & $\g{I}$ & $\g{X}$ & $\g{I}$ & $\g{X}$ & $\g{I}$ & $\g{X}$ & $\g{I}$ & $\g{X}$ & $\g{I}$ & $\g{X}$ & $\g{I}$ & $\g{X}$ & $\g{I}$ & $\g{X}$ & $\g{I}$ & $\g{I}$ & $\g{I}$ & $\g{I}$ & $\g{I}$ & $\g{I}$ & $\g{I}$ & $\g{I}$ & $\g{I}$ & $\g{I}$ & $\g{I}$ & $\g{I}$ & $\g{I}$ & $\g{I}$ & $\g{I}$ & $\g{I}$ & $\g{I}$ & $\g{X}$ & $\g{I}$ & $\g{X}$ & $\g{I}$ & $\g{X}$ & $\g{I}$ & $\g{X}$ & $\g{I}$ & $\g{X}$ & $\g{I}$ & $\g{X}$ & $\g{I}$ & $\g{X}$ & $\g{I}$ & $\g{X}$ \\
         \cline{2-66}
         & $\{1,6\}$ & $\g{I}$ & $\g{I}$ & $\g{I}$ & $\g{I}$ & $\g{I}$ & $\g{I}$ & $\g{I}$ & $\g{I}$ & $\g{I}$ & $\g{I}$ & $\g{I}$ & $\g{I}$ & $\g{I}$ & $\g{I}$ & $\g{I}$ & $\g{I}$ & $\g{I}$ & $\g{I}$ & $\g{I}$ & $\g{I}$ & $\g{I}$ & $\g{I}$ & $\g{I}$ & $\g{I}$ & $\g{I}$ & $\g{I}$ & $\g{I}$ & $\g{I}$ & $\g{I}$ & $\g{I}$ & $\g{I}$ & $\g{I}$ & $\g{I}$ & $\g{X}$ & $\g{I}$ & $\g{X}$ & $\g{I}$ & $\g{X}$ & $\g{I}$ & $\g{X}$ & $\g{I}$ & $\g{X}$ & $\g{I}$ & $\g{X}$ & $\g{I}$ & $\g{X}$ & $\g{I}$ & $\g{X}$ & $\g{I}$ & $\g{X}$ & $\g{I}$ & $\g{X}$ & $\g{I}$ & $\g{X}$ & $\g{I}$ & $\g{X}$ & $\g{I}$ & $\g{X}$ & $\g{I}$ & $\g{X}$ & $\g{I}$ & $\g{X}$ & $\g{I}$ & $\g{X}$ \\
         \cline{2-66}
         & $\{2,3\}$ & $\g{I}$ & $\g{I}$ & $\g{I}$ & $\g{I}$ & $\g{I}$ & $\g{I}$ & $\g{X}$ & $\g{X}$ & $\g{I}$ & $\g{I}$ & $\g{I}$ & $\g{I}$ & $\g{I}$ & $\g{I}$ & $\g{X}$ & $\g{X}$ & $\g{I}$ & $\g{I}$ & $\g{I}$ & $\g{I}$ & $\g{I}$ & $\g{I}$ & $\g{X}$ & $\g{X}$ & $\g{I}$ & $\g{I}$ & $\g{I}$ & $\g{I}$ & $\g{I}$ & $\g{I}$ & $\g{X}$ & $\g{X}$ & $\g{I}$ & $\g{I}$ & $\g{I}$ & $\g{I}$ & $\g{I}$ & $\g{I}$ & $\g{X}$ & $\g{X}$ & $\g{I}$ & $\g{I}$ & $\g{I}$ & $\g{I}$ & $\g{I}$ & $\g{I}$ & $\g{X}$ & $\g{X}$ & $\g{I}$ & $\g{I}$ & $\g{I}$ & $\g{I}$ & $\g{I}$ & $\g{I}$ & $\g{X}$ & $\g{X}$ & $\g{I}$ & $\g{I}$ & $\g{I}$ & $\g{I}$ & $\g{I}$ & $\g{I}$ & $\g{X}$ & $\g{X}$ \\
         \cline{2-66}
         & $\{2,4\}$ & $\g{I}$ & $\g{I}$ & $\g{I}$ & $\g{I}$ & $\g{I}$ & $\g{I}$ & $\g{I}$ & $\g{I}$ & $\g{I}$ & $\g{I}$ & $\g{X}$ & $\g{X}$ & $\g{I}$ & $\g{I}$ & $\g{X}$ & $\g{X}$ & $\g{I}$ & $\g{I}$ & $\g{I}$ & $\g{I}$ & $\g{I}$ & $\g{I}$ & $\g{I}$ & $\g{I}$ & $\g{I}$ & $\g{I}$ & $\g{X}$ & $\g{X}$ & $\g{I}$ & $\g{I}$ & $\g{X}$ & $\g{X}$ & $\g{I}$ & $\g{I}$ & $\g{I}$ & $\g{I}$ & $\g{I}$ & $\g{I}$ & $\g{I}$ & $\g{I}$ & $\g{I}$ & $\g{I}$ & $\g{X}$ & $\g{X}$ & $\g{I}$ & $\g{I}$ & $\g{X}$ & $\g{X}$ & $\g{I}$ & $\g{I}$ & $\g{I}$ & $\g{I}$ & $\g{I}$ & $\g{I}$ & $\g{I}$ & $\g{I}$ & $\g{I}$ & $\g{I}$ & $\g{X}$ & $\g{X}$ & $\g{I}$ & $\g{I}$ & $\g{X}$ & $\g{X}$ \\
         \cline{2-66}
         & $\{2,5\}$ & $\g{I}$ & $\g{I}$ & $\g{I}$ & $\g{I}$ & $\g{I}$ & $\g{I}$ & $\g{I}$ & $\g{I}$ & $\g{I}$ & $\g{I}$ & $\g{I}$ & $\g{I}$ & $\g{I}$ & $\g{I}$ & $\g{I}$ & $\g{I}$ & $\g{I}$ & $\g{I}$ & $\g{X}$ & $\g{X}$ & $\g{I}$ & $\g{I}$ & $\g{X}$ & $\g{X}$ & $\g{I}$ & $\g{I}$ & $\g{X}$ & $\g{X}$ & $\g{I}$ & $\g{I}$ & $\g{X}$ & $\g{X}$ & $\g{I}$ & $\g{I}$ & $\g{I}$ & $\g{I}$ & $\g{I}$ & $\g{I}$ & $\g{I}$ & $\g{I}$ & $\g{I}$ & $\g{I}$ & $\g{I}$ & $\g{I}$ & $\g{I}$ & $\g{I}$ & $\g{I}$ & $\g{I}$ & $\g{I}$ & $\g{I}$ & $\g{X}$ & $\g{X}$ & $\g{I}$ & $\g{I}$ & $\g{X}$ & $\g{X}$ & $\g{I}$ & $\g{I}$ & $\g{X}$ & $\g{X}$ & $\g{I}$ & $\g{I}$ & $\g{X}$ & $\g{X}$ \\
         \cline{2-66}
         & $\{2,6\}$ & $\g{I}$ & $\g{I}$ & $\g{I}$ & $\g{I}$ & $\g{I}$ & $\g{I}$ & $\g{I}$ & $\g{I}$ & $\g{I}$ & $\g{I}$ & $\g{I}$ & $\g{I}$ & $\g{I}$ & $\g{I}$ & $\g{I}$ & $\g{I}$ & $\g{I}$ & $\g{I}$ & $\g{I}$ & $\g{I}$ & $\g{I}$ & $\g{I}$ & $\g{I}$ & $\g{I}$ & $\g{I}$ & $\g{I}$ & $\g{I}$ & $\g{I}$ & $\g{I}$ & $\g{I}$ & $\g{I}$ & $\g{I}$ & $\g{I}$ & $\g{I}$ & $\g{X}$ & $\g{X}$ & $\g{I}$ & $\g{I}$ & $\g{X}$ & $\g{X}$ & $\g{I}$ & $\g{I}$ & $\g{X}$ & $\g{X}$ & $\g{I}$ & $\g{I}$ & $\g{X}$ & $\g{X}$ & $\g{I}$ & $\g{I}$ & $\g{X}$ & $\g{X}$ & $\g{I}$ & $\g{I}$ & $\g{X}$ & $\g{X}$ & $\g{I}$ & $\g{I}$ & $\g{X}$ & $\g{X}$ & $\g{I}$ & $\g{I}$ & $\g{X}$ & $\g{X}$ \\
         \cline{2-66}
         & $\{3,4\}$ & $\g{I}$ & $\g{I}$ & $\g{I}$ & $\g{I}$ & $\g{I}$ & $\g{I}$ & $\g{I}$ & $\g{I}$ & $\g{I}$ & $\g{I}$ & $\g{I}$ & $\g{I}$ & $\g{X}$ & $\g{X}$ & $\g{X}$ & $\g{X}$ & $\g{I}$ & $\g{I}$ & $\g{I}$ & $\g{I}$ & $\g{I}$ & $\g{I}$ & $\g{I}$ & $\g{I}$ & $\g{I}$ & $\g{I}$ & $\g{I}$ & $\g{I}$ & $\g{X}$ & $\g{X}$ & $\g{X}$ & $\g{X}$ & $\g{I}$ & $\g{I}$ & $\g{I}$ & $\g{I}$ & $\g{I}$ & $\g{I}$ & $\g{I}$ & $\g{I}$ & $\g{I}$ & $\g{I}$ & $\g{I}$ & $\g{I}$ & $\g{X}$ & $\g{X}$ & $\g{X}$ & $\g{X}$ & $\g{I}$ & $\g{I}$ & $\g{I}$ & $\g{I}$ & $\g{I}$ & $\g{I}$ & $\g{I}$ & $\g{I}$ & $\g{I}$ & $\g{I}$ & $\g{I}$ & $\g{I}$ & $\g{X}$ & $\g{X}$ & $\g{X}$ & $\g{X}$ \\
         \cline{2-66}
         & $\{3,5\}$ & $\g{I}$ & $\g{I}$ & $\g{I}$ & $\g{I}$ & $\g{I}$ & $\g{I}$ & $\g{I}$ & $\g{I}$ & $\g{I}$ & $\g{I}$ & $\g{I}$ & $\g{I}$ & $\g{I}$ & $\g{I}$ & $\g{I}$ & $\g{I}$ & $\g{I}$ & $\g{I}$ & $\g{I}$ & $\g{I}$ & $\g{X}$ & $\g{X}$ & $\g{X}$ & $\g{X}$ & $\g{I}$ & $\g{I}$ & $\g{I}$ & $\g{I}$ & $\g{X}$ & $\g{X}$ & $\g{X}$ & $\g{X}$ & $\g{I}$ & $\g{I}$ & $\g{I}$ & $\g{I}$ & $\g{I}$ & $\g{I}$ & $\g{I}$ & $\g{I}$ & $\g{I}$ & $\g{I}$ & $\g{I}$ & $\g{I}$ & $\g{I}$ & $\g{I}$ & $\g{I}$ & $\g{I}$ & $\g{I}$ & $\g{I}$ & $\g{I}$ & $\g{I}$ & $\g{X}$ & $\g{X}$ & $\g{X}$ & $\g{X}$ & $\g{I}$ & $\g{I}$ & $\g{I}$ & $\g{I}$ & $\g{X}$ & $\g{X}$ & $\g{X}$ & $\g{X}$ \\
         \cline{2-66}
         & $\{3,6\}$ & $\g{I}$ & $\g{I}$ & $\g{I}$ & $\g{I}$ & $\g{I}$ & $\g{I}$ & $\g{I}$ & $\g{I}$ & $\g{I}$ & $\g{I}$ & $\g{I}$ & $\g{I}$ & $\g{I}$ & $\g{I}$ & $\g{I}$ & $\g{I}$ & $\g{I}$ & $\g{I}$ & $\g{I}$ & $\g{I}$ & $\g{I}$ & $\g{I}$ & $\g{I}$ & $\g{I}$ & $\g{I}$ & $\g{I}$ & $\g{I}$ & $\g{I}$ & $\g{I}$ & $\g{I}$ & $\g{I}$ & $\g{I}$ & 
         $\g{I}$ & $\g{I}$ & $\g{I}$ & $\g{I}$ & $\g{X}$ & $\g{X}$ & $\g{X}$ & $\g{X}$ & $\g{I}$ & $\g{I}$ & $\g{I}$ & $\g{I}$ & $\g{X}$ & $\g{X}$ & $\g{X}$ & $\g{X}$ &
         $\g{I}$ & $\g{I}$ & $\g{I}$ & $\g{I}$ & $\g{X}$ & $\g{X}$ & $\g{X}$ & $\g{X}$ & $\g{I}$ & $\g{I}$ & $\g{I}$ & $\g{I}$ & $\g{X}$ & $\g{X}$ & $\g{X}$ & $\g{X}$ \\
         \cline{2-66}
         & $\{4,5\}$ & $\g{I}$ & $\g{I}$ & $\g{I}$ & $\g{I}$ & $\g{I}$ & $\g{I}$ & $\g{I}$ & $\g{I}$ & $\g{I}$ & $\g{I}$ & $\g{I}$ & $\g{I}$ & $\g{I}$ & $\g{I}$ & $\g{I}$ & $\g{I}$ &
         $\g{I}$ & $\g{I}$ & $\g{I}$ & $\g{I}$ & $\g{I}$ & $\g{I}$ & $\g{I}$ & $\g{I}$ & $\g{X}$ & $\g{X}$ & $\g{X}$ & $\g{X}$ & $\g{X}$ & $\g{X}$ & $\g{X}$ & $\g{X}$ &
         $\g{I}$ & $\g{I}$ & $\g{I}$ & $\g{I}$ & $\g{I}$ & $\g{I}$ & $\g{I}$ & $\g{I}$ & $\g{I}$ & $\g{I}$ & $\g{I}$ & $\g{I}$ & $\g{I}$ & $\g{I}$ & $\g{I}$ & $\g{I}$ &
         $\g{I}$ & $\g{I}$ & $\g{I}$ & $\g{I}$ & $\g{I}$ & $\g{I}$ & $\g{I}$ & $\g{I}$ & $\g{X}$ & $\g{X}$ & $\g{X}$ & $\g{X}$ & $\g{X}$ & $\g{X}$ & $\g{X}$ & $\g{X}$ \\
         \cline{2-66}
         & $\{4,6\}$ & $\g{I}$ & $\g{I}$ & $\g{I}$ & $\g{I}$ & $\g{I}$ & $\g{I}$ & $\g{I}$ & $\g{I}$ & $\g{I}$ & $\g{I}$ & $\g{I}$ & $\g{I}$ & $\g{I}$ & $\g{I}$ & $\g{I}$ & $\g{I}$ & $\g{I}$ & $\g{I}$ & $\g{I}$ & $\g{I}$ & $\g{I}$ & $\g{I}$ & $\g{I}$ & $\g{I}$ & $\g{I}$ & $\g{I}$ & $\g{I}$ & $\g{I}$ & $\g{I}$ & $\g{I}$ & $\g{I}$ & $\g{I}$ &
         $\g{I}$ & $\g{I}$ & $\g{I}$ & $\g{I}$ & $\g{I}$ & $\g{I}$ & $\g{I}$ & $\g{I}$ & $\g{X}$ & $\g{X}$ & $\g{X}$ & $\g{X}$ & $\g{X}$ & $\g{X}$ & $\g{X}$ & $\g{X}$ &
         $\g{I}$ & $\g{I}$ & $\g{I}$ & $\g{I}$ & $\g{I}$ & $\g{I}$ & $\g{I}$ & $\g{I}$ & $\g{X}$ & $\g{X}$ & $\g{X}$ & $\g{X}$ & $\g{X}$ & $\g{X}$ & $\g{X}$ & $\g{X}$ \\
         \cline{2-66}
         & $\{5,6\}$ & $\g{I}$ & $\g{I}$ & $\g{I}$ & $\g{I}$ & $\g{I}$ & $\g{I}$ & $\g{I}$ & $\g{I}$ & $\g{I}$ & $\g{I}$ & $\g{I}$ & $\g{I}$ & $\g{I}$ & $\g{I}$ & $\g{I}$ & $\g{I}$ & $\g{I}$ & $\g{I}$ & $\g{I}$ & $\g{I}$ & $\g{I}$ & $\g{I}$ & $\g{I}$ & $\g{I}$ & $\g{I}$ & $\g{I}$ & $\g{I}$ & $\g{I}$ & $\g{I}$ & $\g{I}$ & $\g{I}$ & $\g{I}$ &
         $\g{I}$ & $\g{I}$ & $\g{I}$ & $\g{I}$ & $\g{I}$ & $\g{I}$ & $\g{I}$ & $\g{I}$ & $\g{I}$ & $\g{I}$ & $\g{I}$ & $\g{I}$ & $\g{I}$ & $\g{I}$ & $\g{I}$ & $\g{I}$ &
         $\g{X}$ & $\g{X}$ & $\g{X}$ & $\g{X}$ & $\g{X}$ & $\g{X}$ & $\g{X}$ & $\g{X}$ & $\g{X}$ & $\g{X}$ & $\g{X}$ & $\g{X}$ & $\g{X}$ & $\g{X}$ & $\g{X}$ & $\g{X}$ \\
         \hline
    \end{tabular}
    }
    \caption{One possible choice of minimum-weight stabilizer generators of $\mathrm{QRM}(6)$ constructed from \cref{thm:reduced_weight}. Here we only display $\g{X}$-type stabilizer generators $g_x(A)$. The $\g{Z}$-type stabilizer generators $g_z(A)$ have the same form except that they are tensor products of $\g{I}$ and $\g{Z}$ instead of $\g{I}$ and $\g{X}$.}
    \label{tab:all_g_from_hA_m6}
\end{table}

\vspace*{\fill}
\end{landscape}

\newpage

\begin{landscape}
\vspace*{\fill}

\begin{table}[!htbp]
    \centering       
    \resizebox{\linewidth}{!}{
    \begin{tabular}{| c | c | K{2.2mm} | K{2.2mm} | K{2.2mm} | K{2.2mm} | K{2.2mm} | K{2.2mm} | K{2.2mm} | K{2.2mm} | K{2.2mm} | K{2.2mm} | K{2.2mm} | K{2.2mm} | K{2.2mm} | K{2.2mm} | K{2.2mm} | K{2.2mm} | K{2.2mm} | K{2.2mm} | K{2.2mm} | K{2.2mm} | K{2.2mm} | K{2.2mm} | K{2.2mm} | K{2.2mm} | K{2.2mm} | K{2.2mm} | K{2.2mm} | K{2.2mm} | K{2.2mm} | K{2.2mm} | K{2.2mm} | K{2.2mm} | K{2.2mm} | K{2.2mm} | K{2.2mm} | K{2.2mm} | K{2.2mm} | K{2.2mm} | K{2.2mm} | K{2.2mm} | K{2.2mm} | K{2.2mm} | K{2.2mm} | K{2.2mm} | K{2.2mm} | K{2.2mm} | K{2.2mm} | K{2.2mm} | K{2.2mm} | K{2.2mm} | K{2.2mm} | K{2.2mm} | K{2.2mm} | K{2.2mm} | K{2.2mm} | K{2.2mm} | K{2.2mm} | K{2.2mm} | K{2.2mm} | K{2.2mm} | K{2.2mm} | K{2.2mm} | K{2.2mm} | K{2.2mm} |}
         \hline
         \parbox[t]{3.5mm}{\multirow{21}{*}{\rotatebox[origin=c]{90}{\textrm{$\g{X}$-type}}}} & $B$ & \multicolumn{64}{|c|}{$\logg{X}(B)$} \\
         \cline{2-66}
         & $B_{1}=\{1,2,3\}$ & $\g{I}$ & $\g{I}$ & $\g{I}$ & $\g{I}$ & $\g{I}$ & $\g{I}$ & $\g{I}$ & $\g{X}$ & $\g{I}$ & $\g{I}$ & $\g{I}$ & $\g{I}$ & $\g{I}$ & $\g{I}$ & $\g{I}$ & $\g{X}$ &
         $\g{I}$ & $\g{I}$ & $\g{I}$ & $\g{I}$ & $\g{I}$ & $\g{I}$ & $\g{I}$ & $\g{X}$ & $\g{I}$ & $\g{I}$ & $\g{I}$ & $\g{I}$ & $\g{I}$ & $\g{I}$ & $\g{I}$ & $\g{X}$ &
         $\g{I}$ & $\g{I}$ & $\g{I}$ & $\g{I}$ & $\g{I}$ & $\g{I}$ & $\g{I}$ & $\g{X}$ & $\g{I}$ & $\g{I}$ & $\g{I}$ & $\g{I}$ & $\g{I}$ & $\g{I}$ & $\g{I}$ & $\g{X}$ &
         $\g{I}$ & $\g{I}$ & $\g{I}$ & $\g{I}$ & $\g{I}$ & $\g{I}$ & $\g{I}$ & $\g{X}$ & $\g{I}$ & $\g{I}$ & $\g{I}$ & $\g{I}$ & $\g{I}$ & $\g{I}$ & $\g{I}$ & $\g{X}$ \\
         \cline{2-66}
         & $B_{2}=\{1,2,4\}$ & $\g{I}$ & $\g{I}$ & $\g{I}$ & $\g{I}$ & $\g{I}$ & $\g{I}$ & $\g{I}$ & $\g{I}$ & $\g{I}$ & $\g{I}$ & $\g{I}$ & $\g{X}$ & $\g{I}$ & $\g{I}$ & $\g{I}$ & $\g{X}$ &
         $\g{I}$ & $\g{I}$ & $\g{I}$ & $\g{I}$ & $\g{I}$ & $\g{I}$ & $\g{I}$ & $\g{I}$ & $\g{I}$ & $\g{I}$ & $\g{I}$ & $\g{X}$ & $\g{I}$ & $\g{I}$ & $\g{I}$ & $\g{X}$ &
         $\g{I}$ & $\g{I}$ & $\g{I}$ & $\g{I}$ & $\g{I}$ & $\g{I}$ & $\g{I}$ & $\g{I}$ & $\g{I}$ & $\g{I}$ & $\g{I}$ & $\g{X}$ & $\g{I}$ & $\g{I}$ & $\g{I}$ & $\g{X}$ &
         $\g{I}$ & $\g{I}$ & $\g{I}$ & $\g{I}$ & $\g{I}$ & $\g{I}$ & $\g{I}$ & $\g{I}$ & $\g{I}$ & $\g{I}$ & $\g{I}$ & $\g{X}$ & $\g{I}$ & $\g{I}$ & $\g{I}$ & $\g{X}$ \\
         \cline{2-66}
         & $B_{3}=\{1,2,5\}$ & $\g{I}$ & $\g{I}$ & $\g{I}$ & $\g{I}$ & $\g{I}$ & $\g{I}$ & $\g{I}$ & $\g{I}$ & $\g{I}$ & $\g{I}$ & $\g{I}$ & $\g{I}$ & $\g{I}$ & $\g{I}$ & $\g{I}$ & $\g{I}$ &
         $\g{I}$ & $\g{I}$ & $\g{I}$ & $\g{X}$ & $\g{I}$ & $\g{I}$ & $\g{I}$ & $\g{X}$ & $\g{I}$ & $\g{I}$ & $\g{I}$ & $\g{X}$ & $\g{I}$ & $\g{I}$ & $\g{I}$ & $\g{X}$ &
         $\g{I}$ & $\g{I}$ & $\g{I}$ & $\g{I}$ & $\g{I}$ & $\g{I}$ & $\g{I}$ & $\g{I}$ & $\g{I}$ & $\g{I}$ & $\g{I}$ & $\g{I}$ & $\g{I}$ & $\g{I}$ & $\g{I}$ & $\g{I}$ &
         $\g{I}$ & $\g{I}$ & $\g{I}$ & $\g{X}$ & $\g{I}$ & $\g{I}$ & $\g{I}$ & $\g{X}$ & $\g{I}$ & $\g{I}$ & $\g{I}$ & $\g{X}$ & $\g{I}$ & $\g{I}$ & $\g{I}$ & $\g{X}$ \\
         \cline{2-66}
         & $B_{4}=\{1,2,6\}$ & $\g{I}$ & $\g{I}$ & $\g{I}$ & $\g{I}$ & $\g{I}$ & $\g{I}$ & $\g{I}$ & $\g{I}$ & $\g{I}$ & $\g{I}$ & $\g{I}$ & $\g{I}$ & $\g{I}$ & $\g{I}$ & $\g{I}$ & $\g{I}$ & $\g{I}$ & $\g{I}$ & $\g{I}$ & $\g{I}$ & $\g{I}$ & $\g{I}$ & $\g{I}$ & $\g{I}$ & $\g{I}$ & $\g{I}$ & $\g{I}$ & $\g{I}$ & $\g{I}$ & $\g{I}$ & $\g{I}$ & $\g{I}$ &
         $\g{I}$ & $\g{I}$ & $\g{I}$ & $\g{X}$ & $\g{I}$ & $\g{I}$ & $\g{I}$ & $\g{X}$ & $\g{I}$ & $\g{I}$ & $\g{I}$ & $\g{X}$ & $\g{I}$ & $\g{I}$ & $\g{I}$ & $\g{X}$ &
         $\g{I}$ & $\g{I}$ & $\g{I}$ & $\g{X}$ & $\g{I}$ & $\g{I}$ & $\g{I}$ & $\g{X}$ & $\g{I}$ & $\g{I}$ & $\g{I}$ & $\g{X}$ & $\g{I}$ & $\g{I}$ & $\g{I}$ & $\g{X}$ \\
         \cline{2-66}
         & $B_{5}=\{1,3,4\}$ & $\g{I}$ & $\g{I}$ & $\g{I}$ & $\g{I}$ & $\g{I}$ & $\g{I}$ & $\g{I}$ & $\g{I}$ & $\g{I}$ & $\g{I}$ & $\g{I}$ & $\g{I}$ & $\g{I}$ & $\g{X}$ & $\g{I}$ & $\g{X}$ &
         $\g{I}$ & $\g{I}$ & $\g{I}$ & $\g{I}$ & $\g{I}$ & $\g{I}$ & $\g{I}$ & $\g{I}$ & $\g{I}$ & $\g{I}$ & $\g{I}$ & $\g{I}$ & $\g{I}$ & $\g{X}$ & $\g{I}$ & $\g{X}$ &
         $\g{I}$ & $\g{I}$ & $\g{I}$ & $\g{I}$ & $\g{I}$ & $\g{I}$ & $\g{I}$ & $\g{I}$ & $\g{I}$ & $\g{I}$ & $\g{I}$ & $\g{I}$ & $\g{I}$ & $\g{X}$ & $\g{I}$ & $\g{X}$ &
         $\g{I}$ & $\g{I}$ & $\g{I}$ & $\g{I}$ & $\g{I}$ & $\g{I}$ & $\g{I}$ & $\g{I}$ & $\g{I}$ & $\g{I}$ & $\g{I}$ & $\g{I}$ & $\g{I}$ & $\g{X}$ & $\g{I}$ & $\g{X}$ \\
         \cline{2-66}
         & $B_{6}=\{1,3,5\}$ & $\g{I}$ & $\g{I}$ & $\g{I}$ & $\g{I}$ & $\g{I}$ & $\g{I}$ & $\g{I}$ & $\g{I}$ & $\g{I}$ & $\g{I}$ & $\g{I}$ & $\g{I}$ & $\g{I}$ & $\g{I}$ & $\g{I}$ & $\g{I}$ &
         $\g{I}$ & $\g{I}$ & $\g{I}$ & $\g{I}$ & $\g{I}$ & $\g{X}$ & $\g{I}$ & $\g{X}$ & $\g{I}$ & $\g{I}$ & $\g{I}$ & $\g{I}$ & $\g{I}$ & $\g{X}$ & $\g{I}$ & $\g{X}$ &
         $\g{I}$ & $\g{I}$ & $\g{I}$ & $\g{I}$ & $\g{I}$ & $\g{I}$ & $\g{I}$ & $\g{I}$ & $\g{I}$ & $\g{I}$ & $\g{I}$ & $\g{I}$ & $\g{I}$ & $\g{I}$ & $\g{I}$ & $\g{I}$ &
         $\g{I}$ & $\g{I}$ & $\g{I}$ & $\g{I}$ & $\g{I}$ & $\g{X}$ & $\g{I}$ & $\g{X}$ & $\g{I}$ & $\g{I}$ & $\g{I}$ & $\g{I}$ & $\g{I}$ & $\g{X}$ & $\g{I}$ & $\g{X}$ \\
         \cline{2-66}
         & $B_{7}=\{1,3,6\}$ & $\g{I}$ & $\g{I}$ & $\g{I}$ & $\g{I}$ & $\g{I}$ & $\g{I}$ & $\g{I}$ & $\g{I}$ & $\g{I}$ & $\g{I}$ & $\g{I}$ & $\g{I}$ & $\g{I}$ & $\g{I}$ & $\g{I}$ & $\g{I}$ & $\g{I}$ & $\g{I}$ & $\g{I}$ & $\g{I}$ & $\g{I}$ & $\g{I}$ & $\g{I}$ & $\g{I}$ & $\g{I}$ & $\g{I}$ & $\g{I}$ & $\g{I}$ & $\g{I}$ & $\g{I}$ & $\g{I}$ & $\g{I}$ &
         $\g{I}$ & $\g{I}$ & $\g{I}$ & $\g{I}$ & $\g{I}$ & $\g{X}$ & $\g{I}$ & $\g{X}$ & 
         $\g{I}$ & $\g{I}$ & $\g{I}$ & $\g{I}$ & $\g{I}$ & $\g{X}$ & $\g{I}$ & $\g{X}$ &
         $\g{I}$ & $\g{I}$ & $\g{I}$ & $\g{I}$ & $\g{I}$ & $\g{X}$ & $\g{I}$ & $\g{X}$ &
         $\g{I}$ & $\g{I}$ & $\g{I}$ & $\g{I}$ & $\g{I}$ & $\g{X}$ & $\g{I}$ & $\g{X}$ \\
         \cline{2-66}
         & $B_{8}=\{1,4,5\}$ & $\g{I}$ & $\g{I}$ & $\g{I}$ & $\g{I}$ & $\g{I}$ & $\g{I}$ & $\g{I}$ & $\g{I}$ & $\g{I}$ & $\g{I}$ & $\g{I}$ & $\g{I}$ & $\g{I}$ & $\g{I}$ & $\g{I}$ & $\g{I}$ &
         $\g{I}$ & $\g{I}$ & $\g{I}$ & $\g{I}$ & $\g{I}$ & $\g{I}$ & $\g{I}$ & $\g{I}$ & $\g{I}$ & $\g{X}$ & $\g{I}$ & $\g{X}$ & $\g{I}$ & $\g{X}$ & $\g{I}$ & $\g{X}$ &
         $\g{I}$ & $\g{I}$ & $\g{I}$ & $\g{I}$ & $\g{I}$ & $\g{I}$ & $\g{I}$ & $\g{I}$ & $\g{I}$ & $\g{I}$ & $\g{I}$ & $\g{I}$ & $\g{I}$ & $\g{I}$ & $\g{I}$ & $\g{I}$ &
         $\g{I}$ & $\g{I}$ & $\g{I}$ & $\g{I}$ & $\g{I}$ & $\g{I}$ & $\g{I}$ & $\g{I}$ & $\g{I}$ & $\g{X}$ & $\g{I}$ & $\g{X}$ & $\g{I}$ & $\g{X}$ & $\g{I}$ & $\g{X}$ \\
         \cline{2-66}
         & $B_{9}=\{1,4,6\}$ & $\g{I}$ & $\g{I}$ & $\g{I}$ & $\g{I}$ & $\g{I}$ & $\g{I}$ & $\g{I}$ & $\g{I}$ & $\g{I}$ & $\g{I}$ & $\g{I}$ & $\g{I}$ & $\g{I}$ & $\g{I}$ & $\g{I}$ & $\g{I}$ & 
         $\g{I}$ & $\g{I}$ & $\g{I}$ & $\g{I}$ & $\g{I}$ & $\g{I}$ & $\g{I}$ & $\g{I}$ & $\g{I}$ & $\g{I}$ & $\g{I}$ & $\g{I}$ & $\g{I}$ & $\g{I}$ & $\g{I}$ & $\g{I}$ & 
         $\g{I}$ & $\g{I}$ & $\g{I}$ & $\g{I}$ & $\g{I}$ & $\g{I}$ & $\g{I}$ & $\g{I}$ & $\g{I}$ & $\g{X}$ & $\g{I}$ & $\g{X}$ & $\g{I}$ & $\g{X}$ & $\g{I}$ & $\g{X}$ &
         $\g{I}$ & $\g{I}$ & $\g{I}$ & $\g{I}$ & $\g{I}$ & $\g{I}$ & $\g{I}$ & $\g{I}$ & $\g{I}$ & $\g{X}$ & $\g{I}$ & $\g{X}$ & $\g{I}$ & $\g{X}$ & $\g{I}$ & $\g{X}$ \\
         \cline{2-66}
         & $B_{10}=\{1,5,6\}$ & $\g{I}$ & $\g{I}$ & $\g{I}$ & $\g{I}$ & $\g{I}$ & $\g{I}$ & $\g{I}$ & $\g{I}$ & $\g{I}$ & $\g{I}$ & $\g{I}$ & $\g{I}$ & $\g{I}$ & $\g{I}$ & $\g{I}$ & $\g{I}$ & 
         $\g{I}$ & $\g{I}$ & $\g{I}$ & $\g{I}$ & $\g{I}$ & $\g{I}$ & $\g{I}$ & $\g{I}$ & $\g{I}$ & $\g{I}$ & $\g{I}$ & $\g{I}$ & $\g{I}$ & $\g{I}$ & $\g{I}$ & $\g{I}$ & 
         $\g{I}$ & $\g{I}$ & $\g{I}$ & $\g{I}$ & $\g{I}$ & $\g{I}$ & $\g{I}$ & $\g{I}$ & $\g{I}$ & $\g{I}$ & $\g{I}$ & $\g{I}$ & $\g{I}$ & $\g{I}$ & $\g{I}$ & $\g{I}$ & $\g{I}$ & $\g{X}$ & $\g{I}$ & $\g{X}$ & $\g{I}$ & $\g{X}$ & $\g{I}$ & $\g{X}$ & $\g{I}$ & $\g{X}$ & $\g{I}$ & $\g{X}$ & $\g{I}$ & $\g{X}$ & $\g{I}$ & $\g{X}$ \\
         \cline{2-66}
        & $B_{11}=\{4,5,6\}$ & $\g{I}$ & $\g{I}$ & $\g{I}$ & $\g{I}$ & $\g{I}$ & $\g{I}$ & $\g{I}$ & $\g{I}$ & $\g{I}$ & $\g{I}$ & $\g{I}$ & $\g{I}$ & $\g{I}$ & $\g{I}$ & $\g{I}$ & $\g{I}$ &
         $\g{I}$ & $\g{I}$ & $\g{I}$ & $\g{I}$ & $\g{I}$ & $\g{I}$ & $\g{I}$ & $\g{I}$ & $\g{I}$ & $\g{I}$ & $\g{I}$ & $\g{I}$ & $\g{I}$ & $\g{I}$ & $\g{I}$ & $\g{I}$ &
         $\g{I}$ & $\g{I}$ & $\g{I}$ & $\g{I}$ & $\g{I}$ & $\g{I}$ & $\g{I}$ & $\g{I}$ & $\g{I}$ & $\g{I}$ & $\g{I}$ & $\g{I}$ & $\g{I}$ & $\g{I}$ & $\g{I}$ & $\g{I}$ &
         $\g{I}$ & $\g{I}$ & $\g{I}$ & $\g{I}$ & $\g{I}$ & $\g{I}$ & $\g{I}$ & $\g{I}$ & $\g{X}$ & $\g{X}$ & $\g{X}$ & $\g{X}$ & $\g{X}$ & $\g{X}$ & $\g{X}$ & $\g{X}$ \\
         \cline{2-66}
        & $B_{12}=\{3,5,6\}$ & $\g{I}$ & $\g{I}$ & $\g{I}$ & $\g{I}$ & $\g{I}$ & $\g{I}$ & $\g{I}$ & $\g{I}$ & $\g{I}$ & $\g{I}$ & $\g{I}$ & $\g{I}$ & $\g{I}$ & $\g{I}$ & $\g{I}$ & $\g{I}$ & $\g{I}$ & $\g{I}$ & $\g{I}$ & $\g{I}$ & $\g{I}$ & $\g{I}$ & $\g{I}$ & $\g{I}$ & $\g{I}$ & $\g{I}$ & $\g{I}$ & $\g{I}$ & $\g{I}$ & $\g{I}$ & $\g{I}$ & $\g{I}$ & $\g{I}$ & $\g{I}$ & $\g{I}$ & $\g{I}$ & $\g{I}$ & $\g{I}$ & $\g{I}$ & $\g{I}$ & $\g{I}$ & $\g{I}$ & $\g{I}$ & $\g{I}$ & $\g{I}$ & $\g{I}$ & $\g{I}$ & $\g{I}$ & $\g{I}$ & $\g{I}$ & $\g{I}$ & $\g{I}$ & $\g{X}$ & $\g{X}$ & $\g{X}$ & $\g{X}$ & $\g{I}$ & $\g{I}$ & $\g{I}$ & $\g{I}$ & $\g{X}$ & $\g{X}$ & $\g{X}$ & $\g{X}$ \\
         \cline{2-66}
        & $B_{13}=\{3,4,6\}$ & $\g{I}$ & $\g{I}$ & $\g{I}$ & $\g{I}$ & $\g{I}$ & $\g{I}$ & $\g{I}$ & $\g{I}$ & $\g{I}$ & $\g{I}$ & $\g{I}$ & $\g{I}$ & $\g{I}$ & $\g{I}$ & $\g{I}$ & $\g{I}$ & $\g{I}$ & $\g{I}$ & $\g{I}$ & $\g{I}$ & $\g{I}$ & $\g{I}$ & $\g{I}$ & $\g{I}$ & $\g{I}$ & $\g{I}$ & $\g{I}$ & $\g{I}$ & $\g{I}$ & $\g{I}$ & $\g{I}$ & $\g{I}$ & $\g{I}$ & $\g{I}$ & $\g{I}$ & $\g{I}$ & $\g{I}$ & $\g{I}$ & $\g{I}$ & $\g{I}$ & $\g{I}$ & $\g{I}$ & $\g{I}$ & $\g{I}$ & $\g{X}$ & $\g{X}$ & $\g{X}$ & $\g{X}$ & $\g{I}$ & $\g{I}$ & $\g{I}$ & $\g{I}$ & $\g{I}$ & $\g{I}$ & $\g{I}$ & $\g{I}$ & $\g{I}$ & $\g{I}$ & $\g{I}$ & $\g{I}$ & $\g{X}$ & $\g{X}$ & $\g{X}$ & $\g{X}$ \\
         \cline{2-66}
        & $B_{14}=\{3,4,5\}$ & $\g{I}$ & $\g{I}$ & $\g{I}$ & $\g{I}$ & $\g{I}$ & $\g{I}$ & $\g{I}$ & $\g{I}$ & $\g{I}$ & $\g{I}$ & $\g{I}$ & $\g{I}$ & $\g{I}$ & $\g{I}$ & $\g{I}$ & $\g{I}$ & $\g{I}$ & $\g{I}$ & $\g{I}$ & $\g{I}$ & $\g{I}$ & $\g{I}$ & $\g{I}$ & $\g{I}$ & $\g{I}$ & $\g{I}$ & $\g{I}$ & $\g{I}$ & $\g{X}$ & $\g{X}$ & $\g{X}$ & $\g{X}$ & $\g{I}$ & $\g{I}$ & $\g{I}$ & $\g{I}$ & $\g{I}$ & $\g{I}$ & $\g{I}$ & $\g{I}$ & $\g{I}$ & $\g{I}$ & $\g{I}$ & $\g{I}$ & $\g{I}$ & $\g{I}$ & $\g{I}$ & $\g{I}$ & $\g{I}$ & $\g{I}$ & $\g{I}$ & $\g{I}$ & $\g{I}$ & $\g{I}$ & $\g{I}$ & $\g{I}$ & $\g{I}$ & $\g{I}$ & $\g{I}$ & $\g{I}$ & $\g{X}$ & $\g{X}$ & $\g{X}$ & $\g{X}$ \\
         \cline{2-66}
        & $B_{15}=\{2,5,6\}$ & $\g{I}$ & $\g{I}$ & $\g{I}$ & $\g{I}$ & $\g{I}$ & $\g{I}$ & $\g{I}$ & $\g{I}$ & $\g{I}$ & $\g{I}$ & $\g{I}$ & $\g{I}$ & $\g{I}$ & $\g{I}$ & $\g{I}$ & $\g{I}$ &
         $\g{I}$ & $\g{I}$ & $\g{I}$ & $\g{I}$ & $\g{I}$ & $\g{I}$ & $\g{I}$ & $\g{I}$ & $\g{I}$ & $\g{I}$ & $\g{I}$ & $\g{I}$ & $\g{I}$ & $\g{I}$ & $\g{I}$ & $\g{I}$ & 
         $\g{I}$ & $\g{I}$ & $\g{I}$ & $\g{I}$ & $\g{I}$ & $\g{I}$ & $\g{I}$ & $\g{I}$ & $\g{I}$ & $\g{I}$ & $\g{I}$ & $\g{I}$ & $\g{I}$ & $\g{I}$ & $\g{I}$ & $\g{I}$ & 
         $\g{I}$ & $\g{I}$ & $\g{X}$ & $\g{X}$ & $\g{I}$ & $\g{I}$ & $\g{X}$ & $\g{X}$ & $\g{I}$ & $\g{I}$ & $\g{X}$ & $\g{X}$ & $\g{I}$ & $\g{I}$ & $\g{X}$ & $\g{X}$ \\
         \cline{2-66}
        & $B_{16}=\{2,4,6\}$ & $\g{I}$ & $\g{I}$ & $\g{I}$ & $\g{I}$ & $\g{I}$ & $\g{I}$ & $\g{I}$ & $\g{I}$ & $\g{I}$ & $\g{I}$ & $\g{I}$ & $\g{I}$ & $\g{I}$ & $\g{I}$ & $\g{I}$ & $\g{I}$ &
         $\g{I}$ & $\g{I}$ & $\g{I}$ & $\g{I}$ & $\g{I}$ & $\g{I}$ & $\g{I}$ & $\g{I}$ & $\g{I}$ & $\g{I}$ & $\g{I}$ & $\g{I}$ & $\g{I}$ & $\g{I}$ & $\g{I}$ & $\g{I}$ &
         $\g{I}$ & $\g{I}$ & $\g{I}$ & $\g{I}$ & $\g{I}$ & $\g{I}$ & $\g{I}$ & $\g{I}$ & $\g{I}$ & $\g{I}$ & $\g{X}$ & $\g{X}$ & $\g{I}$ & $\g{I}$ & $\g{X}$ & $\g{X}$ &
         $\g{I}$ & $\g{I}$ & $\g{I}$ & $\g{I}$ & $\g{I}$ & $\g{I}$ & $\g{I}$ & $\g{I}$ & $\g{I}$ & $\g{I}$ & $\g{X}$ & $\g{X}$ & $\g{I}$ & $\g{I}$ & $\g{X}$ & $\g{X}$ \\
         \cline{2-66}
        & $B_{17}=\{2,4,5\}$ & $\g{I}$ & $\g{I}$ & $\g{I}$ & $\g{I}$ & $\g{I}$ & $\g{I}$ & $\g{I}$ & $\g{I}$ & $\g{I}$ & $\g{I}$ & $\g{I}$ & $\g{I}$ & $\g{I}$ & $\g{I}$ & $\g{I}$ & $\g{I}$ &
         $\g{I}$ & $\g{I}$ & $\g{I}$ & $\g{I}$ & $\g{I}$ & $\g{I}$ & $\g{I}$ & $\g{I}$ & $\g{I}$ & $\g{I}$ & $\g{X}$ & $\g{X}$ & $\g{I}$ & $\g{I}$ & $\g{X}$ & $\g{X}$ &
         $\g{I}$ & $\g{I}$ & $\g{I}$ & $\g{I}$ & $\g{I}$ & $\g{I}$ & $\g{I}$ & $\g{I}$ & $\g{I}$ & $\g{I}$ & $\g{I}$ & $\g{I}$ & $\g{I}$ & $\g{I}$ & $\g{I}$ & $\g{I}$ &
         $\g{I}$ & $\g{I}$ & $\g{I}$ & $\g{I}$ & $\g{I}$ & $\g{I}$ & $\g{I}$ & $\g{I}$ & $\g{I}$ & $\g{I}$ & $\g{X}$ & $\g{X}$ & $\g{I}$ & $\g{I}$ & $\g{X}$ & $\g{X}$ \\
         \cline{2-66}
        & $B_{18}=\{2,3,6\}$ & $\g{I}$ & $\g{I}$ & $\g{I}$ & $\g{I}$ & $\g{I}$ & $\g{I}$ & $\g{I}$ & $\g{I}$ & $\g{I}$ & $\g{I}$ & $\g{I}$ & $\g{I}$ & $\g{I}$ & $\g{I}$ & $\g{I}$ & $\g{I}$ & $\g{I}$ & $\g{I}$ & $\g{I}$ & $\g{I}$ & $\g{I}$ & $\g{I}$ & $\g{I}$ & $\g{I}$ & $\g{I}$ & $\g{I}$ & $\g{I}$ & $\g{I}$ & $\g{I}$ & $\g{I}$ & $\g{I}$ & $\g{I}$ & 
         $\g{I}$ & $\g{I}$ & $\g{I}$ & $\g{I}$ & $\g{I}$ & $\g{I}$ & $\g{X}$ & $\g{X}$ & $\g{I}$ & $\g{I}$ & $\g{I}$ & $\g{I}$ & $\g{I}$ & $\g{I}$ & $\g{X}$ & $\g{X}$ & $\g{I}$ & $\g{I}$ & $\g{I}$ & $\g{I}$ & $\g{I}$ & $\g{I}$ & $\g{X}$ & $\g{X}$ & $\g{I}$ & $\g{I}$ & $\g{I}$ & $\g{I}$ & $\g{I}$ & $\g{I}$ & $\g{X}$ & $\g{X}$ \\
         \cline{2-66}
        & $B_{19}=\{2,3,5\}$ & $\g{I}$ & $\g{I}$ & $\g{I}$ & $\g{I}$ & $\g{I}$ & $\g{I}$ & $\g{I}$ & $\g{I}$ & $\g{I}$ & $\g{I}$ & $\g{I}$ & $\g{I}$ & $\g{I}$ & $\g{I}$ & $\g{I}$ & $\g{I}$ &
         $\g{I}$ & $\g{I}$ & $\g{I}$ & $\g{I}$ & $\g{I}$ & $\g{I}$ & $\g{X}$ & $\g{X}$ & $\g{I}$ & $\g{I}$ & $\g{I}$ & $\g{I}$ & $\g{I}$ & $\g{I}$ & $\g{X}$ & $\g{X}$ &
         $\g{I}$ & $\g{I}$ & $\g{I}$ & $\g{I}$ & $\g{I}$ & $\g{I}$ & $\g{I}$ & $\g{I}$ & $\g{I}$ & $\g{I}$ & $\g{I}$ & $\g{I}$ & $\g{I}$ & $\g{I}$ & $\g{I}$ & $\g{I}$ &
         $\g{I}$ & $\g{I}$ & $\g{I}$ & $\g{I}$ & $\g{I}$ & $\g{I}$ & $\g{X}$ & $\g{X}$ & $\g{I}$ & $\g{I}$ & $\g{I}$ & $\g{I}$ & $\g{I}$ & $\g{I}$ & $\g{X}$ & $\g{X}$ \\
         \cline{2-66}
         & $B_{20}=\{2,3,4\}$ & $\g{I}$ & $\g{I}$ & $\g{I}$ & $\g{I}$ & $\g{I}$ & $\g{I}$ & $\g{I}$ & $\g{I}$ & $\g{I}$ & $\g{I}$ & $\g{I}$ & $\g{I}$ & $\g{I}$ & $\g{I}$ & $\g{X}$ & $\g{X}$ &
         $\g{I}$ & $\g{I}$ & $\g{I}$ & $\g{I}$ & $\g{I}$ & $\g{I}$ & $\g{I}$ & $\g{I}$ & $\g{I}$ & $\g{I}$ & $\g{I}$ & $\g{I}$ & $\g{I}$ & $\g{I}$ & $\g{X}$ & $\g{X}$ &
         $\g{I}$ & $\g{I}$ & $\g{I}$ & $\g{I}$ & $\g{I}$ & $\g{I}$ & $\g{I}$ & $\g{I}$ & $\g{I}$ & $\g{I}$ & $\g{I}$ & $\g{I}$ & $\g{I}$ & $\g{I}$ & $\g{X}$ & $\g{X}$ &
         $\g{I}$ & $\g{I}$ & $\g{I}$ & $\g{I}$ & $\g{I}$ & $\g{I}$ & $\g{I}$ & $\g{I}$ & $\g{I}$ & $\g{I}$ & $\g{I}$ & $\g{I}$ & $\g{I}$ & $\g{I}$ & $\g{X}$ & $\g{X}$ \\
         \hline
        \parbox[t]{3.5mm}{\multirow{21}{*}{\rotatebox[origin=c]{90}{\textrm{$\g{Z}$-type}}}} & $B$ & \multicolumn{64}{|c|}{$\logg{Z}(B)$} \\
         \cline{2-66}
        & $B_{1}=\{1,2,3\}$ & $\g{I}$ & $\g{I}$ & $\g{I}$ & $\g{I}$ & $\g{I}$ & $\g{I}$ & $\g{I}$ & $\g{I}$ & $\g{I}$ & $\g{I}$ & $\g{I}$ & $\g{I}$ & $\g{I}$ & $\g{I}$ & $\g{I}$ & $\g{I}$ &
         $\g{I}$ & $\g{I}$ & $\g{I}$ & $\g{I}$ & $\g{I}$ & $\g{I}$ & $\g{I}$ & $\g{I}$ & $\g{I}$ & $\g{I}$ & $\g{I}$ & $\g{I}$ & $\g{I}$ & $\g{I}$ & $\g{I}$ & $\g{I}$ &
         $\g{I}$ & $\g{I}$ & $\g{I}$ & $\g{I}$ & $\g{I}$ & $\g{I}$ & $\g{I}$ & $\g{I}$ & $\g{I}$ & $\g{I}$ & $\g{I}$ & $\g{I}$ & $\g{I}$ & $\g{I}$ & $\g{I}$ & $\g{I}$ &
         $\g{I}$ & $\g{I}$ & $\g{I}$ & $\g{I}$ & $\g{I}$ & $\g{I}$ & $\g{I}$ & $\g{I}$ & $\g{Z}$ & $\g{Z}$ & $\g{Z}$ & $\g{Z}$ & $\g{Z}$ & $\g{Z}$ & $\g{Z}$ & $\g{Z}$ \\
        \cline{2-66}
        & $B_{2}=\{1,2,4\}$ & $\g{I}$ & $\g{I}$ & $\g{I}$ & $\g{I}$ & $\g{I}$ & $\g{I}$ & $\g{I}$ & $\g{I}$ & $\g{I}$ & $\g{I}$ & $\g{I}$ & $\g{I}$ & $\g{I}$ & $\g{I}$ & $\g{I}$ & $\g{I}$ & $\g{I}$ & $\g{I}$ & $\g{I}$ & $\g{I}$ & $\g{I}$ & $\g{I}$ & $\g{I}$ & $\g{I}$ & $\g{I}$ & $\g{I}$ & $\g{I}$ & $\g{I}$ & $\g{I}$ & $\g{I}$ & $\g{I}$ & $\g{I}$ & $\g{I}$ & $\g{I}$ & $\g{I}$ & $\g{I}$ & $\g{I}$ & $\g{I}$ & $\g{I}$ & $\g{I}$ & $\g{I}$ & $\g{I}$ & $\g{I}$ & $\g{I}$ & $\g{I}$ & $\g{I}$ & $\g{I}$ & $\g{I}$ & $\g{I}$ & $\g{I}$ & $\g{I}$ & $\g{I}$ & $\g{Z}$ & $\g{Z}$ & $\g{Z}$ & $\g{Z}$ & $\g{I}$ & $\g{I}$ & $\g{I}$ & $\g{I}$ & $\g{Z}$ & $\g{Z}$ & $\g{Z}$ & $\g{Z}$ \\
        \cline{2-66}
        & $B_{3}=\{1,2,5\}$ & $\g{I}$ & $\g{I}$ & $\g{I}$ & $\g{I}$ & $\g{I}$ & $\g{I}$ & $\g{I}$ & $\g{I}$ & $\g{I}$ & $\g{I}$ & $\g{I}$ & $\g{I}$ & $\g{I}$ & $\g{I}$ & $\g{I}$ & $\g{I}$ & $\g{I}$ & $\g{I}$ & $\g{I}$ & $\g{I}$ & $\g{I}$ & $\g{I}$ & $\g{I}$ & $\g{I}$ & $\g{I}$ & $\g{I}$ & $\g{I}$ & $\g{I}$ & $\g{I}$ & $\g{I}$ & $\g{I}$ & $\g{I}$ & $\g{I}$ & $\g{I}$ & $\g{I}$ & $\g{I}$ & $\g{I}$ & $\g{I}$ & $\g{I}$ & $\g{I}$ & $\g{I}$ & $\g{I}$ & $\g{I}$ & $\g{I}$ & $\g{Z}$ & $\g{Z}$ & $\g{Z}$ & $\g{Z}$ & $\g{I}$ & $\g{I}$ & $\g{I}$ & $\g{I}$ & $\g{I}$ & $\g{I}$ & $\g{I}$ & $\g{I}$ & $\g{I}$ & $\g{I}$ & $\g{I}$ & $\g{I}$ & $\g{Z}$ & $\g{Z}$ & $\g{Z}$ & $\g{Z}$ \\
        \cline{2-66}
        & $B_{4}=\{1,2,6\}$ & $\g{I}$ & $\g{I}$ & $\g{I}$ & $\g{I}$ & $\g{I}$ & $\g{I}$ & $\g{I}$ & $\g{I}$ & $\g{I}$ & $\g{I}$ & $\g{I}$ & $\g{I}$ & $\g{I}$ & $\g{I}$ & $\g{I}$ & $\g{I}$ & $\g{I}$ & $\g{I}$ & $\g{I}$ & $\g{I}$ & $\g{I}$ & $\g{I}$ & $\g{I}$ & $\g{I}$ & $\g{I}$ & $\g{I}$ & $\g{I}$ & $\g{I}$ & $\g{Z}$ & $\g{Z}$ & $\g{Z}$ & $\g{Z}$ & $\g{I}$ & $\g{I}$ & $\g{I}$ & $\g{I}$ & $\g{I}$ & $\g{I}$ & $\g{I}$ & $\g{I}$ & $\g{I}$ & $\g{I}$ & $\g{I}$ & $\g{I}$ & $\g{I}$ & $\g{I}$ & $\g{I}$ & $\g{I}$ & $\g{I}$ & $\g{I}$ & $\g{I}$ & $\g{I}$ & $\g{I}$ & $\g{I}$ & $\g{I}$ & $\g{I}$ & $\g{I}$ & $\g{I}$ & $\g{I}$ & $\g{I}$ & $\g{Z}$ & $\g{Z}$ & $\g{Z}$ & $\g{Z}$ \\
        \cline{2-66}
        & $B_{5}=\{1,3,4\}$ & $\g{I}$ & $\g{I}$ & $\g{I}$ & $\g{I}$ & $\g{I}$ & $\g{I}$ & $\g{I}$ & $\g{I}$ & $\g{I}$ & $\g{I}$ & $\g{I}$ & $\g{I}$ & $\g{I}$ & $\g{I}$ & $\g{I}$ & $\g{I}$ &
         $\g{I}$ & $\g{I}$ & $\g{I}$ & $\g{I}$ & $\g{I}$ & $\g{I}$ & $\g{I}$ & $\g{I}$ & $\g{I}$ & $\g{I}$ & $\g{I}$ & $\g{I}$ & $\g{I}$ & $\g{I}$ & $\g{I}$ & $\g{I}$ & 
         $\g{I}$ & $\g{I}$ & $\g{I}$ & $\g{I}$ & $\g{I}$ & $\g{I}$ & $\g{I}$ & $\g{I}$ & $\g{I}$ & $\g{I}$ & $\g{I}$ & $\g{I}$ & $\g{I}$ & $\g{I}$ & $\g{I}$ & $\g{I}$ & 
         $\g{I}$ & $\g{I}$ & $\g{Z}$ & $\g{Z}$ & $\g{I}$ & $\g{I}$ & $\g{Z}$ & $\g{Z}$ & $\g{I}$ & $\g{I}$ & $\g{Z}$ & $\g{Z}$ & $\g{I}$ & $\g{I}$ & $\g{Z}$ & $\g{Z}$ \\
        \cline{2-66}
        & $B_{6}=\{1,3,5\}$ & $\g{I}$ & $\g{I}$ & $\g{I}$ & $\g{I}$ & $\g{I}$ & $\g{I}$ & $\g{I}$ & $\g{I}$ & $\g{I}$ & $\g{I}$ & $\g{I}$ & $\g{I}$ & $\g{I}$ & $\g{I}$ & $\g{I}$ & $\g{I}$ &
         $\g{I}$ & $\g{I}$ & $\g{I}$ & $\g{I}$ & $\g{I}$ & $\g{I}$ & $\g{I}$ & $\g{I}$ & $\g{I}$ & $\g{I}$ & $\g{I}$ & $\g{I}$ & $\g{I}$ & $\g{I}$ & $\g{I}$ & $\g{I}$ &
         $\g{I}$ & $\g{I}$ & $\g{I}$ & $\g{I}$ & $\g{I}$ & $\g{I}$ & $\g{I}$ & $\g{I}$ & $\g{I}$ & $\g{I}$ & $\g{Z}$ & $\g{Z}$ & $\g{I}$ & $\g{I}$ & $\g{Z}$ & $\g{Z}$ &
         $\g{I}$ & $\g{I}$ & $\g{I}$ & $\g{I}$ & $\g{I}$ & $\g{I}$ & $\g{I}$ & $\g{I}$ & $\g{I}$ & $\g{I}$ & $\g{Z}$ & $\g{Z}$ & $\g{I}$ & $\g{I}$ & $\g{Z}$ & $\g{Z}$ \\
        \cline{2-66}
        & $B_{7}=\{1,3,6\}$ & $\g{I}$ & $\g{I}$ & $\g{I}$ & $\g{I}$ & $\g{I}$ & $\g{I}$ & $\g{I}$ & $\g{I}$ & $\g{I}$ & $\g{I}$ & $\g{I}$ & $\g{I}$ & $\g{I}$ & $\g{I}$ & $\g{I}$ & $\g{I}$ &
         $\g{I}$ & $\g{I}$ & $\g{I}$ & $\g{I}$ & $\g{I}$ & $\g{I}$ & $\g{I}$ & $\g{I}$ & $\g{I}$ & $\g{I}$ & $\g{Z}$ & $\g{Z}$ & $\g{I}$ & $\g{I}$ & $\g{Z}$ & $\g{Z}$ &
         $\g{I}$ & $\g{I}$ & $\g{I}$ & $\g{I}$ & $\g{I}$ & $\g{I}$ & $\g{I}$ & $\g{I}$ & $\g{I}$ & $\g{I}$ & $\g{I}$ & $\g{I}$ & $\g{I}$ & $\g{I}$ & $\g{I}$ & $\g{I}$ &
         $\g{I}$ & $\g{I}$ & $\g{I}$ & $\g{I}$ & $\g{I}$ & $\g{I}$ & $\g{I}$ & $\g{I}$ & $\g{I}$ & $\g{I}$ & $\g{Z}$ & $\g{Z}$ & $\g{I}$ & $\g{I}$ & $\g{Z}$ & $\g{Z}$ \\
        \cline{2-66}
        & $B_{8}=\{1,4,5\}$ & $\g{I}$ & $\g{I}$ & $\g{I}$ & $\g{I}$ & $\g{I}$ & $\g{I}$ & $\g{I}$ & $\g{I}$ & $\g{I}$ & $\g{I}$ & $\g{I}$ & $\g{I}$ & $\g{I}$ & $\g{I}$ & $\g{I}$ & $\g{I}$ & $\g{I}$ & $\g{I}$ & $\g{I}$ & $\g{I}$ & $\g{I}$ & $\g{I}$ & $\g{I}$ & $\g{I}$ & $\g{I}$ & $\g{I}$ & $\g{I}$ & $\g{I}$ & $\g{I}$ & $\g{I}$ & $\g{I}$ & $\g{I}$ & 
         $\g{I}$ & $\g{I}$ & $\g{I}$ & $\g{I}$ & $\g{I}$ & $\g{I}$ & $\g{Z}$ & $\g{Z}$ & $\g{I}$ & $\g{I}$ & $\g{I}$ & $\g{I}$ & $\g{I}$ & $\g{I}$ & $\g{Z}$ & $\g{Z}$ & $\g{I}$ & $\g{I}$ & $\g{I}$ & $\g{I}$ & $\g{I}$ & $\g{I}$ & $\g{Z}$ & $\g{Z}$ & $\g{I}$ & $\g{I}$ & $\g{I}$ & $\g{I}$ & $\g{I}$ & $\g{I}$ & $\g{Z}$ & $\g{Z}$ \\
        \cline{2-66}
        & $B_{9}=\{1,4,6\}$ & $\g{I}$ & $\g{I}$ & $\g{I}$ & $\g{I}$ & $\g{I}$ & $\g{I}$ & $\g{I}$ & $\g{I}$ & $\g{I}$ & $\g{I}$ & $\g{I}$ & $\g{I}$ & $\g{I}$ & $\g{I}$ & $\g{I}$ & $\g{I}$ &
         $\g{I}$ & $\g{I}$ & $\g{I}$ & $\g{I}$ & $\g{I}$ & $\g{I}$ & $\g{Z}$ & $\g{Z}$ & $\g{I}$ & $\g{I}$ & $\g{I}$ & $\g{I}$ & $\g{I}$ & $\g{I}$ & $\g{Z}$ & $\g{Z}$ &
         $\g{I}$ & $\g{I}$ & $\g{I}$ & $\g{I}$ & $\g{I}$ & $\g{I}$ & $\g{I}$ & $\g{I}$ & $\g{I}$ & $\g{I}$ & $\g{I}$ & $\g{I}$ & $\g{I}$ & $\g{I}$ & $\g{I}$ & $\g{I}$ &
         $\g{I}$ & $\g{I}$ & $\g{I}$ & $\g{I}$ & $\g{I}$ & $\g{I}$ & $\g{Z}$ & $\g{Z}$ & $\g{I}$ & $\g{I}$ & $\g{I}$ & $\g{I}$ & $\g{I}$ & $\g{I}$ & $\g{Z}$ & $\g{Z}$ \\
        \cline{2-66}
        & $B_{10}=\{1,5,6\}$ & $\g{I}$ & $\g{I}$ & $\g{I}$ & $\g{I}$ & $\g{I}$ & $\g{I}$ & $\g{I}$ & $\g{I}$ & $\g{I}$ & $\g{I}$ & $\g{I}$ & $\g{I}$ & $\g{I}$ & $\g{I}$ & $\g{Z}$ & $\g{Z}$ &
         $\g{I}$ & $\g{I}$ & $\g{I}$ & $\g{I}$ & $\g{I}$ & $\g{I}$ & $\g{I}$ & $\g{I}$ & $\g{I}$ & $\g{I}$ & $\g{I}$ & $\g{I}$ & $\g{I}$ & $\g{I}$ & $\g{Z}$ & $\g{Z}$ &
         $\g{I}$ & $\g{I}$ & $\g{I}$ & $\g{I}$ & $\g{I}$ & $\g{I}$ & $\g{I}$ & $\g{I}$ & $\g{I}$ & $\g{I}$ & $\g{I}$ & $\g{I}$ & $\g{I}$ & $\g{I}$ & $\g{Z}$ & $\g{Z}$ &
         $\g{I}$ & $\g{I}$ & $\g{I}$ & $\g{I}$ & $\g{I}$ & $\g{I}$ & $\g{I}$ & $\g{I}$ & $\g{I}$ & $\g{I}$ & $\g{I}$ & $\g{I}$ & $\g{I}$ & $\g{I}$ & $\g{Z}$ & $\g{Z}$ \\
        \cline{2-66}
        & $B_{11}=\{4,5,6\}$ & $\g{I}$ & $\g{I}$ & $\g{I}$ & $\g{I}$ & $\g{I}$ & $\g{I}$ & $\g{I}$ & $\g{Z}$ & $\g{I}$ & $\g{I}$ & $\g{I}$ & $\g{I}$ & $\g{I}$ & $\g{I}$ & $\g{I}$ & $\g{Z}$ &
         $\g{I}$ & $\g{I}$ & $\g{I}$ & $\g{I}$ & $\g{I}$ & $\g{I}$ & $\g{I}$ & $\g{Z}$ & $\g{I}$ & $\g{I}$ & $\g{I}$ & $\g{I}$ & $\g{I}$ & $\g{I}$ & $\g{I}$ & $\g{Z}$ &
         $\g{I}$ & $\g{I}$ & $\g{I}$ & $\g{I}$ & $\g{I}$ & $\g{I}$ & $\g{I}$ & $\g{Z}$ & $\g{I}$ & $\g{I}$ & $\g{I}$ & $\g{I}$ & $\g{I}$ & $\g{I}$ & $\g{I}$ & $\g{Z}$ &
         $\g{I}$ & $\g{I}$ & $\g{I}$ & $\g{I}$ & $\g{I}$ & $\g{I}$ & $\g{I}$ & $\g{Z}$ & $\g{I}$ & $\g{I}$ & $\g{I}$ & $\g{I}$ & $\g{I}$ & $\g{I}$ & $\g{I}$ & $\g{Z}$ \\
         \cline{2-66}
        & $B_{12}=\{3,5,6\}$ & $\g{I}$ & $\g{I}$ & $\g{I}$ & $\g{I}$ & $\g{I}$ & $\g{I}$ & $\g{I}$ & $\g{I}$ & $\g{I}$ & $\g{I}$ & $\g{I}$ & $\g{Z}$ & $\g{I}$ & $\g{I}$ & $\g{I}$ & $\g{Z}$ &
         $\g{I}$ & $\g{I}$ & $\g{I}$ & $\g{I}$ & $\g{I}$ & $\g{I}$ & $\g{I}$ & $\g{I}$ & $\g{I}$ & $\g{I}$ & $\g{I}$ & $\g{Z}$ & $\g{I}$ & $\g{I}$ & $\g{I}$ & $\g{Z}$ &
         $\g{I}$ & $\g{I}$ & $\g{I}$ & $\g{I}$ & $\g{I}$ & $\g{I}$ & $\g{I}$ & $\g{I}$ & $\g{I}$ & $\g{I}$ & $\g{I}$ & $\g{Z}$ & $\g{I}$ & $\g{I}$ & $\g{I}$ & $\g{Z}$ &
         $\g{I}$ & $\g{I}$ & $\g{I}$ & $\g{I}$ & $\g{I}$ & $\g{I}$ & $\g{I}$ & $\g{I}$ & $\g{I}$ & $\g{I}$ & $\g{I}$ & $\g{Z}$ & $\g{I}$ & $\g{I}$ & $\g{I}$ & $\g{Z}$ \\
        \cline{2-66}
        & $B_{13}=\{3,4,6\}$ & $\g{I}$ & $\g{I}$ & $\g{I}$ & $\g{I}$ & $\g{I}$ & $\g{I}$ & $\g{I}$ & $\g{I}$ & $\g{I}$ & $\g{I}$ & $\g{I}$ & $\g{I}$ & $\g{I}$ & $\g{I}$ & $\g{I}$ & $\g{I}$ &
         $\g{I}$ & $\g{I}$ & $\g{I}$ & $\g{Z}$ & $\g{I}$ & $\g{I}$ & $\g{I}$ & $\g{Z}$ & $\g{I}$ & $\g{I}$ & $\g{I}$ & $\g{Z}$ & $\g{I}$ & $\g{I}$ & $\g{I}$ & $\g{Z}$ &
         $\g{I}$ & $\g{I}$ & $\g{I}$ & $\g{I}$ & $\g{I}$ & $\g{I}$ & $\g{I}$ & $\g{I}$ & $\g{I}$ & $\g{I}$ & $\g{I}$ & $\g{I}$ & $\g{I}$ & $\g{I}$ & $\g{I}$ & $\g{I}$ &
         $\g{I}$ & $\g{I}$ & $\g{I}$ & $\g{Z}$ & $\g{I}$ & $\g{I}$ & $\g{I}$ & $\g{Z}$ & $\g{I}$ & $\g{I}$ & $\g{I}$ & $\g{Z}$ & $\g{I}$ & $\g{I}$ & $\g{I}$ & $\g{Z}$ \\
        \cline{2-66}
        & $B_{14}=\{3,4,5\}$ & $\g{I}$ & $\g{I}$ & $\g{I}$ & $\g{I}$ & $\g{I}$ & $\g{I}$ & $\g{I}$ & $\g{I}$ & $\g{I}$ & $\g{I}$ & $\g{I}$ & $\g{I}$ & $\g{I}$ & $\g{I}$ & $\g{I}$ & $\g{I}$ & $\g{I}$ & $\g{I}$ & $\g{I}$ & $\g{I}$ & $\g{I}$ & $\g{I}$ & $\g{I}$ & $\g{I}$ & $\g{I}$ & $\g{I}$ & $\g{I}$ & $\g{I}$ & $\g{I}$ & $\g{I}$ & $\g{I}$ & $\g{I}$ &
         $\g{I}$ & $\g{I}$ & $\g{I}$ & $\g{Z}$ & $\g{I}$ & $\g{I}$ & $\g{I}$ & $\g{Z}$ & $\g{I}$ & $\g{I}$ & $\g{I}$ & $\g{Z}$ & $\g{I}$ & $\g{I}$ & $\g{I}$ & $\g{Z}$ &
         $\g{I}$ & $\g{I}$ & $\g{I}$ & $\g{Z}$ & $\g{I}$ & $\g{I}$ & $\g{I}$ & $\g{Z}$ & $\g{I}$ & $\g{I}$ & $\g{I}$ & $\g{Z}$ & $\g{I}$ & $\g{I}$ & $\g{I}$ & $\g{Z}$ \\
        \cline{2-66}
        & $B_{15}=\{2,5,6\}$ & $\g{I}$ & $\g{I}$ & $\g{I}$ & $\g{I}$ & $\g{I}$ & $\g{I}$ & $\g{I}$ & $\g{I}$ & $\g{I}$ & $\g{I}$ & $\g{I}$ & $\g{I}$ & $\g{I}$ & $\g{Z}$ & $\g{I}$ & $\g{Z}$ &
         $\g{I}$ & $\g{I}$ & $\g{I}$ & $\g{I}$ & $\g{I}$ & $\g{I}$ & $\g{I}$ & $\g{I}$ & $\g{I}$ & $\g{I}$ & $\g{I}$ & $\g{I}$ & $\g{I}$ & $\g{Z}$ & $\g{I}$ & $\g{Z}$ &
         $\g{I}$ & $\g{I}$ & $\g{I}$ & $\g{I}$ & $\g{I}$ & $\g{I}$ & $\g{I}$ & $\g{I}$ & $\g{I}$ & $\g{I}$ & $\g{I}$ & $\g{I}$ & $\g{I}$ & $\g{Z}$ & $\g{I}$ & $\g{Z}$ &
         $\g{I}$ & $\g{I}$ & $\g{I}$ & $\g{I}$ & $\g{I}$ & $\g{I}$ & $\g{I}$ & $\g{I}$ & $\g{I}$ & $\g{I}$ & $\g{I}$ & $\g{I}$ & $\g{I}$ & $\g{Z}$ & $\g{I}$ & $\g{Z}$ \\
        \cline{2-66}
        & $B_{16}=\{2,4,6\}$ & $\g{I}$ & $\g{I}$ & $\g{I}$ & $\g{I}$ & $\g{I}$ & $\g{I}$ & $\g{I}$ & $\g{I}$ & $\g{I}$ & $\g{I}$ & $\g{I}$ & $\g{I}$ & $\g{I}$ & $\g{I}$ & $\g{I}$ & $\g{I}$ &
         $\g{I}$ & $\g{I}$ & $\g{I}$ & $\g{I}$ & $\g{I}$ & $\g{Z}$ & $\g{I}$ & $\g{Z}$ & $\g{I}$ & $\g{I}$ & $\g{I}$ & $\g{I}$ & $\g{I}$ & $\g{Z}$ & $\g{I}$ & $\g{Z}$ &
         $\g{I}$ & $\g{I}$ & $\g{I}$ & $\g{I}$ & $\g{I}$ & $\g{I}$ & $\g{I}$ & $\g{I}$ & $\g{I}$ & $\g{I}$ & $\g{I}$ & $\g{I}$ & $\g{I}$ & $\g{I}$ & $\g{I}$ & $\g{I}$ &
         $\g{I}$ & $\g{I}$ & $\g{I}$ & $\g{I}$ & $\g{I}$ & $\g{Z}$ & $\g{I}$ & $\g{Z}$ & $\g{I}$ & $\g{I}$ & $\g{I}$ & $\g{I}$ & $\g{I}$ & $\g{Z}$ & $\g{I}$ & $\g{Z}$ \\
        \cline{2-66}
        & $B_{17}=\{2,4,5\}$ & $\g{I}$ & $\g{I}$ & $\g{I}$ & $\g{I}$ & $\g{I}$ & $\g{I}$ & $\g{I}$ & $\g{I}$ & $\g{I}$ & $\g{I}$ & $\g{I}$ & $\g{I}$ & $\g{I}$ & $\g{I}$ & $\g{I}$ & $\g{I}$ & $\g{I}$ & $\g{I}$ & $\g{I}$ & $\g{I}$ & $\g{I}$ & $\g{I}$ & $\g{I}$ & $\g{I}$ & $\g{I}$ & $\g{I}$ & $\g{I}$ & $\g{I}$ & $\g{I}$ & $\g{I}$ & $\g{I}$ & $\g{I}$ &
         $\g{I}$ & $\g{I}$ & $\g{I}$ & $\g{I}$ & $\g{I}$ & $\g{Z}$ & $\g{I}$ & $\g{Z}$ & 
         $\g{I}$ & $\g{I}$ & $\g{I}$ & $\g{I}$ & $\g{I}$ & $\g{Z}$ & $\g{I}$ & $\g{Z}$ &
         $\g{I}$ & $\g{I}$ & $\g{I}$ & $\g{I}$ & $\g{I}$ & $\g{Z}$ & $\g{I}$ & $\g{Z}$ &
         $\g{I}$ & $\g{I}$ & $\g{I}$ & $\g{I}$ & $\g{I}$ & $\g{Z}$ & $\g{I}$ & $\g{Z}$ \\
        \cline{2-66}
        & $B_{18}=\{2,3,6\}$ & $\g{I}$ & $\g{I}$ & $\g{I}$ & $\g{I}$ & $\g{I}$ & $\g{I}$ & $\g{I}$ & $\g{I}$ & $\g{I}$ & $\g{I}$ & $\g{I}$ & $\g{I}$ & $\g{I}$ & $\g{I}$ & $\g{I}$ & $\g{I}$ &
         $\g{I}$ & $\g{I}$ & $\g{I}$ & $\g{I}$ & $\g{I}$ & $\g{I}$ & $\g{I}$ & $\g{I}$ & $\g{I}$ & $\g{Z}$ & $\g{I}$ & $\g{Z}$ & $\g{I}$ & $\g{Z}$ & $\g{I}$ & $\g{Z}$ &
         $\g{I}$ & $\g{I}$ & $\g{I}$ & $\g{I}$ & $\g{I}$ & $\g{I}$ & $\g{I}$ & $\g{I}$ & $\g{I}$ & $\g{I}$ & $\g{I}$ & $\g{I}$ & $\g{I}$ & $\g{I}$ & $\g{I}$ & $\g{I}$ &
         $\g{I}$ & $\g{I}$ & $\g{I}$ & $\g{I}$ & $\g{I}$ & $\g{I}$ & $\g{I}$ & $\g{I}$ & $\g{I}$ & $\g{Z}$ & $\g{I}$ & $\g{Z}$ & $\g{I}$ & $\g{Z}$ & $\g{I}$ & $\g{Z}$ \\
        \cline{2-66}
        & $B_{19}=\{2,3,5\}$ & $\g{I}$ & $\g{I}$ & $\g{I}$ & $\g{I}$ & $\g{I}$ & $\g{I}$ & $\g{I}$ & $\g{I}$ & $\g{I}$ & $\g{I}$ & $\g{I}$ & $\g{I}$ & $\g{I}$ & $\g{I}$ & $\g{I}$ & $\g{I}$ & 
         $\g{I}$ & $\g{I}$ & $\g{I}$ & $\g{I}$ & $\g{I}$ & $\g{I}$ & $\g{I}$ & $\g{I}$ & $\g{I}$ & $\g{I}$ & $\g{I}$ & $\g{I}$ & $\g{I}$ & $\g{I}$ & $\g{I}$ & $\g{I}$ & 
         $\g{I}$ & $\g{I}$ & $\g{I}$ & $\g{I}$ & $\g{I}$ & $\g{I}$ & $\g{I}$ & $\g{I}$ & $\g{I}$ & $\g{Z}$ & $\g{I}$ & $\g{Z}$ & $\g{I}$ & $\g{Z}$ & $\g{I}$ & $\g{Z}$ &
         $\g{I}$ & $\g{I}$ & $\g{I}$ & $\g{I}$ & $\g{I}$ & $\g{I}$ & $\g{I}$ & $\g{I}$ & $\g{I}$ & $\g{Z}$ & $\g{I}$ & $\g{Z}$ & $\g{I}$ & $\g{Z}$ & $\g{I}$ & $\g{Z}$ \\
        \cline{2-66}
        & $B_{20}=\{2,3,4\}$ & $\g{I}$ & $\g{I}$ & $\g{I}$ & $\g{I}$ & $\g{I}$ & $\g{I}$ & $\g{I}$ & $\g{I}$ & $\g{I}$ & $\g{I}$ & $\g{I}$ & $\g{I}$ & $\g{I}$ & $\g{I}$ & $\g{I}$ & $\g{I}$ & 
         $\g{I}$ & $\g{I}$ & $\g{I}$ & $\g{I}$ & $\g{I}$ & $\g{I}$ & $\g{I}$ & $\g{I}$ & $\g{I}$ & $\g{I}$ & $\g{I}$ & $\g{I}$ & $\g{I}$ & $\g{I}$ & $\g{I}$ & $\g{I}$ & 
         $\g{I}$ & $\g{I}$ & $\g{I}$ & $\g{I}$ & $\g{I}$ & $\g{I}$ & $\g{I}$ & $\g{I}$ & $\g{I}$ & $\g{I}$ & $\g{I}$ & $\g{I}$ & $\g{I}$ & $\g{I}$ & $\g{I}$ & $\g{I}$ & $\g{I}$ & $\g{Z}$ & $\g{I}$ & $\g{Z}$ & $\g{I}$ & $\g{Z}$ & $\g{I}$ & $\g{Z}$ & $\g{I}$ & $\g{Z}$ & $\g{I}$ & $\g{Z}$ & $\g{I}$ & $\g{Z}$ & $\g{I}$ & $\g{Z}$ \\
         \hline
    \end{tabular}
    }
    \caption{Logical Pauli operators of the quantum Reed--Muller code $\mathrm{QRM}(6)$ according to \cref{def:log_def,def:canonical_indices}.}
        \label{tab:logical_m=6}
\end{table}

\vspace*{\fill}
\end{landscape}

\newpage
\bibliographystyle{quantum}

\bibliography{myref}

\end{document}